\RequirePackage{etoolbox}
\csdef{input@path}{%
 {sty/}
}%
\csgdef{bibdir}{bib/}

\documentclass[generic, preprint]{imsart}
\pubyear{2019}
\volume{TBA}
\issue{TBA}
\firstpage{1}
\lastpage{26}
\doi{10.1214/19-BA1157}

\usepackage{amsthm}
\usepackage{amsmath}
\usepackage{amssymb}
\usepackage{amsfonts}
\usepackage{natbib}
\usepackage[english]{babel}
\usepackage[colorlinks,citecolor=blue,urlcolor=blue,filecolor=blue,backref=page]{hyperref}
\usepackage{graphicx}
\usepackage{booktabs}
\usepackage{dsfont}
\usepackage{multirow}
\usepackage{floatrow}
\newfloatcommand{capbtabbox}{table}[][\FBwidth]

\startlocaldefs
\numberwithin{equation}{section}

\theoremstyle{plain}
\newtheorem{theorem}{Theorem}[section]
\newtheorem{remark}{Remark}[section]
\newtheorem{proposition}{Proposition}[section]
\newtheorem{definition}{Definition}[section]
\newtheorem{lemma}{Lemma}[section]

\newcommand{\re}{\operatorname{\mathbb{R}}}
\newcommand{\abs}[1]{\left|#1\right|}
\newcommand{\Prob}{\mathbb{P}}
\newcommand{\simdist}{\operatorname{\stackrel{\mathcal{D}}{\sim}}}
\newcommand{\convdist}{\operatorname{\stackrel{\mathcal{D}}{\rightarrow}}}
\newcommand{\LPTN}{\operatorname{LPTN}}

\newcommand{\za}[1]{\stackrel{a}{#1}}
\newcommand{\zb}[1]{\stackrel{b}{#1}}
\newcommand{\zc}[1]{\stackrel{c}{#1}}
\newcommand{\zd}[1]{\stackrel{d}{#1}}

\DeclareMathOperator{\yo}{y}
\newcommand{\ee}{\mathrm{e}}
\newcommand{\ind}{\mathds{1}}
\endlocaldefs

\allowdisplaybreaks

\begin{document}

\begin{frontmatter}
\title{A New Bayesian Approach to Robustness Against Outliers in Linear Regression}
\runtitle{A New Bayesian Robust Linear Regression}

\begin{aug}
\author[a]{\fnms{Philippe} \snm{Gagnon}\ead[label=e1]{philippe.gagnon@stats.ox.ac.uk}},
\author[b]{\fnms{Alain} \snm{Desgagn\'{e}}\ead[label=e3]{desgagne.alain@uqam.ca}},
\and
\author[c]{\fnms{Myl\`{e}ne} \snm{B\'{e}dard}\ead[label=e2]{bedard@dms.umontreal.ca}}

\address[a]{Department of Statistics, University of Oxford, 24-29 St Giles', Oxford, OX1 3LB, United Kingdom, \printead{e1}}
\address[b]{D\'{e}partement de math\'{e}matiques, Universit\'{e} du Qu\'{e}bec \`{a} Montr\'{e}al, C.P. 8888, Succursale Centre-ville, Montr\'{e}al, QC, H3C 3P8, Canada, \printead{e3}}
\address[c]{D\'{e}partement de math\'{e}matiques et de statistique, Universit\'{e} de Montr\'{e}al, C.P. 6128, Succursale Centre-ville, Montr\'{e}al, QC, H3C 3J7, Canada, \printead{e2}}

\runauthor{Gagnon P., Desgagn\'{e} A. and B\'{e}dard M.}
\end{aug}

\begin{abstract}
 Linear regression is ubiquitous in statistical analysis. It is well
understood that conflicting sources of information may contaminate the
inference when the classical normality of errors is assumed. The
contamination caused by the light normal tails follows from an
undesirable effect: the posterior concentrates in an area in between the
different sources with a large enough scaling to incorporate them all.
The theory of conflict resolution in Bayesian statistics
(\cite{o2012bayesian}) recommends to address this problem by
limiting the impact of outliers to obtain conclusions consistent with
the bulk of the data. In this paper, we propose a model with super
heavy-tailed errors to achieve this. We prove that it is wholly robust,
meaning that the impact of outliers gradually vanishes as they move
further and further away form the general trend. The super heavy-tailed
density is similar to the normal outside of the tails, which gives rise
to an efficient estimation procedure. In addition, estimates are easily
computed. This is highlighted via a detailed user guide, where all steps
are explained through a simulated case study. The performance is shown
using simulation. All required code is given.
\end{abstract}



\end{frontmatter}

\section{Introduction}

The distribution most commonly assumed on the error term in the linear
regression model $Y = \mathbf{x}^{T} \boldsymbol{\beta }+ \epsilon $ is
without a doubt a normal, denoted
$\epsilon /\sigma \sim \mathcal{N}(0,1)$. Estimating the regression
coefficient vector $\boldsymbol{\beta }$ is in this case equivalent to
using ordinary least squares (OLS) method, whether Bayesian (setting the
usual noninformative prior on $\boldsymbol{\beta }$) or maximum
likelihood estimates (MLE) are computed. Given the remarkable properties
of OLS (under certain conditions) such as minimum variance among
unbiased estimators, the normal model is often considered as a benchmark
in terms of efficiency in the absence of outliers. However, it is
well-known that resulting inferences is very sensitive to conflicting
sources of information. From a Bayesian perspective, these conflicting
sources may represent the prior or outliers; we focus on the latter in
this paper.

\cite{1968tiao119} were the first to propose a Bayesian solution.
They suggested to let the error term be distributed as a mixture of two
normals: one component for the nonoutliers and the other one, with a
larger variance, for the outliers. This approach has been generalised
by \cite{1984west431} who modelled errors with heavy-tailed
distributions constructed as scale mixtures of normals, which include
the Student distribution. A different robust Bayesian approach was
introduced by \cite{2009Pena2196}. From a frequentist perspective,
several methods have also been proposed, e.g., the M-
(\cite{huber1973robust}), MM- (\cite{Yohai1987MM}), S-
(\cite{rousseeuw1984robust}), least trimmed squares (LTS,
\cite{Rousseeuw1985LTS}), and robust and efficient weighted least-square
(REWLSE, \cite{GerviniYohai2002rewlse}) estimators.

The most popular Bayesian solution is modelling using the Student, a
consequence of the simplicity of the strategy, the rationale behind it
(giving higher probabilities to extreme values), and the required
computations. The latter follows from the scale mixture representation
of the Student that leads to a normal conditional distribution for
$Y$ given $\boldsymbol{\beta }$, $\sigma $ and a latent variable, which
in turn allows a straightforward implementation of the Gibbs sampler
(\cite{geman1984stochastic}). This method took over that of
\cite{1968tiao119} because the latter is such that the conditional
distribution is a mixture of normals and requires to ``complete'' the
data with auxiliary variables to implement the Gibbs sampler. This may
make computations much more arduous. On the frequentist side, the most
popular method to gain in robustness is arguably the MM-estimator.

Protection against outliers always comes at a price: a loss of
efficiency when the observations are normally distributed. The best
robust alternatives manages to offer a large protection at a low
premium. This is especially true for the estimation of $
\boldsymbol{\beta }$. In this regard, a new method can hardly do better;
in fact matching their performance is quite an achievement. However, the
performance of the existing robust approaches with respect to
$\sigma $ is far less optimal.

The main objective of this paper is to propose a solution that yields
gold standard performance, namely a large protection at a low premium,
for the estimation of both $\boldsymbol{\beta }$ and $\sigma $. The
importance of good estimation for $\sigma $, in the absence or presence
of outliers, should not be overlooked. This parameter plays a crucial
role every time an assessment has to be made about uncertainty around
the regression coefficients (credible intervals, hypothesis testing, and
so on). The performance of the proposed approach, combined with its
simplicity, will allow to offer an appealing Bayesian alternative to the
Student model.

The first step towards the objective is indeed to employ a strategy as
simple as that of \cite{1984west431}, that is, to assume a
distribution on the error term that accommodates for the eventual
presence of outliers without being a mixture. Our approach differs in
that the density has a slower tail decay. It is based on the work of
\cite{desgagne2015robustness} about robust modelling of location and
scale parameters. The author proposed to use a super heavy-tailed
distribution belonging to the family of log-regularly varying
distributions (LRVD) --- with tails behaving like $|z|^{-1}(\log |z|)^{-
\theta }$ --- to achieve whole robustness for both parameters. The idea
of using heavier tails than the Student came after the work of
\cite{andrade2011bayesian} who, in the location-scale framework,
achieved only partial robustness for\vadjust{\goodbreak} the scale by modelling with
polynomial tails. As mentioned by \cite{1984west431}, an outlying
observation is accommodated if the posterior distribution converges to
that excluding the outlier as this one tends to infinity, which
corresponds to our definition of whole robustness. In contrast, partial
robustness translates into a significant (but limited as the outliers
approach plus or minus infinity) impact on the estimation of the
parameter.

The second step towards the objective is therefore to generalise the
results of \cite{desgagne2015robustness} to linear regression. In
fact, it is a generalisation of the results of
\cite{DesGag2019}, which are essentially an application of those of
\cite{desgagne2015robustness} in simple linear regression through the
origin for robust estimation of ratios. This second step represents our
key theoretical contribution. We provide two sufficient conditions that
lead to whole robustness. The first one is to assume a super-heavy
tailed distribution on the error. The other specifies the breakdown
point, which tends to the optimal value of 0.5 as the sample size goes
to infinity. The validity of our robust method is thus supported by
theoretical results. While these are similar to those of
\cite{DesGag2019}, a more sophisticated proof technique is required
given that the location parameter of the conditional distribution of
$Y$ is now an inner product of a known vector and $
\boldsymbol{\beta }$ containing $p$ unknown parameters. Throughout the
paper, we focus on continuous explanatory variables to simplify
explanation and notation. The results are nonetheless valid in ANOVA and
ANCOVA (analyses of variance and covariance), and for variable selection
where joint posteriors of models and parameters are considered. The
corresponding sufficient conditions are given as remarks after the
theoretical results. The price to pay to achieve whole robustness for
all parameters is that the use of super heavy-tailed distributions
prevents us from obtaining normal conditional distributions. There is
therefore a computational cost, in the sense that we cannot implement
a Gibbs sampler; it will however be noticed that easy-to-use samplers
can be used, which makes the cost negligible.

The third and final step towards the objective is to carefully select
the super heavy-tailed distribution in the wholly robust model. To
achieve this, we start with the premise that applied statisticians are
satisfied with the normal model in the absence of outliers and we
specifically design a robust solution from that. We set the distribution
of the error as a log-Pareto-tailed normal (LPTN), a super heavy-tailed
distribution introduced by \cite{desgagne2015robustness}. Its
density exactly matches the standard normal on the central part having
a mass of $\rho $. The parameter $\rho $ is thus the single one to be
chosen by the user, and is typically set to a value between 0.80 and
0.98. The resulting model produces robust estimates exhibiting a similar
behaviour to OLS in the absence of outliers, where the trade-off between
high degree of similarity with OLS and high degree of robustness is
controlled through $\rho $. The model has built-in robustness that
resolves conflict in a sensitive way (see Figure~\ref{fig_illus_thm}).
It completely considers the nonoutliers (from $30$ to $32.5$ in
Figure~\ref{fig_illus_thm}), essentially excludes the observations that
are clearly outlying (beyond $38$ in Figure~\ref{fig_illus_thm}), and
between these two clear cases, contains and bounds their impact. The
first two cases correspond to the strategy commonly applied in practice,
where an observation is either kept or discarded. In the last case, the
method reflects that in the gray area there is a level of uncertainty
about the fact that those observations really are outliers or not. Our
main practical contribution is therefore to provide an efficient and
robust model that automatically deals with this type of uncertainty,
which is especially valuable in high-dimensional problems and when
several analyses have to be performed.\vadjust{\goodbreak}

\begin{figure}[t]
\includegraphics{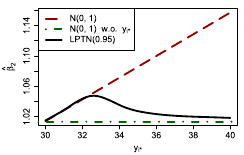}
\caption{Posterior mean of the slope in a simple linear regression as an
observation $y_{i^{*}}\rightarrow \infty $.}\label{fig_illus_thm}
\end{figure}

This rest of the article is organised as follows. The linear regression
model is detailed in Section~\ref{sec-model}, the LRVD family is
presented in Section~\ref{sec-log-regularly} and the theoretical results
are provided in Section~\ref{sec-conflict}. More practically, efficient
and robust regression is investigated in Section~\ref{sec-practice}. The
LPTN distribution is first presented in Section~\ref{sec-LPTN}. A
discussion about efficiency of the robust model with LPTN errors is
provided in Section~\ref{sec_efficiency}. Practical details of our
approach are addressed in Section~\ref{sec_case_study} through a
simulated case study on the modelling of house market values. Numerical
methods such as Markov chain Monte Carlo (MCMC) are discussed for the
computation of different posterior quantities: means, medians, credible
intervals, prediction of future observations and hypothesis testing via
Bayes factors. A powerful tool for outlier identification is also
proposed. In Section~\ref{sec_numerical_study}, a simulation study is
conducted to compare the performance of our approach with different
Bayesian alternatives. Note that even though our approach is Bayesian,
it is possible to use it in a frequentist setting through maximum a
posteriori probability (MAP) estimates, which correspond to MLE when the
prior is set to 1. We thus also include in our study the frequentist
methods mentioned above.

\section{Conflict Resolution in Linear Regression via LRVD}%
\label{sec-robustness}
We henceforth assume that $f$ is a strictly positive continuous
probability density function (PDF) on $\re $ that is symmetric with
respect to the origin, for which all parameters are known and such that
there exists a threshold above which $g(z) = zf(z)$ is monotonic.
Examples of such PDF are the normal, logistic, Laplace, Student (with
prespecified degrees of freedom) and the LPTN (see
Section~\ref{sec-LPTN}).

\subsection{Linear Regression Model}%
\label{sec-model}
\begin{description}%
\item[(i)] Let $Y_{1},\ldots ,Y_{n}\in \re $ be $n$ random variables
representing data points from the dependent variable and $\mathbf{x}
_{1}^{T}:=(1,x_{12},\ldots ,x_{1p}), \ldots , \mathbf{x}_{n}^{T}:=(1,x
_{n2},\ldots ,x_{np})$ be $n$ vectors of observations from the
explanatory variables, where $p\in \{2,3,\ldots \}$, $n\geq p+1$ and
$x_{ij}\in \re $ are assumed to be known. As mentioned in the
introduction, we focus on the situation where all explanatory variables
are continuous. The linear regression model is given by
%
\begin{equation}
\label{eqn_model}
Y_{i}=\mathbf{x}_{i}^{T} \boldsymbol{\beta }+\epsilon _{i},\quad i=1,
\ldots ,n,
\end{equation}
where the $n$ random variables $\epsilon _{1},\ldots ,\epsilon _{n}
\in \re $ and the $p$-dimensional random variable $
\boldsymbol{\beta }:=(\beta _{1},\ldots ,\beta _{p})^{T}\in \re ^{p}$
represent the errors and the vector containing the regression
coefficients, respectively. These $n+1$ random variables are
conditionally independent given $\sigma >0$, a scale parameter, with a
conditional density for $\epsilon _{i}$ given by
\begin{equation*}
\epsilon _{i} \mid \boldsymbol{\beta },\sigma
\stackrel{\mathcal{D}}{=}\epsilon _{i}\mid \sigma \,\simdist \,(1/
\sigma )f\left (\epsilon _{i}/\sigma \right ),\quad i=1,\ldots ,n.
\end{equation*}
\item[(ii)] We assume that the joint prior density of $
\boldsymbol{\beta }$ and $\sigma $, denoted $\pi (\boldsymbol{\beta },
\sigma )$, is bounded by $\max (C,\sigma ^{-1}C)$, where $C>0$ can be any
constant.
\end{description}

A large variety of priors fits within the structure assumed in (ii).
This is the case for non-informative priors such as $\pi (
\boldsymbol{\beta },\sigma )\propto 1/\sigma $ and $\pi (
\boldsymbol{\beta },\sigma )\propto 1$, and practically all proper
densities. Informative priors shall however be used with caution,
especially when they translate into light tailed densities. They may
indeed contaminate the inference if they are in conflict with the
information carried by the data. Establishing the conditions that
guarantee robustness to informative priors in linear regression is not
trivial.

We study robustness of the estimation of $\boldsymbol{\beta }$ and
$\sigma $ in the presence of outliers. In this paper, an observation
$(\mathbf{x}_{i},y_{i})$ is considered as an outlier if its error
$\epsilon _{i}=y_{i}-\mathbf{x}_{i}^{T} \boldsymbol{\beta }$ is
relatively far from 0, where $\boldsymbol{\beta }$ defines the probable
hyperplanes for the bulk of the data. Note that robustness against
outlying errors is a different concept than robustness against outlying
$\mathbf{x}_{i}$ or $y_{i}$. They are generally equivalent though,
except for the unusual case where an observation is outlying in
$\mathbf{x}_{i}$ and $y_{i}$ but still manages to lie in the general
trend, and consequently, be a nonoutlier in error. From a theoretical
perspective, we study the asymptotic behaviour in the sense that we let
outliers' errors $\epsilon _{i}$ approach $+\infty $ or $-\infty $. Our
strategy to mathematically represent this situation is to let their
$y_{i}$ approach $+\infty $ or $-\infty $ while their vector
$\mathbf{x}_{i}$ remains fixed. We thus specify a particular path along
which the outliers move away from the general trend.

We assume that each outlier goes to $-\infty $ or $+\infty $ at its own
specific rate, to the extent that the ratio of two outliers is bounded.
More precisely, we assume that
%
\begin{equation}
\label{eqn-omega}
y_{i}=a_{i}+b_{i} \omega ,
\end{equation}
for $i=1,\ldots ,n$, where $a_{i}, b_{i}\in \re $ are constants such
that $b_{i}=0$ if the point is a nonoutlier and $b_{i}\neq 0$ if it is
an outlier, and then, we let $\omega \rightarrow \infty $. We
mathematically distinguish the outliers from the nonoutliers through the
following. Among the $n$ observations $(y_{1},\ldots ,y_{n})=:
\mathbf{y_{n}}$, we assume that $k$ of them form a group of nonoutlying
observations, that we denote $\mathbf{y_{k}}$, while $\ell =n-k$ of them
are considered as outliers. For $i=1,\ldots ,n$, we define the binary
functions $k_{i}$ and $\ell _{i}$ as follows: if $y_{i}$ is a nonoutlying
value $k_{i}=1$, and if it is an outlier $\ell _{i}=1$. These functions
take the value of 0 otherwise. Therefore, we have $k_{i}+\ell _{i}=1$ for
$i=1,\ldots ,n$, with $\sum _{i=1}^{n} k_{i}=k$, and $\sum _{i=1}^{n}
\ell _{i}=\ell $.

Let the joint posterior density of $\boldsymbol{\beta }$ and
$\sigma $ be denoted by $\pi (\boldsymbol{\beta },\sigma \mid
\mathbf{y_{n}})$ and the marginal density of $(Y_{1},\ldots ,Y_{n})$ be
denoted by $m(\mathbf{y_{n}})$, where
%
\begin{equation}
\label{def_post_dens}
\pi (\boldsymbol{\beta },\sigma \mid \mathbf{y_{n}})=[m(
\mathbf{y_{n}})]^{-1}\pi (\boldsymbol{\beta },\sigma )\prod _{i=1}^{n}
(1/\sigma ) f((y_{i}-\mathbf{x}_{i}^{T}\boldsymbol{\beta })/\sigma ),
\quad \boldsymbol{\beta }\in \re ^{p},\sigma >0.
\end{equation}
Let the joint posterior density of $\boldsymbol{\beta }$ and
$\sigma $ arising from the nonoutlying observations only be denoted by
$\pi (\boldsymbol{\beta },\sigma \mid \mathbf{y_{k}})$ and the
corresponding marginal density be denoted by $m(\mathbf{y_{k}})$, where
\begin{equation*}
\label{eqn-nonoutlier}
\pi (\boldsymbol{\beta },\sigma \mid \mathbf{y_{k}})=[m(
\mathbf{y_{k}})]^{-1}\pi (\boldsymbol{\beta },\sigma )\prod _{i=1}^{n}
\left [(1/\sigma )f((y_{i}-\mathbf{x}_{i}^{T} \boldsymbol{\beta })/
\sigma )\right ]^{k_{i}},\quad \boldsymbol{\beta }\in \re ^{p},\sigma
>0.
\end{equation*}

\begin{proposition}[{Tail behaviour of the posteriors}]
\label{proposition-proper}\
%
\begin{itemize}
\item[(i)] If $n>p+1$, the density $\pi (\boldsymbol{\beta },\sigma
\mid \mathbf{y_{n}})$ is proper.
\item[(ii)] If $k>p+1$ (stronger than $n>p+1$), the density
$\pi (\boldsymbol{\beta },\sigma \mid \mathbf{y_{k}})$ is also proper.
\item[(iii)] If $n>p + 1 + M$, then $\mathbb{E}[\beta _{j}^{M}\mid
\mathbf{y_{n}}]$ for any $j\in \{1,\ldots ,p\}$ and $\mathbb{E}[\sigma
^{M}\mid \mathbf{y_{n}}]$ exist.
\item[(iv)] If $k>p + 1 + M$, then $\mathbb{E}[\beta _{j}^{M}\mid
\mathbf{y_{k}}]$ for any $j\in \{1,\ldots ,p\}$ and $\mathbb{E}[\sigma
^{M}\mid \mathbf{y_{k}}]$ exist.
\end{itemize}
\end{proposition}

\begin{proof}
See Section~\ref{sec-proof}.
\end{proof}

\begin{remark}
\label{rmk_proper}
When any type of explanatory variables is considered (continuous,
discrete as in ANOVA or a mix of both as in ANCOVA), the densities are
proper if we additionally assume that the design matrix (comprised of
$n$ or $k$ observations) has full rank. In variable selection, when the
joint posterior of the models and parameters is considered, this joint
posterior is proper if the assumptions are verified for the ``complete''
model (the model with all variables). The assumptions are more technical
for the moments and are not provided here. We essentially need enough
of ``different'' $\mathbf{x}_{i}$ vectors. In the proof, it is made
clear what is required.
\end{remark}

\subsection{Log-Regularly Varying Distributions}%
\label{sec-log-regularly}
We now provide an overview of the class of log-regularly varying
functions (LRVF), as introduced in \cite{desgagne2013full} and
\cite{desgagne2015robustness}, following the idea of regularly varying
functions developed by \cite{Karamata1930}. They form an
interesting class of functions with useful properties for robustness.

\begin{definition}[LRVF]
\label{def-log-regularly}
We say that a measurable function $g$ is
\textit{log-regularly varying} at $\infty $ with index
$\theta \in \re $, written $g\in L_{\theta }(\infty )$, if
\begin{equation*}
\lim _{z\rightarrow \infty }g(z^{\nu })/g(z)=\nu ^{-\theta },
\end{equation*}
uniformly in any set $\nu \in [1/\eta ,\eta ]$ (for any $\eta \ge 1$).
If $\theta =0$, $g$ is said to be \textit{log-slowly varying} at
$\infty $.
\end{definition}

In \cite{desgagne2015robustness}, it is shown that
Definition~\ref{def-log-regularly} is equivalent to the following: there
exists a constant $A>1$ and a function $s\in L_{0}(\infty )$ such that
for $z\ge A$, $g$ can be written as
\begin{equation*}
g(z)=(\log z)^{-\theta } s(z).
\end{equation*}
Examples of LRVF are $g(z)=(\log z)^{-\theta }$ (with $s(z)=1$) and
$g(z)=(\log z)^{-\theta }\log (\log z)$.

\begin{definition}[LRVD]
\label{def-log-regularly-distribution}
A random variable $Z$ and its distribution are said to be log-regularly
varying with index $\theta \ge 1$ if their density $f$ is such that
$z f(z)\in L_{\theta }(\infty )$.
\end{definition}

Definition~\ref{def-log-regularly-distribution} implies that any density
$f$ with tails behaving like $|z|^{-1}(\log |z|)^{-\theta }$ with
$\theta > 1$ is a LRVD. Some examples like the LPTN distribution are
given in \cite{desgagne2015robustness}. The most important
property of this class of distributions follows from
Definition~\ref{def-log-regularly}: the asymptotic location-scale
invariance of their density, as stated in
Proposition~\ref{prop-location-scale-transformation}.
%
\begin{proposition}[Location-scale invariance]
\label{prop-location-scale-transformation}
If $z f(z)\in L_{\theta }(\infty )$, then we have
\begin{equation*}
(1/\sigma ) f((z-\mu )/\sigma )/f(z)\rightarrow 1 \text{ as } z\rightarrow
\infty ,
\end{equation*}
uniformly on $(\mu ,\sigma )\in [-\vartheta , \vartheta ]\times [1/
\eta ,\eta ]$, for any $\vartheta \ge 0$ and $\eta \ge 1$.
\end{proposition}
\begin{proof}
See \cite{desgagne2015robustness}.
\end{proof}
Proposition~\ref{prop-location-scale-transformation} essentially implies
that the conditional density of an outlier $(1/\sigma )$ $ f((y -
\mathbf{x}^{T}\boldsymbol{\beta })/\sigma )$ asymptotically behaves like
$f(y)$ as $y\rightarrow \infty $. The densities of the outliers at the
numerator of posterior densities cancel each other out with those at the
denominator in the marginal, provided that the integral can be
interchanged with the limit. This is the idea of the proof of our
robustness result presented in the next section. The greatest challenge
is however to prove that we can indeed interchange the limit and the
integral. This part leads to the condition about the maximum number of
outliers to guarantee robustness.

\subsection{Resolution of Conflicts}%
\label{sec-conflict}
We now present Theorem~\ref{thm-main}, the main theoretical contribution
of this paper.
%
\begin{theorem}
\label{thm-main}
If
\begin{description}%
\item[(i)] $z f(z)\in L_{\theta }(\infty )$ with $\theta \ge 1$,
i.e. $f$ is a LRVD,
\item[(ii)]\hspace{6.5mm}$\ell \leq n/2-(p-1/2)$, i.e. \#outliers $\le $ half the
sample $-\,(p-1/2)$,%
\\
$\Leftrightarrow k\geq n/2+(p-1/2)$, i.e. \#nonoutliers $\ge $ half the
sample $+\,(p-1/2)$,%
\\
$\Leftrightarrow k-\ell \geq 2(p-1/2)$,
\hspace{1.0mm}
i.e. \#nonoutliers $-$ \#outliers $\ge $ $2(p-1/2)$,
\end{description}
then, as $\omega \rightarrow \infty $ (where $\omega $ is defined in
(\ref{eqn-omega})), we obtain the following results:
\begin{description}%
\item[(a)]
\begin{equation*}
\frac{m(\mathbf{y_{n}})}{\prod _{i=1}^{n}[f(y_{i})]^{\ell _{i}}}\rightarrow
m(\mathbf{y_{k}}),
\end{equation*}%
\item[(b)]
\begin{equation*}
\pi (\boldsymbol{\beta },\sigma \mid \mathbf{y_{n}})\rightarrow
\pi (\boldsymbol{\beta },\sigma \mid \mathbf{y_{k}}),
\end{equation*}
uniformly on $(\boldsymbol{\beta },\sigma )\in [-\vartheta ,\vartheta
]^{p}\times [1/\eta ,\eta ]$, for any $\vartheta \ge 0$ and
$\eta \ge 1$,
\item[(c)]
\begin{equation*}
\boldsymbol{\beta },\sigma \mid \mathbf{y_{n}} \convdist
\boldsymbol{\beta },\sigma \mid \mathbf{y_{k}},
\end{equation*}
and in particular
\begin{equation*}
\beta _{j}\mid \mathbf{y_{n}} \convdist \beta _{j}\mid \mathbf{y_{k}},
\, j=1,\ldots ,p,
\hspace{5mm}\text{ and }
\hspace{5mm}
\sigma \mid \mathbf{y_{n}} \convdist \sigma \mid \mathbf{y_{k}},
\end{equation*}%
\item[(d)] if additionally $k\geq n/2+(p-1/2) + M$, then
\begin{equation*}
\mathbb{E}[\beta _{j}^{M}\mid \mathbf{y_{n}}]\rightarrow \mathbb{E}[
\beta _{j}^{M}\mid \mathbf{y_{k}}], \, j=1,\ldots ,p,
\hspace{5mm}\text{ and }
\hspace{5mm}
\mathbb{E}[\sigma ^{M}\mid \mathbf{y_{n}}] \rightarrow \mathbb{E}[
\sigma ^{M}\mid \mathbf{y_{k}}].
\end{equation*}
\end{description}
\end{theorem}

\begin{proof}
See Section~\ref{sec-proof}.
\end{proof}

The two sufficient conditions of Theorem~\ref{thm-main} are remarkably
simple. Condition (i) indicates that modelling must be performed using
a super heavy-tailed density $f$, more precisely using a LRVD, e.g. a
LPTN as proposed. Condition (ii) gives in fact the breakdown point,
generally defined as the proportion of outliers $(\ell /n)$ that an
estimator can handle. We have $\ell /n\leq 1/2-(p-1/2)/n$, which
translates into a breakdown point of $50\%$ as $n\rightarrow \infty $
(for fixed $p$), usually considered as the maximum and best desired
value. Condition (ii) is thus generally satisfied in practice.

Results (a) to (d) are different representations of whole robustness.
Essentially, the posterior inference arising from the whole sample
converges towards the posterior inference based on the nonoutliers only.
The impact of outliers then gradually vanishes as they approach plus or
minus infinity.

In Result~(a), the asymptotic behaviour of the marginal $m(
\mathbf{y_{n}})$ is described. This result is used in
Section~\ref{sec_case_study} to assess robustness of Bayes factors for
testing $H_{0}:\beta _{i}=0$ versus $H_{0}:\beta _{i}\neq 0$ (when
$i\geq 2$). Result~(a) is in fact the centrepiece of Theorem 1; its
demonstration requires considerable work, and leads relatively easily
to the other results of the theorem.

The convergence of the posterior density in Result~(b) enables to assess
that the MAP estimates of $\boldsymbol{\beta }$ and $\sigma $ are wholly
robust. Given that these estimators correspond to the MLE when the prior
is proportional to 1, the frequentist estimates are, as a result, also
wholly robust. This allows establishing a connection between Bayesian
and frequentist robustness.

Result~(c) indicates that any estimation of $\boldsymbol{\beta }$ and
$\sigma $ based on posterior quantiles (e.g. using posterior medians
and Bayesian credible intervals) is robust to outliers. Note that in
fact we obtain the stronger result of $L^{1}$ convergence:
\begin{equation*}
\int _{0}^{\infty }\int _{\re ^{p}}\big |\pi (\boldsymbol{\beta },\sigma
\mid \mathbf{y_{n}})-\pi (\boldsymbol{\beta },\sigma \mid
\mathbf{y_{k}})\big |\,d\boldsymbol{\beta }\,d\sigma \rightarrow 0,
\end{equation*}
which in turn implies that $\Prob (\boldsymbol{\beta },\sigma \in E
\mid \mathbf{y_{n}})\rightarrow \Prob (\boldsymbol{\beta },\sigma
\in E \mid \mathbf{y_{k}})$ as $\omega \rightarrow \infty $, uniformly
for all sets $E\subset \re ^{p}\times \re ^{+}$, a slightly stronger
than convergence in distribution given in Result~(c) which requires only
pointwise convergence.

Posterior expectations are wholly robust as well, as indicated by
Result~(d). It is interesting to notice that all these results guarantee
the robustness of a variety of Bayes estimators.

\begin{remark}
\label{rmk_thm}
When any type of explanatory variables is considered, the same results
as in Theorem~\ref{thm-main} hold under the following additional
assumption: it is possible to choose $n/2+(p-1/2)$ (or $n/2+(p-1/2) +
M$) nonoutliers --- the required number of nonoutliers depending on
which results we target (Results~(a) to (c) or Results~(a) to (d)) ---
that have $p$-wise linearly independent $\mathbf{x}_{i}$ vectors. This
means that any $p$ vectors $\mathbf{x}_{i_{1}},\ldots ,\mathbf{x}_{i
_{p}}$ among the chosen subgroup must be linearly independent. In
variable selection, the convergence of the joint posterior of the models
and their parameters, and of the expectations, hold if the assumptions
are verified for the complete model.
\end{remark}

\begin{remark}
We prove that modelling with $f$ having tails behaving like
$|z|^{-1} (\log |z|)^{-\theta }$ is sufficient to obtain the results in
Theorem~\ref{thm-main}. It seems ``almost'' necessary because, on one
hand, a tail behaviour of $z^{-2}$ (corresponding to a Student density)
is not sufficient, and on the other hand, $|z|^{-1}$ is not integrable.
\end{remark}

\section{Efficient and Robust Regression Using LPTN}%
\label{sec-practice}
In Section~\ref{sec-conflict}, we stated theoretical results which
essentially indicate that using a LRVD for the errors ensures a high
breakpoint of $1/2-(p-1/2)/n$ with a whole rejection of the outliers as
their error goes to $+\infty $ or $-\infty $. The conflict is thus
resolved and the linear regression is in agreement with the bulk of the
data.

In this section, we build on these results to propose a solution in the
realistic situation where a statistician satisfied with the normal model
in the absence of outliers seeks protection in the eventuality of
contamination by outliers. Mathematically, we consider the context where
the errors have a mixture distribution, with a normal component for the
bulk of the data and another component $F_{0}$ for the outliers, that
is
%
\begin{equation}
\label{eqn_mixture}
\epsilon _{i}/\sigma \,\sim \, \alpha \,\mathcal{N}(0,1) + (1 - \alpha
) F_{0}, \quad i=1,\ldots ,n,
\end{equation}
where $0<\alpha \le 1$ represents the proportion of normal observations
in the sample. We thus look for efficient estimators that perform well
in the absence of outliers, that is when $\omega =1$ and the model is
the pure normal. As mentioned in the introduction, OLS (or equivalently
the normal model) is considered as the benchmark in this situation. Our
efficient estimators must also be robust and perform in the presence of
outliers, and this, for as many scenarios of $\alpha <1$ and
$F_{0}$ as possible.

\subsection{LPTN Distribution}%
\label{sec-LPTN}
The solution we propose consists in assuming that the errors have a LPTN
distribution with a prespecified parameter $\rho \in (2 \Phi (1)-1, 1)
\approx (0.6827, 1)$, denoted LPTN($\rho $). More precisely, we still
have $\epsilon _{i}\mid \sigma \,\simdist \,(1/\sigma )f\left (\epsilon
_{i}/\sigma \right )$, but the density $f$ is now assumed to be
%
\begin{equation}
\label{eqn_log_pareto}
f(z)=\left \{
\begin{array}{lcc}
\varphi (z) & \text{ if } & \abs{z}\leq \tau ,
\\
\varphi (\tau )\,\frac{\tau }{|z|}\left (\frac{\log \tau }{\log |z|}\right )
^{\lambda +1} & \text{ if } & \abs{z}>\tau ,
\\
\end{array}
\right .
\end{equation}
where $z\in \re $, and $\tau >1$ and $\lambda >0$ are functions of
$\rho $ with
%
\begin{align}
\label{eqn_tau}
& \tau =\Phi ^{-1}((1+\rho )/2) := \{\tau : \Prob (-\tau \leq Z \leq
\tau )= \rho \,\text{ for }\, Z\, \simdist \, \mathcal{N}(0,1)\},
\\
& \lambda =2(1-\rho )^{-1}\varphi (\tau ) \, \tau \log (\tau ),
\nonumber
\end{align}
$\varphi (\cdot )$, $\Phi (\cdot )$ and $\Phi ^{-1}(\cdot )$ being the
PDF, cumulative distribution function (CDF) and inverse CDF of a
standard normal, respectively.

The LPTN distribution was introduced by
\cite{desgagne2015robustness}, who in fact presents a more general
version than that shown here. The parameter $\lambda $ that controls the
tail decay was originally free and a multiplicative normalising constant
$K(\rho , \lambda )$ was needed. For example, the center of the density
(the area $\abs{z}\leq \tau $) was given by $K(\rho , \lambda )\varphi
(z)$. In order to pursue our efficiency objective, we set the constant
to 1, which in return forces $\lambda $ to be automatically set as a
function of $\rho $. The parameter $\rho $, chosen by the user, thus
represents the mass of the central part that exactly matches the
$\mathcal{N}(0, 1)$ density.

As $\rho $ increases, $f$ approaches the normal. An increase in
$\rho $ also implies an increase in $\lambda $ and $\tau $, which
translates into a density $f$ with lighter tails. Efficiency is also
expected to increase, but robustness to decrease. A compromise has
therefore to be made and it is controlled by the statistician through
the parameter $\rho $. In other words, this parameter represents the
tolerance to (bounded) impact from outliers at the benefit of efficiency
when the data set is not contaminated. The user can also select its
value based on prior opinion about the probable proportion of outliers,
by setting it to 1 minus this proportion.

The rationale behind proposing the LPTN is thus that, in addition to
exactly matching the normal density on the part with highest
probability, this distribution has log-Pareto tails ensuring that our
theoretical robustness result hold, and this for any value of
$\rho $. This type of tails consequently accommodates for a large
spectrum of $\alpha $ and $F_{0}$ in the mixture (\ref{eqn_mixture})
when $\alpha <1$ and generates efficient inference when $\alpha =1$ as
well (this latter characteristic is discussed in
Section~\ref{sec_efficiency}). A comparison between different LPTN
densities is shown in Figure~\ref{fig_comp_normal_log_pareto}. Note
that, as required for our theoretical results of
Section~\ref{sec-robustness}, the LPTN distribution has a strictly
positive continuous PDF on $\re $ that is symmetric with respect to the
origin and such that $z f(z)$ is monotonic for $z>\tau $.\looseness=1

\begin{figure}[t]
\includegraphics{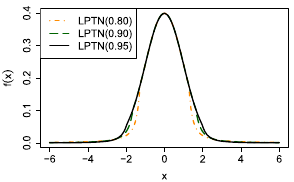}
\caption{Densities of the LPTN(0.80), LPTN(0.90) and LPTN(0.95).}
\label{fig_comp_normal_log_pareto}
\end{figure}

\subsection{Efficiency of the LPTN Model}%
\label{sec_efficiency}
To theoretically study the efficiency of the LPTN Model, we consider the
situation where the data are generated from a normal and evaluate the
performance of the robust estimators in the asymptotic situation
$n\rightarrow \infty $. We start by providing evidences that the
estimators for $\boldsymbol{\beta }$ are consistent, while it depends
on $\rho $ for $\sigma $. We consider that the generative normal model
has $\boldsymbol{\beta }_{0}\in \re ^{p}$ and $\sigma _{0}>0$ as true
parameter values, and denote the associated density of one data point
$g:=\mathcal{N}(\mathbf{x}_{i}^{T}\boldsymbol{\beta }_{0},\sigma _{0}
^{2})$. Denote that associated with the LPTN model $p_{(
\boldsymbol{\beta },\sigma )}(y_{i}):=(1/\sigma )f((y_{i}-\mathbf{x}
_{i}^{T}\boldsymbol{\beta })/\sigma )$, where $f$ is a $\operatorname{LPTN}(\rho )$.
In \cite{bunke1998asymptotic}, it is proved that if the divergence
%
\begin{align}
\label{eqn_KL}
\text{KL}(\boldsymbol{\beta }, \sigma ):=\int \log (g(y_{i}) / p_{(
\boldsymbol{\beta },\sigma )}(y_{i})) \, g(y_{i}) \, dy_{i}
\end{align}
is minimised at a unique $(\boldsymbol{\beta }^{*},\sigma ^{*})$ and some
regularity conditions are satisfied, then
\begin{equation*}
\lim _{n\rightarrow \infty } \mathbb{E}[(\boldsymbol{\beta },\sigma )
\mid \mathbf{y_{n}}] = (\boldsymbol{\beta }^{*},\sigma ^{*})
\quad \text{with probability 1,}
\end{equation*}
where the expectation is with respect to the posterior arising from the
LPTN model. This is proved through the strong consistency of the MAP.

In the supplementary material (Section~\ref{sec_supp}), we prove that the first derivative of
(\ref{eqn_KL}) with respect to $\boldsymbol{\beta }$ equals 0 at
$\boldsymbol{\beta }_{0}$, and this for any value of $\sigma $. While
setting $\boldsymbol{\beta }=\boldsymbol{\beta }_{0}$ in (\ref{eqn_KL}),
we show that it is minimised at $\sigma ^{*}$ which depends on
$\rho $ (see Figure~\ref{fig_convergence_sigma}). We also show that most
of the regularity conditions in \cite{bunke1998asymptotic} are
satisfied. This analysis suggests that the true values for the
regression coefficients are recovered even though the LPTN model is
misspecified. For $\sigma $, the closer $\rho $ is to 1, the more
similar are $\sigma ^{*}$ and $\sigma _{0}$. For instance, when
$\rho =0.9$, $\sigma ^{*}/\sigma _{0}=1.03$, and beyond $\rho =0.95$, this
ratio is essentially 1.

When the data are generated from the normal model, estimators arising
from it are certainly more efficient. We however numerically verified
that the learning rate for the robust estimators is the same as the
normal ones, suggesting that the efficiency is bounded away from 0 for
all $n$. Some additional details are needed to rigorously prove the
consistency of the Bayes estimates and to accurately conclude about
efficiency.\looseness=1

\begin{figure}[t]
\includegraphics{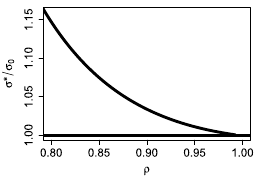}
\caption{Minimiser of the divergence $\sigma ^{*}$ when $ \boldsymbol{\beta
}=\boldsymbol{\beta }_{0}$, as a function of $\rho $.}
\label{fig_convergence_sigma}
\end{figure}

\subsection{Simulated Case Study}%
\label{sec_case_study}
We carry out in this section a linear regression analysis on a given
data set using our robust approach and also the classical method with
the normal assumption for comparison. In doing so, we address all
practical considerations, resulting in a straightforward implementation
by users. In this regard, all R code used to produce numerical results
is provided at
\href{https://arxiv.org/abs/1612.06198}{https://arxiv.org/abs/1612.06198},
which also allows reproducing these results.\looseness=-1

For a given city, we want to model the market value of a house in
thousands of dollars using the average home value in its residential
sector in thousands of dollars, the living area in square metre (sq.m.)
and the land area in sq.m. We consider a simulated sample of size
$n=50$ that contains 3 outliers (it is given in detail in the provided
R code). To give an overview of it, we present in Table~\ref{table1} the
data for Home 2 and for the outliers: Homes 1, 3 and 49.

\begin{table}[ht]
\begin{tabular}{l rrrr}
\hline
\textbf{Characteristics} & \textbf{Home 2} & \textbf{Home 1} & \textbf{Home 3} & \textbf{Home 49} \cr
\hline
Home value (in \$1,000) & 326 & 137 & 20 & 1,000 \cr
Value of the sector (in \$1,000) & 343 & 670 & 350 & 560 \cr
Living area (in sq.m) & 205 & 149 & 222 & 269 \cr
Land area (in sq.m) & 345 & 372 & 434 & 655 \cr
\hline
\end{tabular}
\caption{Data from the studied sample.}\label{table1}
\end{table}

Home 2 has a value of \$326,000 (the sample mean is \$504,900), is
located in a residential sector where houses are valued at \$343,000 in
average (the sample mean is \$508,880), has a living surface of 205
sq.m. (the sample mean is 200 sq.m) and a land of 345 sq.m. (the sample
mean is 500 sq.m). Homes 1 and 3 both have aberrantly low values, while
it is the opposite for Home 49. They are meant to represent a damaged
house, a data entry error and an eco-friendly house, respectively.

To improve the interpretation of the linear regression, the explanatory
variables are centred around their respective sample mean. Therefore,
for each house, we define $x_{i2}$ as the average value in its
residential sector (in \$1,000) minus 508.88, $x_{i3}$ as the living
area minus 200 and $x_{i4}$ as the land area minus 500. Note that
centring affects only the constant of the model, $\beta _{1}$, which can
now be interpreted as the predicted value of the typical house with
average features $x_{i2}=x_{i3}=x_{i4}=0$. The model used to generate
the data (except the outliers) is $Y_{i}=\mathbf{x}_{i}^{T}
\boldsymbol{\beta }+ \epsilon _{i}$ with $\boldsymbol{\beta }:=(508.88,
1, 1, 0.5)^{T}$ and $\epsilon _{i}\mid \sigma \simdist (1/\sigma ) f(
\epsilon _{i}/\sigma )$, where $f=\mathcal{N}(0,1)$ and $\sigma = 40$.

In Figure~\ref{fig_y_vs_x}, we plot the dependent variable against each
explanatory variable to depict their respective linear relation. The
pairwise correlations between the explanatory variables are all below
0.10, suggesting that these graphs provide a fair representation of the
multivariate relation. The parameters of the generative model have been
set to create the expected situation in which an increase in any feature
is associated with an increase in home value.

\begin{figure}[ht]
\includegraphics[scale=0.98]{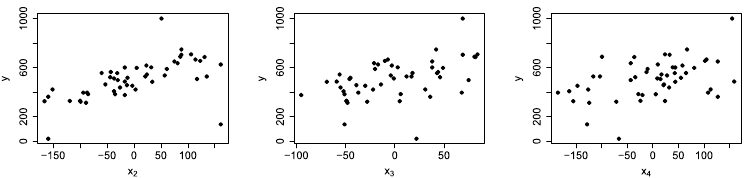}
\caption{The dependent variable versus each of the covariates.}
\label{fig_y_vs_x}\vspace*{-3pt}
\end{figure}

For the analysis, the density $f$ is assumed to be a LPTN($\rho =0.95$)
for the robust model and a $\mathcal{N}(0, 1)$ under the classical
model. We also set $\pi (\boldsymbol{\beta }, \sigma ) \propto 1/
\sigma $, the usual noninformative prior.\vadjust{\goodbreak} The estimation of the
parameters is done through the posterior density as expressed in
(\ref{def_post_dens}). The posterior means, medians and credible
intervals are computed through a random walk Metropolis (RWM) algorithm,
one of the easiest to implement Metropolis--Hastings
(\cite{metropolis1953equation} and
\cite{hastings1970monte}) algorithms. More sophisticated methods like
the Hamiltonian Monte Carlo (HMC, see, e.g.,
\cite{neal2011mcmc}) could be used given that the likelihood function
is differentiable almost everywhere. The MAP and MLE are computed
through optimisation procedures; we use the general-purpose
\texttt{optim} function in R based on Nelder--Mead algorithm. It is of
common knowledge that maximisers (MAP and MLE) may not provide a
posterior summary as good as posterior means, for instance. The
advantage is that they can be computed quickly. We find them
particularly useful for directly giving starting points for the RWM
algorithm and for conducting simulation studies as in
Section~\ref{sec_numerical_study}.

These estimates are presented in Table~\ref{tab_est_sect_4}, in which
the numbers in square brackets are those based on the 47 nonoutliers
only (the sample without Homes 1, 3 and 49). The lower and upper bounds of the credible intervals (CI -- LB and CI -- UB)
are computed from the regions with highest posterior density using
the \texttt{coda} package. Some interesting observations are now made.
First, in the absence of outliers (results in brackets), the results of
the robust LPTN model are very similar to those of the nonrobust normal
model. As mentioned in Section~\ref{sec-LPTN}, the LPTN(0.95) is very
similar to the $\mathcal{N}(0, 1)$, in fact identical except for the 5\%
tails. The normal model is the benchmark in terms of efficiency. All
presented point estimators of $\boldsymbol{\beta }$ under the normal
model indeed correspond to OLS, which are known to produce the best
estimates (in a frequentist sense) when the errors are uncorrelated with
zero mean and homoscedastic with finite variance. This is the case for
the nonoutliers. Our example thus suggests that the choice between the
posterior means, medians, MAP or MLE is not crucial for the robust model
as well. Second, we observe that in the presence of the 3 outliers
(i.e. using the whole sample of size $n=50$), the results of the LPTN
model are barely affected, showing similar results to those excluding
the outliers, while the normal model is clearly contaminated by the
outliers. This is consistent with our theoretical asymptotic results
which indicate agreement with the bulk of the data under the robust
model. In particular, the estimate for $\sigma $ under the LPTN model
is about half that arising from the normal model, resulting in much
shorter credibility intervals for the robust model. Those patterns in
the estimates are typical of the normal and LPTN models. That is
reflected in the thorough performance evaluation presented in the next
section.\looseness=-1

\begin{table}[ht]
\tabcolsep=5.5pt
\begin{tabular}{llrrrrr}
\hline
& & \multicolumn{5}{c}{\textbf{Posterior estimates for}} \cr
\cline{3-7}
& & $\beta _{1}$ & $\beta _{2}$ & $\beta _{3}$ & $\beta _{4}$ & $\sigma $ \cr
\hline
Means & LPTN & $514.0$ [$514.5$] & $1.03$ [$1.03$] & $1.12$ [$1.09$] & $0.39$ [$0.36$] & $47.9$ [$43.8$] \cr
& $\mathcal{N}$ & $504.9$ [$514.3$] & $0.97$ [$1.02$] & $1.40$ [$1.09$] & $0.70$ [$0.36$] & $96.5$ [$43.1$] \cr
\hline
Medians & LPTN & $514.0$ [$514.6$] & $1.03$ [$1.03$] & $1.12$ [$1.09$] & $0.39$ [$0.36$] & $47.4$ [$43.5$] \cr
& $\mathcal{N}$ & $504.9$ [$514.3$] & $0.97$ [$1.02$] & $1.40$ [$1.09$] & $0.70$ [$0.36$] & $95.6$ [$42.7$] \cr
\hline
MAP & LPTN & $513.0$ [$513.7$] & $1.00$ [$1.01$] & $1.11$ [$1.10$] & $0.40$ [$0.37$] & $44.3$ [$40.8$] \cr
& $\mathcal{N}$ & $504.9$ [$514.3$] & $0.97$ [$1.02$] & $1.40$ [$1.09$] & $0.70$ [$0.36$] & $90.1$ [$40.1$] \cr
\hline
MLE & LPTN & $513.1$ [$513.8$] & $1.00$ [$1.01$] & $1.11$ [$1.10$] & $0.40$ [$0.37$] & $44.7$ [$41.1$] \cr
& $\mathcal{N}$ & $504.9$ [$514.3$] & $0.97$ [$1.02$] & $1.40$ [$1.09$] & $0.70$ [$0.36$] & $91.0$ [$40.5$] \cr
\hline
CI -- LB & LPTN & $500.3$ [$501.9$] & $0.86$ [$0.87$] & $0.81$ [$0.81$] & $0.22$ [$0.21$] & $36.9$ [$34.5$] \cr
& $\mathcal{N}$ & $478.1$ [$501.8$] & $0.66$ [$0.87$] & $0.81$ [$0.82$] & $0.38$ [$0.21$] & $77.3$ [$34.4$] \cr
\hline
CI -- UB & LPTN & $527.7$ [$527.0$] & $1.20$ [$1.19$] & $1.42$ [$1.37$] & $0.56$ [$0.52$] & $59.8$ [$53.7$] \cr
& $\mathcal{N}$ & $532.2$ [$526.8$] & $1.29$ [$1.18$] & $2.00$ [$1.37$] & $1.02$ [$0.51$] & $117.1$ [$52.7$] \cr
\hline
\end{tabular}
\caption{Posterior means and medians, MAP, MLE and credible intervals (CI -- LB and CI -- UB), under the LPTN($\rho =0.95$) and
$\mathcal{N}(0,1)$ assumptions for $f$; the numbers in square brackets are the
estimates based on the 47 nonoutliers only.}\label{tab_est_sect_4}
\end{table}

With the posterior in hand, one can take the inference one step further
with outlier identification and prediction. The former is first\vadjust{\goodbreak}
discussed. For each observation $i=1,\ldots ,n$, one can estimate the
value fitted by the hyperplane $\mathbf{x}_{i}^{T}
\boldsymbol{\beta }$, the realisation of the error $y_{i} -
\mathbf{x}_{i}^{T} \boldsymbol{\beta }$ and its standardised version
$z_{i}:=(y_{i} - \mathbf{x}_{i}^{T} \boldsymbol{\beta }) / \sigma $.
This can be achieved through their MAP estimates (or MLE) by simply
plugging in the MAP estimates (or~MLE) of $\boldsymbol{\beta }$ and
$\sigma $ (as given in Table~\ref{tab_est_sect_4}) in their expression.
Or possibly better, they can be estimated by their posterior mean or
median. For this purpose, samples can be directly generated from their
posterior distribution through the values of $\boldsymbol{\beta }$ and
$\sigma $ already generated from the RWM algorithm (or obviously, it can
be done at the same time the algorithm runs). Consider for instance Home
49, which is valued at $y_{49}=1,\!000$, the posterior means give fitted
values of 704.0 (LPTN) and 759.7 (normal), errors of 296.0 (LPTN) and
240.3 (normal) and standardised errors of 6.28 (LPTN) and 2.52 (normal).
We note that the hyperplane is attracted towards the outlier under the
normal model, which leads to an estimated error less extreme than that
under the LPTN model.

Naturally, large estimates for standardised errors $|z_{i}|$ indicate
strong evidence of outlyingness. A threshold of 2.5 is sometimes
recommended to differentiate outliers from nonoutliers, see, e.g.,
\cite{GerviniYohai2002rewlse}. On this basis, Home 49 appears clearly
as an outlier under the LPTN model, while the conclusion is unclear for
the normal model.

To provide a measure of outlyingness, we evaluate the probability for
a (unrealised) standardised error $\epsilon _{i_{0}}/\sigma $ --- which
density is $f$ --- to be more extreme than $|z_{i}|$:
\begin{equation*}
\varrho (z_{i}):=\Prob (|\epsilon _{i_{0}}/\sigma |>|z_{i}|)=\Prob
\left (|\epsilon _{i_{0}}/\sigma |>|y_{i} - \mathbf{x}_{i}^{T}
\boldsymbol{\beta }| / \sigma \right ).
\end{equation*}
Under the normal model, we have
\begin{center}
$\varrho ^{\mathcal{N}}(z_{i}):=2 (1-\Phi (|z_{i}|))$,
\end{center}
whereas under the LPTN($\rho $) it is
\begin{equation*}
\label{eqn_pi_LPTN}
\varrho ^{\LPTN }(z_{i}):=\left \{
\begin{array}{ccc}
2 (\Phi (\tau )-\Phi (|z_{i}|)) + 2\varphi (\tau )\tau (\log \tau )
\lambda ^{-1} & \text{ if } & |z_{i}|\leq \tau ,
\\
2\varphi (\tau )\tau (\log \tau ) \lambda ^{-1}
\left (\frac{\log
\tau }{\log |z_{i}|}\right )^{\lambda } & \text{ if } & |z_{i}|>
\tau ,
\\
\end{array}
\right .
\end{equation*}
where $\tau = 1.96$ and $\lambda = 3.08$ when $\rho =0.95$, as computed
with (\ref{eqn_tau}).

The measure $\varrho (z_{i})$ is a random variable as it is a function
of the unknown parameters $\boldsymbol{\beta }$ and $\sigma $, and can
be estimated \textit{a posteriori} using the same technique as above.
In the same spirit as \cite{GerviniYohai2002rewlse}, one can flag
observations with estimates for $\varrho (z_{i})$ lesser than a chosen
threshold. A reasonable threshold, in our opinion, should lie between
0.01 and 0.02. This corresponds to a range of $2.47$ to $3.11$ of MAP
estimates for $|z_{i}|$ under the LPTN model if $\varrho $ is estimated
through its MAP (because this is achieved by plugging in the MAP of
$|z_{i}|$).

If we look again at results of Home 49, the posterior means for
$\varrho (z_{i})$ give 0.0024 and 0.0208 for the LPTN and normal models,
respectively. Home 49 appears again clearly as an outlier under the LPTN
model, whereas it is much less convincing for the normal model. At a
threshold of 0.02 or less, this observation would not be considered as
an outlier. Outlier detection using the wholly robust LPTN model is
effective; outliers do not mask each other, a well-known phenomenon
arising with nonrobust models typically due to overestimation of the
scale $\sigma $, and sometimes because of attraction of hyperplanes. The
posterior means for the standardised errors $z_{i}$ are plotted in
Figure~\ref{fig_errors}, along with the posterior means for
$\varrho (z_{i})$ for the three outliers.

\begin{figure}[ht]
\includegraphics{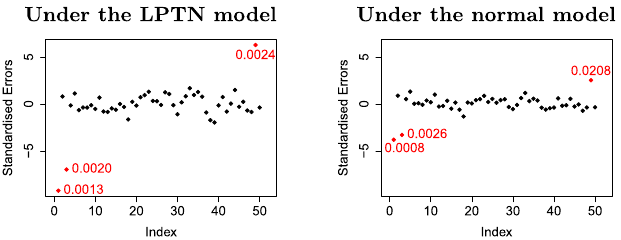}
\caption{Posterior mean for the standardised errors $z_{i}$ and outlier
identification measures $\varrho (z_{i})$, under the LPTN and normal models.}
\label{fig_errors}
\end{figure}

For predicting a future observation, say $Y_{n+1} =\mathbf{x}_{n+1}
^{T} \boldsymbol{\beta }+\epsilon _{n+1}$, we estimate its posterior
predictive density by sampling from it through the RWM algorithm as
before. For each realisation of $(\boldsymbol{\beta }, \sigma )$ in the
Markov chains, we generate $\epsilon _{n+1}$ from an LPTN (or a normal
for the nonrobust model) centred at 0 with a scale parameter
$\sigma $, to which we add $\mathbf{x}_{n+1}^{T} \boldsymbol{\beta }$.
We can thus easily compute posterior predictive quantities such as the
median, credible intervals, probabilities and so on. Note that the
expectation does not exist under the LPTN (because it does not exist for
$\epsilon _{n+1}$). MAP can be approximated from the sample, but because
it requires extra work, we suggest using the median for prediction.

If for example we consider the future observation of the typical house
with $x_{n+1,2}=x_{n+1,3}=x_{n+1,4}=0$, the posterior predictive medians
for $Y_{n+1}$ are 514.0 and 504.9 under the LPTN and normal models,
respectively; they are as expected around the posterior medians of the
intercept $\beta _{1}$. The credible intervals are $(417.4, 611.6)$ and
$(313.7, 698.6)$ for the LPTN and normal models, respectively. We note
the shorter length for the robust model, which is attributable to the
robust estimation of the scale $\sigma $.

Finally, we easily perform statistical hypothesis testing through Bayes
factors. For this, we implement a reversible jump algorithm
(\cite{green1995reversible}) with two models and uniform prior on
these. If, for instance, we want to test for hypotheses $H_{0}:\beta
_{4}= 0$ versus $H_{1}:\beta _{4}\neq 0$, the implementation essentially
requires the tuning of an additional RWM algorithm; that for sampling
the parameters of the model without $x_{4}$. In our example, the Bayes
factors are $1.68 \times 10^{3}$ and $1.74 \times 10^{3}$ for the LPTN
and normal models, respectively. If we exclude the outliers, they become
$2.80 \times 10^{3}$ and $2.12 \times 10^{3}$ for the LPTN and normal
models, respectively.

The Bayes factor is a robust measure under the model with a LPTN
distribution on the error term. Indeed, Result~(a) of
Theorem~\ref{thm-main} states that the marginal $m(\mathbf{y_{n}})$
behaves like $m(\mathbf{y_{k}})\prod _{i=1}^{n}[f(y_{i})]^{\ell _{i}}$.
Furthermore, the marginal $m(\mathbf{y_{n}}\mid H_{0})$ behaves like
$m(\mathbf{y_{k}}\mid H_{0})\prod _{i=1}^{n}[f(y_{i})$ $]^{\ell _{i}}$,
because when the assumptions of Theorem~\ref{thm-main} are satisfied for
the larger model, they are automatically satisfied for the smaller. As
a result, the Bayes factor $m(\mathbf{y_{n}})/m(\mathbf{y_{n}}\mid H
_{0})$ behaves like $m(\mathbf{y_{k}})/m(\mathbf{y_{k}}\mid H_{0})$.

\subsection{Performance Evaluation}%
\label{sec_numerical_study}
In this section, we evaluate the performance of the robust LPTN model
through a simulation study. We consider the same data set and model as
in Section~\ref{sec_case_study}, but get rid of $\mathbf{y_{n}}$ which
are generated. Several values for $\rho $ are considered: $\rho = 0.80,
0.84, 0.90, 0.93, 0.95$, and $0.98$. As in the last Section, it is
compared with the nonrobust normal model. We add the Bayesian approach
of \cite{1968tiao119} with normal mixtures and the model with the
Student distribution. For the latter, we consider different degrees of
freedom (df): 1, 2, 4, 6, and 10. We set $\pi (\boldsymbol{\beta },
\sigma )\propto 1$ and estimate the parameters using the MAP, which
therefore corresponds to the MLE. The Bayesian methods thus become
direct competitors to the frequentist robust estimators like the popular
M- and S-estimators. These as well as MM-, REWLSE (the two best
frequentist methods according to the recent review by
\cite{yuyao2017review}) and LTS estimators are included in the
simulation study.

The data $\mathbf{y_{n}}$ are generated through the errors $\epsilon
_{i}\mid \sigma \, \simdist \, (1/\sigma )f(\epsilon _{i}/\sigma )$ under
the following scenarios:
\begin{description}
\setlength{\itemsep }{1pt}
\item[\textbullet ] \textbf{Scenario 0}: $f=\mathcal{N}(0, 1)$,
\item[\textbullet ] \textbf{Scenario 1}: $f=95\% \, \mathcal{N}(0, 1)
+ 5\% \, \mathcal{N}(7, 1)$,
\item[\textbullet ] \textbf{Scenario 2}: $f=90\% \, \mathcal{N}(0, 1)
+ 10\% \, \mathcal{N}(7, 1)$,
\item[\textbullet ] \textbf{Scenario 3}: $f=95\% \, \mathcal{N}(0, 1)
+ 5\% \, \mathcal{N}(3, 1)$, where the $\mathbf{x}_{i}$ of the outliers
are modified to make them high-leverage points (the procedure is
explained in detail below),
\item[\textbullet ] \textbf{Scenario 4}: $f=90\% \, \mathcal{N}(0, 1)
+ 10\% \, \mathcal{N}(3, 1)$, where the $\mathbf{x}_{i}$ of the outliers
are modified to make them high-leverage points.
\end{description}

Nonoutliers are generated from the first mixture component, whereas
outliers are generated from the second one. The choice of locations for
the outliers aims at producing challenging and interesting situations,
where a vast spectrum of behaviours are observed for especially the LPTN
and Student models with their different sets of parameters $\rho $ and
df. Scenarios 2 and 4 are studied to show how performance varies when
the number of outliers is doubled, from 5\% to 10\% of the sample size.
For each scenario, we consider two sample sizes: $n=50$ and
$n=100$. The case $n=50$ corresponding to the original $\mathbf{x}
_{1},\ldots , \mathbf{x}_{50}$, $50$ additional observations from the
explanatory variables are generated in the same fashion as the original
ones for the case $n=100$.

For Scenarios 3 and 4, when an error is generated from the second
mixture component (that generating extreme values), say $\epsilon _{i
_{0}}$, we modify one of the coordinates of the associated $
\mathbf{x}_{i_{0}}$ to make the observation an high-leverage point. More
precisely, we randomly choose a covariable number, say $j_{0}\in \{2,3,4
\}$, and set $x_{i_{0} j_{0}}=1.5 \, \max _{i} x_{i j_{0}}$.

The performance of each model/estimator is evaluated through the
\textit{premium versus protection} approach of
\cite{anscombe1960outliers}. This approach consists in computing the
premium to pay for using a robust alternative $\mathcal{R}$ to the
normal $\mathcal{N}$ when there are no outliers (Scenario 0), and the
protection provided by this alternative when the data sets are
contaminated (which is likely in the other scenarios). The premium and
protections associated with a robust alternative $\mathcal{R}$ are
evaluated through the following:
\begin{align*}
&\text{Premium}(\mathcal{R}, \hat{\boldsymbol{\beta }}) := \frac{
\mathcal{M}_{\mathcal{R}}(\hat{\boldsymbol{\beta }}) - \mathcal{M}
_{\mathcal{N}}(\hat{\boldsymbol{\beta }}) }{\mathcal{M}_{\mathcal{N}}(
\hat{\boldsymbol{\beta }})},
\cr
&\text{Protection}(\mathcal{R}, \hat{\boldsymbol{\beta }}\mid
\mathcal{S}) := \frac{\mathcal{M}_{\mathcal{N}}(
\hat{\boldsymbol{\beta }}\mid \mathcal{S}) -
\mathcal{M}_{\mathcal{R}}(
\hat{\boldsymbol{\beta }}\mid \mathcal{S})}{\mathcal{M}_{\mathcal{N}}(
\hat{\boldsymbol{\beta }}\mid \mathcal{S})},
\end{align*}
where $\mathcal{S}$ represents the scenario under which the protection
is evaluated (1, 2, 3 or~4), and $\mathcal{M}_{\mathcal{N}}(
\hat{\boldsymbol{\beta }}\mid \mathcal{S})$, for instance, denotes an
error measure $\mathcal{M}$ for estimating $\boldsymbol{\beta }$ by
$\hat{\boldsymbol{\beta }}$ using the normal model $\mathcal{N}$, in
Scenario $\mathcal{S}$. The scenario is not specified for the premium
because it does not vary; it is Scenario 0. The premiums and protections
with respect to $\hat{\sigma }$ --- $\text{Premium}(\mathcal{R},
\hat{\sigma })$ and $\text{Protection}(\mathcal{R}, \hat{\sigma }
\mid \mathcal{S})$ --- have the analogous definitions.

We consider two distinct error measures ($\mathcal{M}_{\mathcal{R}}(
\hat{\boldsymbol{\beta }}\mid \mathcal{S})$ and $\mathcal{M}_{
\mathcal{R}}(\hat{\sigma }\mid \mathcal{S})$) to highlight the
difference between them, and also because there is no natural way of
combining them. We propose to define $\mathcal{M}_{\mathcal{R}}(
\hat{\boldsymbol{\beta }}\mid \mathcal{S})$ as the square root of the
expectation with respect to $\mathbf{Y_{n}}$ (and therefore the
estimates associated with each realisation) of the average squared
vertical distances between the estimated and true hyperplanes measured
at each observation $\mathbf{x}_{i}$:
\begin{align*}
\mathcal{M}_{\mathcal{R}}(\hat{\boldsymbol{\beta }}\mid \mathcal{S})
&:=
\left (\mathbb{E}\left [\frac{1}{n}\sum _{i=1}^{n} (\mathbf{x}_{i}^{T}
\hat{\boldsymbol{\beta }} - \mathbf{x}_{i}^{T}\boldsymbol{\beta })^{2}\right ]\right )
^{1/2} = \left (\frac{1}{n}\, \mathbb{E}\left [(
\hat{\boldsymbol{\beta }} - \boldsymbol{\beta })^{T} \mathbf{X}^{T}
\mathbf{X} (\hat{\boldsymbol{\beta }} - \boldsymbol{\beta })\right ]\right )
^{1/2},
\end{align*}
where $\mathbf{X}$ is the design matrix with rows $\mathbf{x}_{1}^{T},
\ldots ,\mathbf{x}_{n}^{T}$. The expression after the second equality
provides us with another interpretation. The measure represents an
alternative to $ (\mathbb{E} [(\hat{\boldsymbol{\beta }} -
\boldsymbol{\beta })^{T} (\hat{\boldsymbol{\beta }} -
\boldsymbol{\beta }) ] )^{1/2}$, the square root of the
trace of the mean square error (MSE) matrix for $
\hat{\boldsymbol{\beta }}$. Given that under the normal model
$\sigma ^{2} (\mathbf{X}^{T}\mathbf{X})^{-1}$ is the covariance matrix
of $\hat{\boldsymbol{\beta }}$, standardisation is applied to
$\hat{\boldsymbol{\beta }}$ in our measure. For $\hat{\sigma }$, we
simply use the square root of its MSE: $\mathcal{M}_{\mathcal{R}}(
\hat{\sigma }\mid \mathcal{S}):= (\mathbb{E}[(\hat{\sigma }-
\sigma )^{2}] )^{1/2}$. Note that the expectations are
approximated through the simulation of 250,000 vectors $
\mathbf{y_{n}}$.

The premium and protection for a given robust alternative $
\mathcal{R}$ in a given scenario $\mathcal{S}$ are therefore the
relative increase and decrease in $\mathcal{M}_{\mathcal{R}}(\cdot
\mid \mathcal{S})$ due to the use of the robust alternative instead of
the normal (the benchmark model), respectively. For each robust
alternative, there are four premiums to compute: one for the measure for
$\hat{\boldsymbol{\beta }}$ and one for the measure for $
\hat{\sigma }$, in the cases $n=50$ and $n=100$. There are sixteen
protections to compute given that we also do this for Scenarios 1, 2,
3, and 4. The idea is to graphically present the results by plotting the
couples $(\text{Premium}(\mathcal{R}, \hat{\boldsymbol{\beta }}),
\text{Protection}(\mathcal{R}, \hat{\boldsymbol{\beta }} \mid
\mathcal{S}))$ for all robust alternatives. The results for Scenarios
1 and 2 are shown in Figure~\ref{fig_sce_1}, and those for Scenarios 3
and 4 in Figure~\ref{fig_sce_2}.

\begin{figure}[t!]
\includegraphics[scale=0.99]{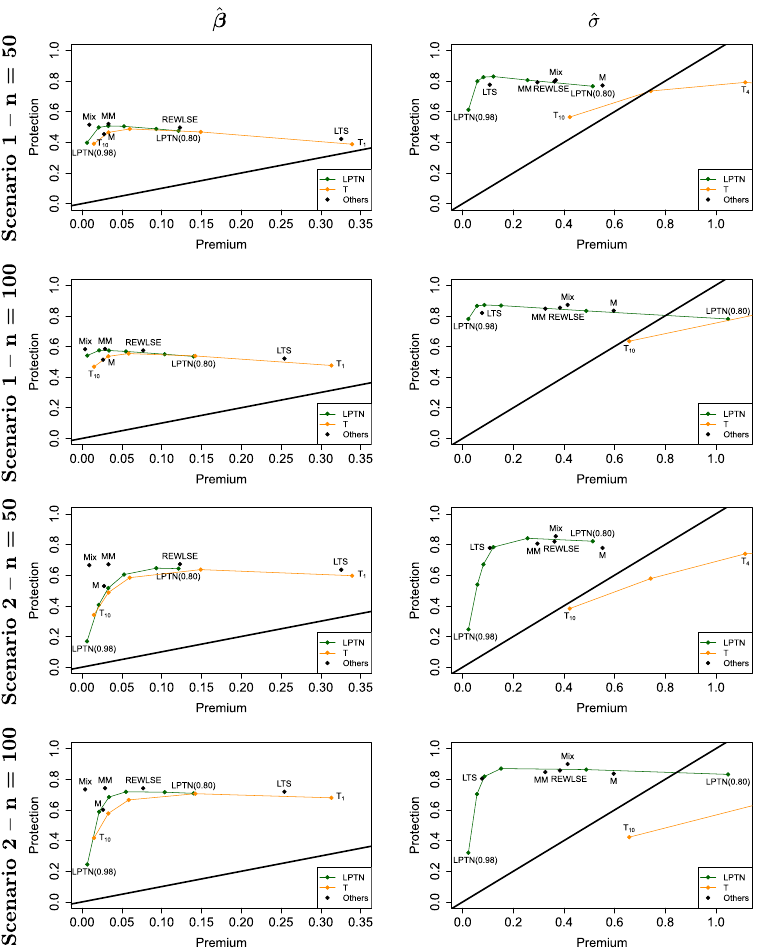}
\caption{Premiums vs protections in Scenarios 1 and 2, and lines premium $=$
protection to identify the robust alternatives that offer better protections
than their premium.}
\label{fig_sce_1}
\end{figure}

\begin{figure}[t!]
\includegraphics[scale=0.99]{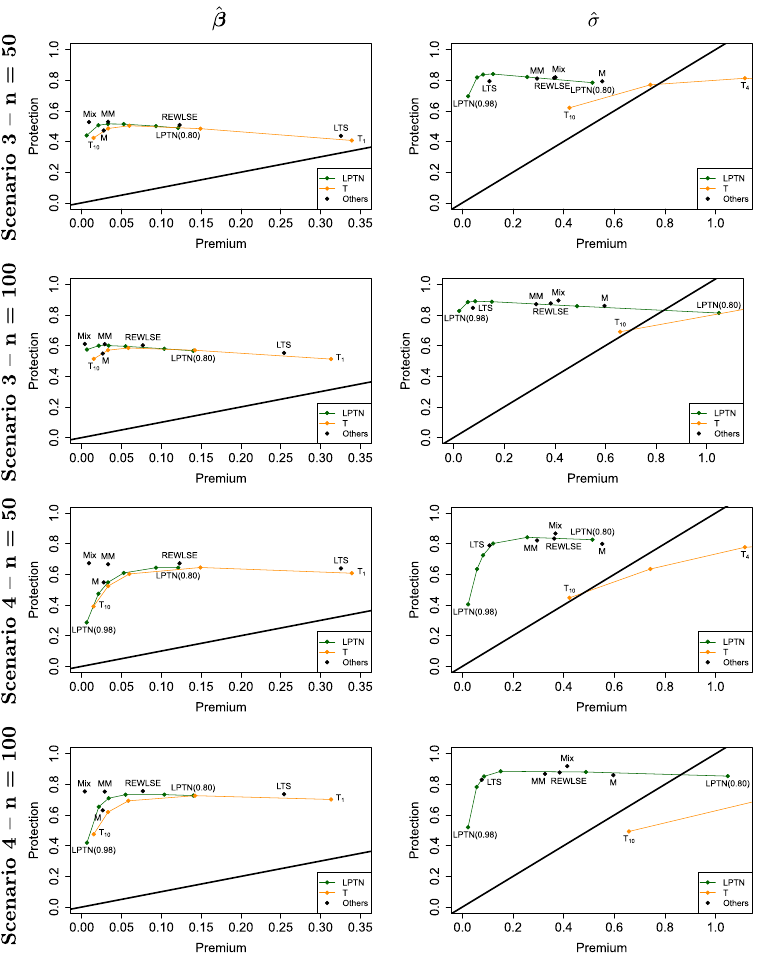}
\caption{Premiums vs protections in Scenarios 3 and 4, and lines premium $=$
protection to identify the robust alternatives that offer better protections
than their premium.}
\label{fig_sce_2}
\end{figure}

From this \textit{premium versus protection} perspective, a robust
alternative dominates another if its premium is smaller and protection
larger. This means that in Figures~\ref{fig_sce_1} and
\ref{fig_sce_2}, we are looking for points in the upper left parts. It
is noticed that the robust alternatives are all excellent candidates,
except maybe for \textit{S}-estimator that we choose not to show because
of its large premium for $\hat{\boldsymbol{\beta }}$ and its same
behaviour as \textit{MM}-estimator for $\hat{\sigma }$. In particular,
the presented robust alternatives all handle high-leverage points.

By looking at Figures~\ref{fig_sce_1} and \ref{fig_sce_2}, we notice
that the LPTN curve (in green) dominates the Student curve (in orange),
more remarkably for $\hat{\sigma }$, but also for $
\hat{\boldsymbol{\beta }}$. We also notice that the optimal values for
$\rho $ for the LPTN are around the nonoutlier percentages, i.e. around
0.95 (the second point starting from the lower left corner) in Scenarios
1 and 3, and around 0.90 (the fourth point starting from the lower left
corner) in Scenarios 2 and 4. This justifies our suggestion in
Section~\ref{sec-LPTN} for selecting $\rho $ based on prior knowledge
about probable proportions of outliers, if users do not have other
preferences. The best LPTN models in all scenarios essentially dominate
all the other alternatives with respect to $\hat{\sigma }$. As for
$\hat{\boldsymbol{\beta }}$, the performance of these LPTN models is
among the best. The mixture model appears better in this case, but often
by little. The difference varies depending on the number of outliers and
the sample size. For instance, look at the LPTN(0.95) in Scenarios 1 and
3 (and also at the scale of the $x$ axis), and notice how the LPTN(0.98)
gets closer to the mixture model in these scenarios when doubling the
sample size, which makes this model almost the best. This allows to make
an interesting remark: for a given percentage of outliers (and therefore
of nonoutliers), a larger sample size translates into enhanced
protection, because there are more nonoutliers. This is especially true
for LPTN models with $\rho $ close to 1.

\section{Conclusion}%
\label{sec_conclusion_linear}
The goal of this paper, which was to provide a solution that reaches
gold standards in terms of \textit{premium versus protection} for all
parameters, is now achieved. The foundations for great protection were
established through our main theoretical contribution: the proof of
whole robustness results for linear regression. The key result is the
convergence of the posterior distribution towards that based on the
nonoutliers only when the outliers approach plus or minus infinity
(Result~(c), Theorem~\ref{thm-main}). The robustness results hold under
two simple and intuitive conditions. Firstly, the error term must follow
a super heavy-tailed distribution, namely a LRVD, to accommodate for the
presence of outliers. Secondly, the number of outliers must not exceed
half the sample $n/2$ minus $p-1/2$ (the number of regression
coefficients minus $1/2$). This last condition translates into a
limiting breakdown point of $0.5$ as $n\rightarrow \infty $.

Although the whole robustness results are theoretical and asymptotic,
their practical relevance has been shown through a comprehensive study
of the LPTN model. This specific choice of super heavy-tailed
distribution represented our main practical contribution as the
resulting model is remarkably efficient and deals with outlying
observations in an automatic and sensitive manner, succeeding in
achieving low premium in addition to large protection. The procedure for
analysing data sets to which it gives rise is also easy to use. These
characteristics of the LPTN model make it a particularly appealing
Bayesian alternative to the partially robust Student model.

\section{Proofs}%
\label{sec-proof}
We in fact provide in this section sketches of the proofs of
Proposition~\ref{proposition-proper} and Theorem~\ref{thm-main} for
space considerations. The detailed proofs can be found in the
supplementary material in Section~\ref{sec_supp}.

\subsection{Proof of Proposition~\ref{proposition-proper}}%
\label{proof-proposition-proper}
Let us pretend for now that the scale parameter is known and that its
value is $\sigma _{0}$. To simplify, we denote the posterior density as
$\pi (\boldsymbol{\beta }\mid \mathbf{y_{n}}):=\pi (
\boldsymbol{\beta },\sigma = \sigma _{0} \mid \mathbf{y_{n}})$. To prove
that it is proper, we show that the marginal $m(\mathbf{y_{n}})$ is
finite. We have that
\begin{align*}
&\int _{\re ^{p}}\pi (\boldsymbol{\beta },\sigma _{0}) \prod _{i=1}^{n}\frac{1}{
\sigma _{0}}f\left (\frac{y_{i}-\mathbf{x}_{i}^{T} \boldsymbol{\beta }}{
\sigma _{0}}\right ) \,d\boldsymbol{\beta }
\cr
&
\qquad
\leq B^{n-p+1}\max \left (1, \frac{1}{\sigma _{0}}\right )\frac{1}{
\sigma _{0}^{n-p}} \int _{\re ^{p}} \prod _{i=1}^{p}\frac{1}{\sigma _{0}}
f\left (\frac{y_{i}-\mathbf{x}_{i}^{T} \boldsymbol{\beta }}{\sigma
_{0}}\right ) \,d\boldsymbol{\beta }
\cr
&
\qquad
\leq B^{n-p+1}\max \left (1, \frac{1}{\sigma _{0}}\right )\frac{1}{
\sigma _{0}^{n-p}} \abs{\text{det}\left (\begin{array}{c}\mathbf{x}_{1}^{T} \cr \vdots \cr \mathbf{x}_{p}^{T}\end{array}\right )}
^{-1} \prod _{i=1}^{p} \int _{\re }f(u_{i}) \, du_{i},
\end{align*}
using $\pi (\boldsymbol{\beta },\sigma _{0})\leq B\max (1,1/\sigma _{0})$
(by assumption) and $f\leq B$ (because of the assumptions on this PDF),
and the changes of variables $u_{i}=(y_{i} - \mathbf{x}_{i}^{T}
\boldsymbol{\beta })/\sigma _{0}, i=1,\ldots ,p$, $B$ being a positive
constant. The last quantity above is finite given that the determinant
is different from 0 because all explanatory variables are continuous.
Note that this justifies also the assumption mentioned in
Remark~\ref{rmk_proper} about the full rank of the design matrix when
any type of explanatory variables is considered.

An additional integral with respect to $\sigma $ is added in front when
$\pi (\boldsymbol{\beta },\sigma \mid \mathbf{y_{n}})$ is considered.
For $\sigma $ not too small (bounded from below), it is easy to see that
the additional integral is finite because $\max (1,1/\sigma )$ is
bounded and $\sigma ^{-(n-p)}$ is integrable if $n-p\geq 2$. This is the
case because $n>p+1$ by assumption. For small $\sigma $, the proof is
more technical and requires to bound more carefully the densities
$f$ than above. See the supplementary material for details.

Proving that $\pi (\boldsymbol{\beta },\sigma \mid \mathbf{y_{k}})$ is
proper is similar. For the moments, we use that
\begin{align*}
\mathbb{E}[\sigma ^{M}\mid \mathbf{y_{n}}]
&=\int \sigma ^{M} \pi (
\boldsymbol{\beta },\sigma \mid \mathbf{y_{n}}) \, d
\boldsymbol{\beta }\, d\sigma
\cr
&\leq [m(\mathbf{y_{n}})]^{-1}B^{M} \int \pi (\boldsymbol{\beta },
\sigma ) \prod _{i=M+1}^{n}\frac{1}{\sigma }f\left (\frac{y_{i}-
\mathbf{x}_{i}^{T} \boldsymbol{\beta }}{\sigma }\right ) d
\boldsymbol{\beta }\, d\sigma ,
\end{align*}
using $f\leq B$. This is finite given that $m(\mathbf{y_{n}})<\infty
$ and the integral is finite because it corresponds to the marginal of
$n-M$ observations, and $n-M>p+1$ by assumption.

For the moments of $\beta _{j}$, it is more technical. Consider the first
moment. We would like to compute instead the first moment of
$|y_{i} - \mathbf{x}_{i}^{T}\boldsymbol{\beta }|$ because $(|y_{i} -
\mathbf{x}_{i}^{T}\boldsymbol{\beta }|/\sigma ) f(|y_{i} - \mathbf{x}
_{i}^{T}\boldsymbol{\beta }|/\sigma )\leq B$ (because of the assumptions
on $f$), and as for the moments of $\sigma $, it would be easy to show
that the integral is finite. The strategy is to write $\beta _{j}$ as
$\mathbf{e}_{j}^{T}\boldsymbol{\beta }$, where $\mathbf{e}_{j}$ is a
vector of size $p$ having 1 at the $j$-th position and 0's elsewhere,
and to write $\mathbf{e}_{j}^{T}$ as a linear combination of $p$ vectors
$\mathbf{x}_{i_{1}},\ldots ,\mathbf{x}_{i_{p}}$ to essentially retrieve
what we were looking for. See the supplementary material for details.

\subsection{Proof of Theorem~\ref{thm-main}}%
\label{proof-thm}
\begin{proof}[Proof of Result (a)]
To prove this result, we use that
\begin{align*}
\frac{m(\mathbf{y_{n}})}{m(\mathbf{y_{k}})\prod _{i=1}^{n}[f(y_{i})]^{
\ell _{i}}}
&= \frac{m(\mathbf{y_{n}})}{m(\mathbf{y_{k}})\prod _{i=1}
^{n}[f(y_{i})]^{\ell _{i}}}
\int _{\re ^{p}}\int _{0}^{\infty }\pi (
\boldsymbol{\beta },\sigma \mid \mathbf{y_{n}})\,d\sigma \,d
\boldsymbol{\beta }
\\
&= \int _{\re ^{p}}\int _{0}^{\infty }\frac{\pi (\boldsymbol{\beta },
\sigma )\prod _{i=1}^{n}
\left [(1/\sigma )f((y_{i}-\mathbf{x}_{i}^{T}
\boldsymbol{\beta })/\sigma )\right ]^{k_{i}+\ell _{i}}}{m(
\mathbf{y_{k}})\prod _{i=1}^{n}[f(y_{i})]^{\ell _{i}}}\,d\sigma \,d
\boldsymbol{\beta }
\\
&= \int _{\re ^{p}}\int _{0}^{\infty }
\pi (\boldsymbol{\beta },\sigma
\mid \mathbf{y_{k}}) \prod _{i=1}^{n}\left [\frac{(1/\sigma )f((y_{i}-
\mathbf{x}_{i}^{T}\boldsymbol{\beta })/\sigma )}{f(y_{i})}\right ]
^{\ell _{i}}\,d\sigma \,d\boldsymbol{\beta },
\end{align*}
and show that this integral converges towards 1 as $\omega \rightarrow
\infty $. Assuming that we can interchange the limit and the integral,
we have that
\begin{align*}
&\lim _{\omega \rightarrow \infty }\int _{\re ^{p}}\int _{0}^{\infty }
\pi (\boldsymbol{\beta },\sigma \mid \mathbf{y_{k}})
\prod _{i=1}^{n}
\left [\frac{(1/\sigma )f((y_{i}-\mathbf{x}_{i}^{T}
\boldsymbol{\beta })/\sigma )}{f(y_{i})}\right ]^{\ell _{i}}\,d\sigma
\,d\boldsymbol{\beta }
\\
&
\qquad
= \int _{\re ^{p}}\int _{0}^{\infty }\lim _{\omega \rightarrow \infty }
\pi (\boldsymbol{\beta },\sigma \mid \mathbf{y_{k}})\prod _{i=1}^{n}
\left [\frac{(1/\sigma )f((y_{i}-\mathbf{x}_{i}^{T}
\boldsymbol{\beta })/\sigma )}{f(y_{i})}\right ]^{\ell _{i}}
\,d\sigma
\,d\boldsymbol{\beta }
\\
&
\qquad
= \int _{\re ^{p}}\int _{0}^{\infty }
\pi (\boldsymbol{\beta },\sigma
\mid \mathbf{y_{k}})
\,d\sigma \,d\boldsymbol{\beta }= 1,
\end{align*}
using Proposition~\ref{prop-location-scale-transformation} in the second
equality, and next Proposition~\ref{proposition-proper}. Note that the
conditions of Proposition~\ref{proposition-proper} are satisfied given
that $k\geq \ell +2p-1\Rightarrow k\geq p+2$, assuming that
$\ell \geq 1$ (otherwise the proof is trivial) and because $p\geq 2$.

To interchange the limit and the integral, we need to prove that the
integrand is bounded by an integrable function of $
\boldsymbol{\beta }$ and $\sigma $ that does not depend on $\omega $.
As in Section~\ref{proof-proposition-proper}, let us set for now the
scale parameter to a positive value $\sigma _{0}$. We know that
%
\begin{align}
\label{eqn_fct_to_bound}
&\pi (\boldsymbol{\beta },\sigma _{0}\mid \mathbf{y_{k}})
\prod _{i=1}
^{n}\left [\frac{(1/\sigma _{0})f((y_{i}-\mathbf{x}_{i}^{T}
\boldsymbol{\beta })/\sigma _{0})}{f(y_{i})}\right ]^{\ell _{i}}
\cr
&  =[m(\mathbf{y_{k}})]^{-1}\pi (\boldsymbol{\beta },\sigma _{0})
\prod _{i=1}^{n}[(1/\sigma _{0})f((y_{i}-\mathbf{x}_{i}^{T}
\boldsymbol{\beta })/\sigma _{0})]^{k_{i}}\left [\frac{(1/\sigma _{0})f((y
_{i}-\mathbf{x}_{i}^{T}\boldsymbol{\beta })/\sigma _{0})}{f(y_{i})}\right ]
^{\ell _{i}}.\cr
\end{align}
Consider that $\boldsymbol{\beta }\in \mathcal{F}$, a set such that the
hyperplanes pass (relatively) close to the nonoutliers (fixed
observations), and therefore, (relatively) far to the outliers. In this
case, for large enough $\omega $, we have that
\begin{equation*}
\prod _{i=1}^{n}\left [\frac{(1/\sigma _{0})f((y_{i}-\mathbf{x}_{i}^{T}
\boldsymbol{\beta })/\sigma _{0})}{f(y_{i})}\right ]^{\ell _{i}}
\end{equation*}
is bounded above using
Proposition~\ref{prop-location-scale-transformation} because
$\mathbf{x}_{i}^{T}\boldsymbol{\beta }$ is bounded (recall that
$y_{i}=a_{i}+b_{i}\omega $), and the remaining terms on the right-hand
side (RHS) in (\ref{eqn_fct_to_bound}) give $\pi (\boldsymbol{\beta },
\sigma _{0}\mid \mathbf{y_{k}})$ which is integrable.

\begin{figure}[b]
\includegraphics{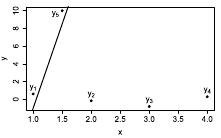}
\caption{Example of a case where the line passes close to a nonoutlier and an
outlier.}
\label{fig_proof}
\end{figure}

Consider now that $\boldsymbol{\beta }\in \mathcal{O}$, a set such that
the hyperplanes pass (relatively) close to the outliers. The difference
is that we are not sure that these hyperplanes do not pass close to the
nonoutliers (see Figure~\ref{fig_proof}). In this example, $n=5$,
$k=4$ and $\ell =1$, which satisfy the assumptions in
Theorem~\ref{thm-main}: $k-\ell =3\geq 2(p-1/2)=3$. We also have that
\begin{equation*}
\frac{(1/\sigma _{0})f((y_{4}-\mathbf{x}_{4}^{T}\boldsymbol{\beta })/
\sigma _{0})}{f(y_{5})}
\end{equation*}
is bounded above using again
Proposition~\ref{prop-location-scale-transformation} but now because
$|y_{4} -\mathbf{x}_{4}^{T}\boldsymbol{\beta }|$ is close to
$\omega $ (this is explained in greater detail in the supplementary
material). Note that it would not be true if $\mathbf{x}_{1} =
\mathbf{x}_{4}$, which is why we require to have enough of different
vectors $\mathbf{x}_{i}$ in Remark~\ref{rmk_thm}. The remaining terms
on the RHS in (\ref{eqn_fct_to_bound}) are
\begin{equation*}
[m(\mathbf{y_{k}})]^{-1}\pi (\boldsymbol{\beta },\sigma _{0})
\prod _{i=1 (i\neq 4)}^{n}[(1/\sigma _{0})f((y_{i}-\mathbf{x}_{i}^{T}
\boldsymbol{\beta })/\sigma _{0})],
\end{equation*}
which after multiplying and dividing by the right marginal is
proportional to the posterior density based on $y_{1},y_{2},y_{3},y
_{5}$, which is integrable given that $4=n-\ell \geq p+2=2p+\ell - 1=
4$. This justifies the assumption on the number of nonoutliers in
Theorem~\ref{thm-main} given by $k=n-\ell \geq 2p+\ell - 1$.

The strategy to do the proof in general is to rewrite the domain of
$\boldsymbol{\beta }$ (which is $\re ^{p}$) as a finite number of
mutually exclusive sets, in which it is always possible to proceed as
above. The function to bound thus becomes a finite sum, where each term
is bounded above by integrable function. When $\sigma $ is free, an
additional level of technicalities is added because $|y_{i}-
\mathbf{x}_{i}^{T}\boldsymbol{\beta }|$ can be large, but not
$|y_{i}-\mathbf{x}_{i}^{T}\boldsymbol{\beta }|/\sigma $. See the
supplementary material for all the details.
\end{proof}

\begin{proof}[Proof of Result (b)]
We have that
\begin{align*}
\abs{\pi (\boldsymbol\beta ,\sigma \mid \mathbf{y_{n}}) - \pi (\boldsymbol\beta ,\sigma \mid \mathbf{y_{k}})}=
\pi (\boldsymbol{\beta },\sigma \mid \mathbf{y_{k}})\abs{\frac{m(\mathbf{y_{k}})}{m(\mathbf{y_{n}})}\prod _{i=1}^{n} [(1/\sigma )f((y_{i}-\mathbf{x}_{i}^{T} \boldsymbol\beta )/\sigma )]^{\ell _{i}} - 1}.
\end{align*}
The absolute value on the RHS converges to 0 as $\omega \rightarrow
\infty $ uniformly on $(\boldsymbol{\beta },\sigma )\in [-\vartheta ,
\vartheta ]^{p}\times [1/\eta ,\eta ]$ using
Proposition~\ref{prop-location-scale-transformation} and Result (a), for
any $\vartheta \ge 0$ and $\eta \ge 1$. On this set, $\pi (
\boldsymbol{\beta },\sigma \mid \mathbf{y_{k}})$ is bounded using the
assumptions on the prior and $f$ and the fact that $m(\mathbf{y_{k}})$
is finite. This concludes the proof.
\end{proof}

\begin{proof}[Proof of Result (c)]
Result (c) is a direct consequence of Result (b) using Scheff\'{e}'s
theorem (see \cite{scheffe1947useful}). See the supplementary
material for details.
\end{proof}

\begin{proof}[Proof of Result (d)]
Result (d) is proved through a mix of the strategies used to show Result
(a) and that the moments exist in Proposition
\ref{proposition-proper}. Assuming that we can interchange the limit and
the integral, we have
\begin{align*}
\lim _{\omega \rightarrow \infty }\mathbb{E}[\sigma ^{M}\mid
\mathbf{y_{n}}]
&= \lim _{\omega \rightarrow \infty } \int _{0}^{\infty
}\int _{\re ^{p}} \sigma ^{M} \pi (\boldsymbol{\beta },\sigma \mid
\mathbf{y_{n}})\,d\boldsymbol{\beta }\,d\sigma
\cr
&= \int _{0}^{\infty }\int _{\re ^{p}} \lim _{\omega \rightarrow \infty
} \sigma ^{M} \pi (\boldsymbol{\beta },\sigma \mid \mathbf{y_{n}})\,d
\boldsymbol{\beta }\,d\sigma
\cr
&= \int _{0}^{\infty }\int _{\re ^{p}} \sigma ^{M} \pi (
\boldsymbol{\beta },\sigma \mid \mathbf{y_{k}})\,d\boldsymbol{\beta }
\,d\sigma =\mathbb{E}[\sigma ^{M}\mid \mathbf{y_{k}}],
\end{align*}
using Result (b). Again, we have to prove the integrand is bounded by
an integrable function of $\boldsymbol{\beta }$ and $\sigma $ that does
not depend on $\omega $. To achieve this, we bound above $\sigma ^{M}
\pi (\boldsymbol{\beta },\sigma \mid \mathbf{y_{n}})$ by a constant
times a function similar to the one that is shown to be bounded by an
integrable function of $\boldsymbol{\beta }$ and $\sigma $ in the proof
of Result (a). See the supplementary material for details. We proceed
with the same strategy for $\mathbb{E}[\beta _{j}^{M}\mid
\mathbf{y_{n}}]$.
\end{proof}

\bibliographystyle{imsart-nameyear}
\bibliography{reference}


\section{Acknowledgements}

The authors acknowledge support from NSERC (Natural Sciences and
Engineering Research Council of Canada), FRQNT (Le Fonds de recherche
du Qu\'{e}bec -- Nature et technologies) and SOA (Society of Actuaries).
They also acknowledge enlightening discussions with Professor Judith
Rousseau about consistency of Bayes estimates. They finally thank an
anonymous referee and an associate editor for their helpful comments.

\section{Supplementary Material}\label{sec_supp}

Proposition~\ref{proposition-proper} and Theorem~\ref{thm-main} are proved in detail in Sections~\ref{proof-proposition-proper-supp} and \ref{sec-proof-a}, respectively. In Section~\ref{sec_complement_3_2}, we complete Section~\ref{sec_efficiency} regarding the claims about the divergence and the regularity conditions in \cite{bunke1998asymptotic}. Finally, we provide a result in Section~\ref{sec_other} that was used to verify that all point estimators of $\boldsymbol\beta$ under the normal model correspond to OLS, as mentioned in Section~\ref{sec_case_study}.

\subsection{Proofs}\label{sec-proof-supp}

Recall the assumptions on $f$: $f$ is a strictly positive continuous PDF on $\re$ that is symmetric with respect to the origin, and such that both tails of $|z| f(z)$ are monotonic, which implies that the tails of $f(z)$ are also monotonic. The monotonicity of the tails of $f(z)$ and $|z| f(z)$ implies that there exists a constant $M> 0$ such that
\begin{equation}\label{eqn-monotonic}
|y|\ge |z|\ge M\Rightarrow f(y)\le f(z) \,\text{ and }\, |y| f(y)\le |z| f(z).
\end{equation}
All these assumptions on $f$ imply that $f(z)$ and  $|z| f(z)$ are bounded on the real line, and both converge to 0 as $|z|\rightarrow \infty$. We can therefore define the constant $B>0$ as follows:
\begin{equation*}
B:=\max\left\{\sup_{z\in\re}f(z),\sup_{z\in\re}|z| f(z),\sup_{\boldsymbol\beta\in\re^p,\,\sigma>0}\pi(\boldsymbol\beta,\sigma)/\max(1,1/\sigma)\right\}.
\end{equation*}

\subsubsection{Proof of Proposition~2.1}\label{proof-proposition-proper-supp}

To prove that  $\pi(\boldsymbol\beta,\sigma\mid\mathbf{y_n})$ is proper (the proof for $\pi(\boldsymbol\beta,\sigma\mid\mathbf{y_k})$ is omitted because it is similar), it suffices to show that the marginal $m(\mathbf{y_n})$ is finite. Recall that we require that $n>p+1$. The reader will notice that only $n\geq p+1$ is required if $\pi(\boldsymbol\beta,\sigma)$ is bounded by $B/\sigma$ for all $\sigma >0$ and $\boldsymbol\beta\in\re^p$ (instead of $\pi(\boldsymbol\beta,\sigma)$ is bounded by $B\max(1,1/\sigma)$).

We first show that the function is integrable on the area where the ratio $1/\sigma$ is bounded. More precisely, we consider $\boldsymbol\beta\in\re^p$ and  $\delta M^{-1}\le \sigma<\infty$, where $\delta$ is a positive constant that can be chosen as small as we want (upper bounds are provided in the proof). We next show that the function is integrable on the area where the ratio $1/\sigma$ approaches infinity, that is $0< \sigma<\delta M^{-1}$. We have
\begin{align*}
&\int_{\delta M^{-1}}^{\infty}\int_{\re^p}\pi(\boldsymbol\beta,\sigma)
 \prod_{i=1}^{n}(1/\sigma)f((y_i-\mathbf{x}_i^T \boldsymbol\beta)/\sigma )\,d\boldsymbol\beta\,d\sigma\\
&\quad\za{\le}B^{n-p+1} \int_{\delta M^{-1}}^{\infty}\max\left(1,\frac{1}{\sigma}\right)\frac{1}{\sigma^{n-p }}\int_{\re^p}\prod_{i=1}^p\frac{1}{\sigma}f\left(\frac{y_i-\mathbf{x}_i^T \boldsymbol\beta}{\sigma}\right)\,d\boldsymbol\beta\,d\sigma\\
&\quad\zb{\leq} \max\left(1,\frac{M}{\delta}\right)B^{n-p+1}\abs{\text{det}\left(\begin{array}{c}\mathbf{x}_1^T \cr \vdots \cr \mathbf{x}_p^T\end{array}\right)}^{-1}\int_{\delta M^{-1}}^{\infty}\frac{1}{\sigma^{n-p}}\,d\sigma \prod_{i=1}^p \int_{-\infty}^{\infty}f(u_i)\,du_i \cr
&\quad\propto\int_{\delta M^{-1}}^{\infty}\frac{1}{\sigma^{n-p}} d\sigma \zc{=} (M/\delta)^{n-p-1}/(n-p-1)<\infty.
\end{align*}
In Step $a$, we use $\pi(\boldsymbol\beta,\sigma)\leq B\max(1,1/\sigma)$ and we bound each of $n-p$ densities $f$ by $B$. In Step $b$, we use the change of variables $u_i=(y_i-\mathbf{x}_i^T\boldsymbol\beta)/\sigma$ for $i=1,\ldots,p$. The determinant is non-null because all explanatory variables are continuous. Indeed, consider the case $p=2$ for instance (i.e.\ the simple linear regression); the determinant is different from 0 provided that $x_{12}\neq x_{22}$, and this happens with probability 1. When any type of explanatory variables is considered, we need to be able to select $p$ observations, say those with $\mathbf{x}_{i_1},\ldots,\mathbf{x}_{i_p}$, such that the matrix with rows $\mathbf{x}_{i_1}^T,\ldots,\mathbf{x}_{i_p}^T$ has a non-null determinant. This is possible when the design matrix has full rank, which is specified in Remark~\ref{rmk_proper}. In Step $c$, we use $n>p+1$. Note that if, instead, we bound $\pi(\boldsymbol\beta,\sigma)$ by $B/\sigma$ in Step $a$, one can verify that the condition $n\ge p+1$ is sufficient to bound above the integral.

We now show that the integral is finite on $\boldsymbol\beta\in\re^p$ and $0<\sigma<\delta M^{-1}$. In this area, the ratio $(1/\sigma)$ approaches infinity. We have to carefully analyse the subareas where $y_i-\mathbf{x}_i^T\boldsymbol\beta$ is close to 0 in order to deal with the $0/0$ form of the ratios $(y_i-\mathbf{x}_i^T \boldsymbol\beta)/\sigma$. In order to achieve this, we split the domain of $\boldsymbol\beta$ as follows:
 \begin{align}\label{eqn_R_p}
  \re^p&=\left[\cap_{i_1=1}^n \mathcal{R}_{i_1}^c\right]\cup \left[\cup_{i_1=1}^n\left(\mathcal{R}_{i_1}\cap\left(\cap_{i_2=1 (i_2\neq i_1)}^n \mathcal{R}_{i_2}^c\right)\right)\right] \cr
   &\quad\cup\left[\cup_{i_1,i_2=1 (i_1\neq i_2)}^n  \left(\mathcal{R}_{i_1}\cap \mathcal{R}_{i_2}\cap\left(\cap_{i_3=1 (i_3\neq i_1,i_2)}^n\mathcal{R}_{i_3}^c \right)\right)\right] \cr
  &\quad \cup \cdots \cup \left[\cup_{i_1,i_2,\ldots,i_p=1 (i_{j}\neq i_s\,\forall i_j,i_s\text{ s.t. }j\neq s)}^n  \left(\mathcal{R}_{i_1}\cap \mathcal{R}_{i_2}\cap \ldots\cap \mathcal{R}_{i_p}\right)\right],
 \end{align}
 where $\mathcal{R}_{i}:=\{\boldsymbol\beta: |y_i-\mathbf{x}_i^T \boldsymbol\beta|<\delta \}, i\in\{1,\ldots,n\}$. The set $\mathcal{R}_i$ represents the hyperplanes $y=\mathbf{x}_i^T\boldsymbol\beta$ characterised by the different values of $\boldsymbol\beta$ that satisfy $|y_i-\mathbf{x}_i^T \boldsymbol\beta|<\delta$. In other words, it represents the hyperplanes passing near the point $(\mathbf{x}_i,y_i)$, and more precisely, at a vertical distance of less than $\delta$. The set $\cap_{i_1=1}^n \mathcal{R}_{i_1}^c$ is therefore comprised of the hyperplanes that are not passing close to any point. The set $\cup_{i_1=1}^n(\mathcal{R}_{i_1}\cap(\cap_{i_2=1 (i_2\neq i_1)}^n \mathcal{R}_{i_2}^c))$ represents the hyperplanes passing near one (and only one) point. The set $\cup_{i_1,i_2=1 (i_1\neq i_2)}^n  (\mathcal{R}_{i_1}\cap \mathcal{R}_{i_2}\cap(\cap_{i_3=1 (i_3\neq i_1,i_2)}^n\mathcal{R}_{i_3}^c ))$ represents the hyperplanes passing near two (and only two) points, and so on.

We choose $\delta$ small enough to ensure that $\mathcal{R}_{i_1}\cap \mathcal{R}_{i_2}\cap \ldots\cap \mathcal{R}_{i_p}\cap \mathcal{R}_{i_{p+1}}=\varnothing$ when $i_1,\ldots,i_{p+1}$ are all different. It is possible to do so because an hyperplane passes through no more than $p$ points. This implies that
\begin{align*}
 &\left[\cup_{i_1,i_2,\ldots,i_p=1 (i_{j}\neq i_s\,\forall i_j,i_s\text{ s.t. }j\neq s)}^n  \left(\mathcal{R}_{i_1}\cap \mathcal{R}_{i_2}\cap \ldots\cap \mathcal{R}_{i_p}\right)\right] \cr
 &\qquad=\left[\cup_{i_1,i_2,\ldots,i_p=1 (i_{j}\neq i_s\,\forall i_j,i_s\text{ s.t. }j\neq s)}^n  \left(\mathcal{R}_{i_1}\cap \mathcal{R}_{i_2}\cap \ldots\cap \mathcal{R}_{i_p}\right.\right. \cr
 &\hspace{60mm}\left.\left.\cap\left(\cap_{i_{p+1}=1 (i_{p+1}\neq i_1,i_2,\ldots,i_p)}^n\mathcal{R}_{i_{p+1}}^c \right)\right)\right].
\end{align*}
Note that all sets $\mathcal{R}_{i_1}\cap \mathcal{R}_{i_2}\cap \ldots\cap \mathcal{R}_{i_p}$ are nonempty when $i_1,\ldots,i_{p}$ are all different, because all explanatory variables are continuous (which implies that the $p\times p$ matrix with rows given by $\mathbf{x}_{i_1}^T,\ldots,\mathbf{x}_{i_p}^T$ has a determinant different from 0). As mentioned above, if they are not all continuous, we have to select them in order to have to have a matrix with a non-null determinant. Note also that $\mathcal{R}_{i_1}\cap(\cap_{i_2=1 (i_2\neq i_1)}^n \mathcal{R}_{i_2}^c)$ is nonempty for all $i_1$, $\mathcal{R}_{i_1}\cap \mathcal{R}_{i_2}\cap(\cap_{i_3=1 (i_3\neq i_1,i_2)}^n\mathcal{R}_{i_3}^c )$ is nonempty for all $i_1,i_2$ such that $i_1\neq i_2$, and so on. Finally note that the decomposition of $\re^p$ given in (\ref{eqn_R_p}) is comprised of $\sum_{i=0}^p {{n} \choose {i}}$ mutually exclusive sets given by $\cap_{i_1=1}^n \mathcal{R}_{i_1}^c$, $\mathcal{R}_{i_1}\cap(\cap_{i_2=1 (i_2\neq i_1)}^n \mathcal{R}_{i_2}^c),i_1=1,\ldots,n$, $\mathcal{R}_{i_1}\cap \mathcal{R}_{i_2}\cap(\cap_{i_3=1 (i_3\neq i_1,i_2)}^n\mathcal{R}_{i_3}^c ),i_1,i_2=1,\ldots,n$ with $i_1\neq i_2$, and so on.

We now consider $0<\sigma<\delta M^{-1}$ and $\boldsymbol\beta$ in one of the $\sum_{i=0}^p {{n} \choose {i}}$ mutually exclusive sets given in (\ref{eqn_R_p}). As explained above, the difficulty lies in dealing with the hyperplanes $\boldsymbol\beta$ that are such that $|y_i-\mathbf{x}_i^T \boldsymbol\beta|<\delta$ for some points $(\mathbf{x}_i,y_i)$. The strategy is essentially to use the product of $(1/\sigma)f((y_i-\mathbf{x}_i^T \boldsymbol\beta)/\sigma )$ of these points to integrate over $\boldsymbol\beta$, and to bound the other terms of $m(\mathbf{y_n})$. Therefore, if $\boldsymbol\beta\in\mathcal{R}_{i_1}\cap \mathcal{R}_{i_2}\cap \ldots\cap \mathcal{R}_{i_p}$, we consider the points $(\mathbf{x}_{i_1},y_{i_1}),(\mathbf{x}_{i_2},y_{i_2}), \ldots, (\mathbf{x}_{i_p},y_{i_p})$ to integrate over $\boldsymbol\beta$. If $\boldsymbol\beta\in\mathcal{R}_{i_1}\cap \mathcal{R}_{i_2}\cap \ldots\mathcal{R}_{i_{p-1}}\cap (\cap_{i_p=1 (i_p\neq i_1,\ldots,i_{p-1})}^n\mathcal{R}_{i_p}^c)$, we consider the points $(\mathbf{x}_{i_1},y_{i_1}),(\mathbf{x}_{i_2},y_{i_2}), \ldots, (\mathbf{x}_{i_{p-1}},y_{i_{p-1}})$, and any other point $(\mathbf{x}_{i_{p}},y_{i_{p}})$ (leading to a matrix with a non-null determinant) to integrate over $\boldsymbol\beta$, and so on. We have
 \begin{align*}
  \pi(\boldsymbol\beta,\sigma) \prod_{i=1}^{n}(1/\sigma)f((y_i-\mathbf{x}_i^T \boldsymbol\beta)/\sigma )&\za{\leq}(B/\sigma)\max(\delta M^{-1},1)\prod_{i=1}^{n}(1/\sigma)f((y_i-\mathbf{x}_i^T \boldsymbol\beta)/\sigma ) \cr
  &\hspace{-35mm}\propto (1/\sigma)\prod_{i\in\{i_1,\ldots,i_p\}}(1/\sigma)f((y_{i}-\mathbf{x}_i^T \boldsymbol\beta)/\sigma) \prod_{i\notin\{i_1,\ldots,i_p\}}(1/\sigma)f((y_i-\mathbf{x}_i^T\boldsymbol\beta)/\sigma) \cr
  &\hspace{-35mm}\zb{\leq}(1/\sigma)[(1/\sigma)f\left(\delta/\sigma \right)]^{n-p}\prod_{i\in\{i_1,\ldots,i_p\}}(1/\sigma)f((y_{i}-\mathbf{x}_i^T \boldsymbol\beta)/\sigma)  \cr
  &\hspace{-35mm}\zc{\leq} [B/\delta]^{n-p-1}(1/\sigma^2)f\left(\delta/\sigma \right)\prod_{i\in\{i_1,\ldots,i_p\}}(1/\sigma)f((y_{i}-\mathbf{x}_i^T \boldsymbol\beta)/\sigma) .
 \end{align*}
In Step $a$, we use $\pi(\boldsymbol\beta,\sigma)\le B\max(1,1/\sigma)= (B/\sigma) \max(\sigma,1)\le (B/\sigma) \max(\delta M^{-1},$ $1)$. In Step $b$, for all $i\notin\{i_1,\ldots,i_p\}$ we use $f((y_i-\mathbf{x}_i^T \boldsymbol\beta)/\sigma)\le f(\delta/\sigma)$ by the monotonicity of the tails of $f$ because $|y_i-\mathbf{x}_i^T\boldsymbol\beta|/\sigma\ge \delta/\sigma\ge \delta\delta^{-1} M=M$. In Step $c$, we bound $n-p-1$ terms $(1/\sigma)f(\delta/\sigma)$ by $B/\delta$.

Finally, we bound the integral of $(1/\sigma^2)f\left(\delta/\sigma \right)\prod_{i\in\{i_1,\ldots,i_p\}}(1/\sigma)f((y_{i}-\mathbf{x}_i^T \boldsymbol\beta)/\sigma)$ by
\begin{align*}
&\int_{0}^{\infty}(1/\sigma^2)f\left(\delta/\sigma \right)\int_{\re^p}\prod_{i\in\{i_1,\ldots,i_p\}}(1/\sigma)f((y_{i}-\mathbf{x}_i^T \boldsymbol\beta)/\sigma)\,d\boldsymbol\beta\,d\sigma \cr
&\za{=}\abs{\text{det}\left(\begin{array}{c}\mathbf{x}_{i_1}^T \cr \vdots \cr \mathbf{x}_{i_p}^T\end{array}\right)}^{-1}\int_{0}^{\infty}(1/\sigma^2)f\left(\delta/\sigma \right)\,d\sigma\zb{=}\abs{\text{det}\left(\begin{array}{c}\mathbf{x}_{i_1}^T \cr \vdots \cr \mathbf{x}_{i_p}^T\end{array}\right)}^{-1}\int_{0}^{\infty}f\left(\sigma' \right)\,d\sigma'<\infty.
\end{align*}
In Step $a$, we use the same change of variables as above: $u_j=(y_{i_j}-\mathbf{x}_{i_j}^T\boldsymbol\beta)/\sigma$ for $j=1,\ldots,p$. In Step $b$, we use the change of variable  $\sigma'=\delta/\sigma$.

We now prove that the $M$-moments exist if $n>p+1+M$ (considering the whole data set, the proof for the posterior expectations based on the nonoutliers only is omitted because it is similar). It is easy to prove that $\mathbb{E}[\sigma^M\mid \mathbf{y_n}]<\infty$, from what has been demonstrated above. Indeed,
\begin{align*}
 \mathbb{E}[\sigma^M\mid \mathbf{y_n}] &= [m(\mathbf{y_n})]^{-1}\int_{0}^{\infty}\int_{\re^p} \sigma^M \pi(\boldsymbol\beta,\sigma)
 \prod_{i=1}^{n}(1/\sigma)f((y_i-\mathbf{x}_i^T \boldsymbol\beta)/\sigma )\,d\boldsymbol\beta\,d\sigma \cr
 &\leq B^M [m(\mathbf{y_n})]^{-1} \int_{0}^{\infty}\int_{\re^p}  \pi(\boldsymbol\beta,\sigma)
 \prod_{i=M+1}^{n}(1/\sigma)f((y_i-\mathbf{x}_i^T \boldsymbol\beta)/\sigma )\,d\boldsymbol\beta\,d\sigma,
\end{align*}
where $M$ densities $f$ have been bounded by $B$. We know that $m(\mathbf{y_n})$ is finite because $n>p+1$. We also know that the last integral above is finite because it corresponds to the marginal of $n-M$ observations, which is finite given that $n-M>p+1$.

For the expectations $\mathbb{E}[|\beta_j|^M \mid \mathbf{y_n}]$, we detail the proof for the cases $M=1$ and $M=2$. From these, it will be clear that the result holds in general, with further technicalities. For the proof, we use that $\beta_j$ can be rewritten as $\mathbf{e}_{j}^T \boldsymbol\beta$, where $\mathbf{e}_{j}$ is a vector of size $p$ having 1 at the $j$-th position and 0's elsewhere. This vector can be expressed as a linear combination of $p$ vectors $\mathbf{x}_{i_1}, \ldots, \mathbf{x}_{i_p}$, $i_1,\ldots,i_p\in\{1,\ldots,n\}$, because these vectors are linearly independent and form a basis of $\re^p$ (given that all explanatory variables are continuous). Using the triangle inequality, we have
\begin{align*}
 \mathbb{E}[|\beta_{j}|\mid \mathbf{y_n}]=\mathbb{E}[|\mathbf{e}_{j}^T \boldsymbol\beta|\mid \mathbf{y_n}]&=\mathbb{E}\left[\left|\sum_{s=1}^p a_s \mathbf{x}_{i_s}^T \boldsymbol\beta\right|\mid \mathbf{y_n}\right]\leq \sum_{s=1}^p |a_s| \mathbb{E}\left[\left| \mathbf{x}_{i_s}^T \boldsymbol\beta\right|\mid \mathbf{y_n}\right] \cr
 &\hspace{10mm}\leq \sum_{s=1}^p |a_s|\left(|y_{i_s}|+ \mathbb{E}\left[\left|y_{i_s}- \mathbf{x}_{i_s}^T \boldsymbol\beta\right|\mid \mathbf{y_n}\right]\right),
\end{align*}
where $a_1,\ldots,a_p\in\re$. We can prove that $\mathbb{E}[|y_{i_s}- \mathbf{x}_{i_s}^T \boldsymbol\beta|\mid \mathbf{y_n}]<\infty$ in the same way that we proved that $\mathbb{E}[\sigma^M\mid \mathbf{y_n}]<\infty$, using instead that
\[
 (|y_{i_s}- \mathbf{x}_{i_s}^T \boldsymbol\beta|/\sigma)  f((y_{i_s}- \mathbf{x}_{i_s}^T \boldsymbol\beta)/\sigma ) \leq B.
\]
Therefore, $\mathbb{E}[|\beta_{j}|\mid \mathbf{y_n}]<\infty$ if $n>p+2$. For the second moment, using again the triangle inequality we have
\begin{align*}
 \mathbb{E}\left[\left|\beta_{j}^2\right|\mid \mathbf{y_n}\right]&=\mathbb{E}\left[\left|\left(\sum_{s=1}^p a_s \mathbf{x}_{i_s}^T\boldsymbol\beta\right)^2 \right|\mid \mathbf{y_n}\right] \cr
 &= \mathbb{E}\left[\left|\sum_{s=1}^p (a_s \mathbf{x}_{i_s}^T\boldsymbol\beta)^2 + \sum_{s,t (s\neq t)} a_s \mathbf{x}_{i_s}^T\boldsymbol\beta \, a_t \mathbf{x}_{i_t}^T\boldsymbol\beta   \right|\mid \mathbf{y_n}\right] \cr
 &\leq \mathbb{E}\left[\sum_{s=1}^p (a_s \mathbf{x}_{i_s}^T\boldsymbol\beta)^2 \mid \mathbf{y_n}\right]  + \mathbb{E}\left[\left| \sum_{s,t (s\neq t)} a_s \mathbf{x}_{i_s}^T\boldsymbol\beta \, a_t \mathbf{x}_{i_t}^T\boldsymbol\beta   \right|\mid \mathbf{y_n}\right].
\end{align*}
We analyse the two last terms separately, starting with the second one. Using again the triangle inequality we have,
\begin{align*}
 &\mathbb{E}\left[\left| \sum_{s,t (s\neq t)} a_s \mathbf{x}_{i_s}^T\boldsymbol\beta \, a_t \mathbf{x}_{i_t}^T\boldsymbol\beta   \right|\mid \mathbf{y_n}\right]\leq \sum_{s,t (s\neq t)} \left|a_s a_t\right|\mathbb{E}\left[\left|\mathbf{x}_{i_s}^T\boldsymbol\beta\right| \left| \mathbf{x}_{i_t}^T \boldsymbol\beta\right| \mid \mathbf{y_n}\right] \cr
 &\leq \sum_{s,t (s\neq t)} \left|a_s a_t\right|\left(\mathbb{E}\left[\left|y_{i_s} - \mathbf{x}_{i_s}^T\boldsymbol\beta\right| \left|y_{i_t} - \mathbf{x}_{i_t}^T \boldsymbol\beta\right| \mid \mathbf{y_n}\right]+ |y_{i_s}|\mathbb{E}\left[\left|y_{i_t} - \mathbf{x}_{i_t}^T \boldsymbol\beta\right| \mid \mathbf{y_n}\right]\right. \cr
&\hspace{30mm} \left.+ |y_{i_t}|\mathbb{E}\left[\left|y_{i_s} - \mathbf{x}_{i_s}^T\boldsymbol\beta\right| \mid \mathbf{y_n}\right]+|y_{i_s}y_{i_t}|\right).
\end{align*}
All the terms in the sum are finite if $n>p+3$. Also,
\begin{align*}
 \mathbb{E}\left[\sum_{s=1}^p (a_s \mathbf{x}_{i_s}^T\boldsymbol\beta)^2 \mid \mathbf{y_n}\right]&=\sum_{s=1}^p a_s^2 \mathbb{E}\left[\left|\mathbf{x}_{i_s}^T\boldsymbol\beta \, \mathbf{x}_{i_s}^T\boldsymbol\beta \right| \mid \mathbf{y_n}\right] \cr
 &= \sum_{s=1}^p a_s^2 \mathbb{E}\left[\left|\sum_{t=1}^p b_t \mathbf{x}_{i_t}^T\boldsymbol\beta \, \mathbf{x}_{i_s}^T\boldsymbol\beta \right| \mid \mathbf{y_n}\right] \cr
 & \leq \sum_{s=1}^p a_s^2 \sum_{t=1}^p b_t \mathbb{E}\left[\left|\mathbf{x}_{i_t}^T\boldsymbol\beta \, \mathbf{x}_{i_s}^T\boldsymbol\beta \right| \mid \mathbf{y_n}\right],
\end{align*}
and we proceed as before. In the second equality, we write $\mathbf{x}_{i_s}$ as a linear combination of $\mathbf{x}_{i_1},\ldots,\mathbf{x}_{i_p}$, $i_1,\ldots,i_p\in\{1,\ldots,n\}\setminus\{i_s\}$. To be able to do this, we need to select $p$ linearly independent vectors among the remaining $n-1\geq p$. It is possible given that $n>p+3$ and all explanatory variables are continuous. Therefore, $\mathbb{E}[|\beta_{j}|^2\mid \mathbf{y_n}]<\infty$ if $n>p+3$.

\subsubsection{Proof of Theorem~\ref{thm-main}}\label{sec-proof-a}

Recall that we assume that $\ell\leq n/2-(p-1/2)\Leftrightarrow k\geq n/2+(p-1/2)\Leftrightarrow k\geq \ell+2p-1$. In addition, we will assume that $\ell\geq 1$, i.e.\ that there is at least one outlier, otherwise the proof is trivial. A proposition and a lemma that are used in the proof are first given, and the proofs of Results~(a) to (d) follow. The proofs of this proposition and this lemma can be found in \cite{desgagne2015robustness}.

\begin{proposition}[Dominance]\label{prop-dominance}
If $s\in L_{0}(\infty)$ and $g\in L_{\rho}(\infty)$, then for all $\delta>0$, there exists a constant $A(\delta)> 1$ such that
$z\ge A(\delta)$ implies that
\begin{equation*}
(\log z)^{-\delta}< s(z) < (\log z)^{\delta}\hspace{3mm}\text{ and }\hspace{3mm}(\log z)^{-\rho-\delta}< g(z) < (\log z)^{-\rho+\delta}.
\end{equation*}
\end{proposition}

\begin{lemma}\label{cor-location-scale-transformation}
For all $\lambda\ge 0$, $\forall\tau\ge 1$, there exists a constant $D(\lambda,\tau)\ge 1$  such that $z\in\re$  and
$(\mu,\sigma)\in [-\lambda,\lambda]\times[1/\tau, \tau]$ implies that
\begin{equation*}
1/D(\lambda,\tau)\le (1/\sigma) f((z-\mu)/\sigma)/f(z)\le D(\lambda,\tau).
\end{equation*}
\end{lemma}
  Note that Lemma~\ref{cor-location-scale-transformation} is a corollary of Proposition~\ref{prop-location-scale-transformation}.

\begin{proof}[Proof of Result~(a)]

We first observe that
\begin{align*}
   \frac{m(\mathbf{y_n})}{m(\mathbf{y_k})\prod_{i=1}^{n}[f(y_i)]^{\ell_i}}&= \frac{m(\mathbf{y_n})}{m(\mathbf{y_k})\prod_{i=1}^{n}[f(y_i)]^{\ell_i}}
      \int_{\re^p}\int_{0}^{\infty}\pi(\boldsymbol\beta,\sigma\mid \mathbf{y_n})\,d\sigma\,d\boldsymbol\beta\\
   &=   \int_{\re^p}\int_{0}^{\infty}\frac{\pi(\boldsymbol\beta,\sigma)\prod_{i=1}^{n}
        \left[(1/\sigma)f((y_i-\mathbf{x}_i^T\boldsymbol\beta)/\sigma)\right]^{k_i+\ell_i}}{m(\mathbf{y_k})\prod_{i=1}^{n}[f(y_i)]^{\ell_i}}\,d\sigma\,d\boldsymbol\beta\\
   &=   \int_{\re^p}\int_{0}^{\infty}
        \pi(\boldsymbol\beta,\sigma\mid \mathbf{y_k})  \prod_{i=1}^{n}\left[\frac{(1/\sigma)f((y_i-\mathbf{x}_i^T\boldsymbol\beta)/\sigma)}{f(y_i)}\right]^{\ell_i}\,d\sigma\,d\boldsymbol\beta.
\end{align*}

We show that the last integral converges towards 1 as $\omega\rightarrow\infty$ to prove Result~(a).
If we use Lebesgue's dominated convergence theorem to interchange the limit $\omega\rightarrow\infty$ and the integral, we have
\begin{align*}
&\lim_{\omega\rightarrow\infty}\int_{\re^p}\int_{0}^{\infty}\pi(\boldsymbol\beta,\sigma\mid \mathbf{y_k})
 \prod_{i=1}^{n}\left[\frac{(1/\sigma)f((y_i-\mathbf{x}_i^T\boldsymbol\beta)/\sigma)}{f(y_i)}\right]^{\ell_i}\,d\sigma\,d\boldsymbol\beta\\
 &\qquad = \int_{\re^p}\int_{0}^{\infty}\lim_{\omega\rightarrow\infty}
  \pi(\boldsymbol\beta,\sigma\mid \mathbf{y_k})\prod_{i=1}^{n}\left[\frac{(1/\sigma)f((y_i-\mathbf{x}_i^T\boldsymbol\beta)/\sigma)}{f(y_i)}\right]^{\ell_i}
  \,d\sigma\,d\boldsymbol\beta \\
 &\qquad = \int_{\re^p}\int_{0}^{\infty}
  \pi(\boldsymbol\beta,\sigma\mid \mathbf{y_k})
  \,d\sigma\,d\boldsymbol\beta = 1,
\end{align*}
using Proposition~\ref{prop-location-scale-transformation} in the second equality, since $\mathbf{x}_1,\ldots,\mathbf{x}_n$ are fixed, and then Proposition~\ref{proposition-proper}. Note that the conditions of Proposition~\ref{proposition-proper} are satisfied because $k\geq \ell+2p-1\Rightarrow k\geq p+2$ (because $\ell\geq 1$ and $p\geq 2$). When any type of explanatory variables is considered, we select $p$ observations, say those with $\mathbf{x}_{i_1},\ldots,\mathbf{x}_{i_p}$, such that the matrix with rows $\mathbf{x}_{i_1}^T,\ldots,\mathbf{x}_{i_p}^T$ has a non-null determinant. This is possible given the assumption mentioned in Remark~\ref{rmk_thm}. Note also that pointwise convergence is sufficient, for any value of $\boldsymbol\beta\in\re^p$ and $\sigma>0$, once the limit is inside the integral. However, in order to use Lebesgue's dominated convergence theorem, we need to  prove that the integrand is bounded by an integrable function of $\boldsymbol\beta$ and $\sigma$ that does not depend on $\omega$, for any value of $\omega\ge \yo$, where $\yo$ is a constant. The constant $\yo$ can be chosen as large as we want, and minimum values for $\yo$ will be given throughout the proof. In order to bound the integrand, we divide the domain of integration into two areas: $1\leq\sigma<\infty$ and $0<\sigma<1$. Again, we want to separately analyse the area where the ratio $1/\sigma$ approaches infinity.

We assumed that $y_i$ can be written as $y_i=a_i+b_i \omega$, where $\omega\rightarrow\infty$, and $a_i$ and $b_i$ are constants such that $a_i\in\re$ and $b_i\ne 0$ if $\ell_i=1$ (if the observation is an outlier). Therefore, the ranking of the elements in the set $\{|y_i| : \ell_i=1\}$ is primarily determined by the values $|b_1|,\ldots,|b_n|$, and we can choose the constant $\yo$ larger than a certain threshold to ensure that this ranking remains unchanged for all $\omega\ge \yo$. Without loss of generality, we assume for convenience that
\begin{equation*}
\omega=\min_{\{i:\,\ell_i=1\}}|y_i|\hspace{5mm}\text{ and consequently }\hspace{5mm} \min_{\{i:\,\ell_i=1\}} |b_i|=1.
\end{equation*}
We now bound above the integrand on the first area.

\textbf{Area~1:} Consider $1\le\sigma<\infty$ and assume without loss of generality that $y_1,\ldots,$ $y_{\ell+2p-1}$ are $\ell+2p-1$ nonoutliers (therefore $k_1=\ldots=k_{\ell+2p-1}=1$). We have
\begin{align*}
&\pi(\boldsymbol\beta,\sigma\mid \mathbf{y_k})\prod_{i=1}^{n}\left[\frac{(1/\sigma)f\left((y_i-\mathbf{x}_i^T\boldsymbol\beta)/\sigma\right)}{f(y_i)}\right]^{\ell_i}
\propto\frac{\pi(\boldsymbol\beta,\sigma)}{\sigma^{n}}\prod_{i=1}^{n}\frac{f((y_i-\mathbf{x}_i^T\boldsymbol\beta)/\sigma)}{\left[f(y_i)\right]^{\ell_i}}\\
&\za{\le}\frac{B}{\sigma^{n}}\prod_{i=1}^{n} \frac{D(|a_i|,1)f((b_i \omega-\mathbf{x}_i^T\boldsymbol\beta)/\sigma)}{\left[f(y_i)\right]^{\ell_i}}\\
&\zb{\le}\frac{1}{[f(\omega)]^{\ell}}\frac{B}{\sigma^{n}}\prod_{i=1}^{n} D(|a_i|,1)f((b_i \omega-\mathbf{x}_i^T\boldsymbol\beta)/\sigma)\left[|b_i|D(|a_i|,|b_i|)\right]^{\ell_i}\\
&\propto\frac{1}{[f(\omega)]^{\ell}}\frac{1}{\sigma^{n}}\prod_{i=1}^{n} f((b_i \omega-\mathbf{x}_i^T\boldsymbol\beta)/\sigma)\\
&\zc{=}\frac{1}{[f(\omega)]^{\ell}}\frac{1}{\sigma^{n}}
 \prod_{i=1}^{n}\left[f(\mathbf{x}_i^T\boldsymbol\beta/\sigma)\right]^{k_i}\left[f((b_i \omega-\mathbf{x}_i^T\boldsymbol\beta)/\sigma)\right]^{\ell_i}\\
 &\zd{=} \frac{\prod_{i=1}^p(1/\sigma)f(\mathbf{x}_i^T\boldsymbol\beta/\sigma)}{\sigma^{k-p-1/2}}\left[\frac{\omega/\sigma}{\omega f(\omega)}\right]^{\ell}\frac{1}{\sigma^{1/2}}\prod_{i=p+1}^{n}\left[f(\mathbf{x}_i^T\boldsymbol\beta/\sigma)\right]^{k_i}\left[f((b_i \omega-\mathbf{x}_i^T\boldsymbol\beta)/\sigma)\right]^{\ell_i}.
\end{align*}
In Step $a$, we use $y_i=a_i+b_i \omega$ and Lemma~\ref{cor-location-scale-transformation} to obtain
\begin{equation*}
f((y_i-\mathbf{x}_i^T\boldsymbol\beta)/\sigma )=f((b_i \omega-\mathbf{x}_i^T\boldsymbol\beta)/\sigma+a_i/\sigma)\le D(|a_i|,1)f((b_i \omega-\mathbf{x}_i^T\boldsymbol\beta)/\sigma),
\end{equation*}
because $|a_i/\sigma|\le |a_i|$ for all $i$. We also use $\pi(\boldsymbol\beta,\sigma)\le B \max(1,1/\sigma)= B$.
In Step $b$, we use again Lemma~\ref{cor-location-scale-transformation} to obtain $f(\omega)/f(y_i)=f((y_i-a_i)/b_i)/f(y_i)\le |b_i|D(|a_i|,|b_i|)$.
In Step $c$, we set $b_i=0$ if $k_i=1$ and we use the symmetry of $f$ to obtain $f(-\mathbf{x}_i^T\boldsymbol\beta/\sigma)=f(\mathbf{x}_i^T\boldsymbol\beta/\sigma)$.
In Step $d$, we use the assumption $k_1=\ldots=k_p=1$.

Now it suffices to demonstrate that
\begin{equation}\label{fct1_area1}
\left[\frac{\omega/\sigma}{\omega f(\omega)}\right]^{\ell}\frac{1}{\sigma^{1/2}}\prod_{i=p+1}^{n}\left[f(\mathbf{x}_i^T\boldsymbol\beta/\sigma)\right]^{k_i}\left[f((b_i \omega-\mathbf{x}_i^T\boldsymbol\beta)/\sigma)\right]^{\ell_i}
  \end{equation}
is bounded by a constant that does not depend on $\omega,\boldsymbol\beta$ and $\sigma$ since  $\prod_{i=1}^p(1/\sigma)f(\mathbf{x}_i^T\boldsymbol\beta/\sigma)$ \newline
$\times(1/\sigma)^{k-p-1/2}$ is an integrable function on area~1. Indeed, since $k>p+1$,  we have
\begin{align*}
&\int_{1}^{\infty}(1/\sigma)^{k-p-1/2}\int_{\re^p}\prod_{i=1}^p(1/\sigma) f(\mathbf{x}_i^T\boldsymbol\beta/\sigma) \,d\boldsymbol\beta\,d\sigma \cr
&\qquad=\abs{\text{det}\left(\begin{array}{c}\mathbf{x}_1^T \cr \vdots \cr \mathbf{x}_p^T\end{array}\right)}^{-1}  \int_{1}^{\infty}\frac{1}{\sigma^{k-p-1/2}}\,d\sigma<\infty,
\end{align*}
using the following change of variables: $u_i=\mathbf{x}_i^T\boldsymbol\beta/\sigma$, $i=1,\ldots,p$. The determinant is different from 0 because all the explanatory variables are continuous. When any type of explanatory variables is considered, we assume without loss of generality that $\mathbf{x}_{1},\ldots,\mathbf{x}_{p}$ are additionally linearly independent. Note that if instead, in Step $a$ above, we bound $\pi(\boldsymbol\beta,\sigma)$ by $B/\sigma$, one can verify that the condition $k\ge p+1$ is sufficient to bound above the integral.

In order to bound the function in (\ref{fct1_area1}), we split Area~1 into three parts: $1\leq \sigma<\omega^{1/2}$, $\omega^{1/2}\leq \sigma< \omega/(\gamma M)$ and $\omega/(\gamma M)\leq \sigma<\infty$, where $M$ is defined in (\ref{eqn-monotonic}) and $\gamma$ is a positive constant that can be chosen as large as we want (lower bounds are provided in the proof). Note that this split is well defined if $\yo> \max(1,(\gamma M)^2)$ because $\omega\ge \yo$.

First, consider $\omega/(\gamma M)\le\sigma<\infty$. We have
\begin{align*}
&\left[\frac{\omega/\sigma}{\omega f(\omega)}\right]^{\ell}\frac{1}{\sigma^{1/2}}\prod_{i=p+1}^{n}\left[f(\mathbf{x}_i^T\boldsymbol\beta/\sigma)\right]^{k_i} \left[f((b_i \omega-\mathbf{x}_i^T\boldsymbol\beta)/\sigma)\right]^{\ell_i}
\za{\le} \frac{B^{n-p}}{\sigma^{1/2}}\left[\frac{\omega/\sigma}{\omega f(\omega)}\right]^{\ell}\cr
&\hspace{10mm}\zb{\le}B^{n-p}(\gamma M)^{\ell+1/2}\frac{(1/\omega)^{1/2}}{[\omega f(\omega)]^{\ell}}
\zc{\le}B^{n-p}(\gamma M)^{\ell+1/2}\frac{(1/\omega)^{1/2}}{{(\log \omega)^{-(\rho+1)\ell}}}\\
&\hspace{10mm}\zd{\le}B^{n-p}(\gamma M)^{\ell+1/2}[2(\rho+1)\ell/\ee]^{(\rho+1)\ell}<\infty.
\end{align*}
In Step $a$, we use $f\le B$.
In Step $b$, we use $\omega/\sigma\le \gamma M$ and $1/\sigma\le \gamma M/\omega$.
In Step $c$, we use $\omega f(\omega)>(\log \omega)^{-\rho-1}$ if $\omega\ge \yo\ge A(1)$, where $A(1)$ comes from
    Proposition~\ref{prop-dominance}.
For Step $d$, it is purely algebraic to show that the maximum of $(\log \omega)^{\xi}/\omega^{1/2}$ is $(2\xi/\ee)^{\xi}$
      for $\omega>1$ and $\xi>0$, where $\xi=(\rho+1)\ell$ in our situation.

Now, consider the two other parts combined (we will split them in the next step), that is $1\le\sigma\le \omega/(\gamma M)$. We have
\begin{align*}
&\left[\frac{\omega/\sigma}{\omega f(\omega)}\right]^{\ell}\frac{1}{\sigma^{1/2}}\prod_{i=p+1}^{n}\left[f(\mathbf{x}_i^T\boldsymbol\beta/\sigma)\right]^{k_i}\left[f((b_i \omega-\mathbf{x}_i^T\boldsymbol\beta)/\sigma)\right]^{\ell_i}\\
&\qquad=\frac{1}{\sigma^{1/2}}\left[\frac{(\omega/\sigma)f(\omega/\sigma)}{\omega f(\omega)}\right]^{\ell}
  \prod_{i=p+1}^{n}\left[f(\mathbf{x}_i^T\boldsymbol\beta/\sigma)\right]^{k_i}\left[\frac{f((b_i \omega-\mathbf{x}_i^T\boldsymbol\beta)/\sigma)}{f(\omega/\sigma)}\right]^{\ell_i} \cr
&\qquad\za{\le} \frac{1}{\sigma^{1/2}}
  \left[\frac{(\omega/\sigma)f(\omega/\sigma)}{\omega f(\omega)}\right]^{\ell}B^{k-p}[D(0,\gamma)\gamma]^{\ell}.
\end{align*}
In Step $a$, we use
 \begin{align}\label{inequality1}
 &\prod_{i=p+1}^{n}\left[f(\mathbf{x}_i^T\boldsymbol\beta/\sigma)\right]^{k_i} \left[\frac{f((b_i \omega-\mathbf{x}_i^T\boldsymbol\beta)/\sigma)}{f(\omega/\sigma)}\right]^{\ell_i}\leq B^{k-p}[D(0,\gamma)\gamma]^{\ell}.
\end{align}
The proof of this inequality is substantial. Therefore, to ease the reading, it is deferred after the demonstration that the remaining term, i.e.\
$$
  \frac{1}{\sigma^{1/2}}
  \left[\frac{(\omega/\sigma)f(\omega/\sigma)}{\omega f(\omega)}\right]^{\ell},
$$
is bounded.

We first consider $\omega^{1/2}\le\sigma\le \omega/(\gamma M)$. We have
\begin{align*}
\frac{1}{\sigma^{1/2}}
  \left[\frac{(\omega/\sigma)f(\omega/\sigma)}{\omega f(\omega)}\right]^{\ell} \za{\le} B^{\ell}
 \frac{(1/\omega)^{1/4}}{[\omega f(\omega)]^{\ell}}&\zb{\le} B^{\ell}
 \frac{(1/\omega)^{1/4}}{(\log \omega)^{-(\rho+1)\ell}} \cr
 &\zc{\le} B^{\ell}[4(\rho+1)\ell/\ee]^{(\rho+1)\ell}<\infty.
\end{align*}
In Step $a$, we use $(\omega/\sigma)f(\omega/\sigma)\le B$ and $(1/\sigma)^{1/2}\le (1/\omega)^{1/4}$.
In Step $b$, we use $\omega f(\omega)>(\log \omega)^{-\rho-1}$ if $\omega\ge \yo\ge A(1)$, where $A(1)$ comes from
    Proposition~\ref{prop-dominance}.
In Step $c$, it is purely algebraic to show that the maximum of $(\log \omega)^{\xi}/\omega^{1/4}$ is $(4\xi/\ee)^{\xi}$
      for $\omega>1$ and $\xi>0$, where $\xi=(\rho+1)\ell$ in our situation.

We now consider $1\le\sigma\le \omega^{1/2}$. We have
\begin{equation*}
\frac{1}{\sigma^{1/2}}
  \left[\frac{(\omega/\sigma)f(\omega/\sigma)}{\omega f(\omega)}\right]^{\ell}\za{\le}
  \left[\frac{\omega^{1/2}f(\omega^{1/2})}{\omega f(\omega)}\right]^{\ell}\zb{\le} 2^{(\rho+1)\ell}<\infty.
\end{equation*}
In Step $a$, we use $1/\sigma\le 1$ and $(\omega/\sigma)f(\omega/\sigma)\le \omega^{1/2} f(\omega^{1/2})$ by the monotonicity of the tails of
$|z| f(z)$ since $\omega/\sigma\ge \omega^{1/2}\ge \yo^{1/2}\ge M$ if $\yo\ge M^2$. In Step $b$, we use $\omega^{1/2} f(\omega^{1/2})/(\omega
f(\omega))$ $\le 2(1/2)^{-\rho}=2^{\rho+1}$
if $\omega\ge \yo\ge A$, where $A$ is a positive constant, see the definition of log-regularly varying functions (Definition~\ref{def-log-regularly}).

Finally, we prove the inequality in (\ref{inequality1}). Recall that we assumed without loss of generality that the first $\ell+2p-1$ observations are nonoutliers (therefore $k_1=\ldots=k_{\ell+2p-1}=1$). We know that $\mathbf{x}_1,\ldots,\mathbf{x}_p$ have been used earlier to integrate over $\boldsymbol\beta$ and $\sigma$. We also know that there are at least $p$ remaining nonoutliers among observations $1$ to $\ell+2p-1$ because $\ell+2p-1 -p=\ell+p-1\geq p$ (because we assume that $\ell\geq1$).

In order to prove the result, we split the domain of $\boldsymbol\beta$ as follows:
\begin{align}\label{eqn_domain_beta}
 &\re^p=\left[\cap_i \mathcal{O}_i^c \right]\cup \left[\cup_i \left( \mathcal{O}_i\cap\left(\cap_{i_1} \mathcal{F}_{i_1}^c\right)\right)\right]\cup \left[\cup_{i,i_1}\left(\mathcal{O}_i\cap \mathcal{F}_{i_1}\cap\left(\cap_{i_2\neq i_1}\mathcal{F}_{i_2}^c\right)\right)\right] \cr
 &\cup\cdots\cup \left[\cup_{i,i_1,\ldots,i_{p-1} (i_j\neq i_s \, \forall i_j,i_s \text{ s.t. } j\neq s)}\left(\mathcal{O}_i\cap \mathcal{F}_{i_1}\cap \cdots\cap\mathcal{F}_{i_{p-1}}\cap \left(\cap_{i_p\neq i_1,\ldots,i_{p-1}}\mathcal{F}_{i_p}^c\right)\right)\right] \cr
 &\qquad\qquad\qquad \cup \left[\cup_{i,i_1,\ldots,i_{p} (i_j\neq i_s \, \forall i_j,i_s \text{ s.t. } j\neq s)}\left(\mathcal{O}_i\cap \mathcal{F}_{i_1}\cap \cdots\cap\mathcal{F}_{i_{p}}\right)\right],
\end{align}
where
\begin{align}\label{def_O_i}
\mathcal{O}_i&:=\{\boldsymbol\beta:|b_i\omega-\mathbf{x}_i^T\boldsymbol\beta|< \omega/2\},\forall i\in \mathcal{I}_\mathcal{O},
\end{align}
\begin{align}\label{def_F_i}
\mathcal{F}_i&:=\{\boldsymbol\beta:|\mathbf{x}_i^T\boldsymbol\beta|<\omega/\gamma\},\forall i\in\mathcal{I}_\mathcal{F},
\end{align}
 $\mathcal{I}_\mathcal{O}:=\{i:i\in\{\ell+2p-1,\ldots,n\} \text{ and } \ell_i=1\}$ and $\mathcal{I}_\mathcal{F}:=\{p+1,\ldots,\ell+2p-1\}$ are the sets of indexes of outliers and remaining fixed observations (nonoutliers) among observations $1$ to $\ell+2p-1$, respectively.

 The set $\mathcal{O}_{i}$ represents the hyperplanes $y= \mathbf{x}_i^T\boldsymbol\beta$ characterised by the different values of $\boldsymbol\beta$ that satisfy $|b_i\omega-\mathbf{x}_i^T\boldsymbol\beta|< \omega/2$. In other words, it represents the hyperplanes that pass at a vertical distance of less than $\omega/2$ of the point $(\mathbf{x}_i,b_i\omega)$, which is considered as an outlier since $\omega\rightarrow\infty$ (recall that $b_i\omega=y_i-a_i$). Analogously, the set $\mathcal{F}_{i}$ represents the hyperplanes that pass at a vertical distance of less than $\omega/\gamma$ of the point $(\mathbf{x}_i,0)$, which is considered as a nonoutlier. Therefore, the set $\cap_i \mathcal{O}_i^c$ represents the hyperplanes that pass at a vertical distance of at least $\omega/2$ of all the points $(\mathbf{x}_i,b_i\omega)$ (all the outliers). The set $\cup_i (\mathcal{O}_i\cap(\cap_{i_1} \mathcal{F}_{i_1}^c))$ represents the hyperplanes that pass at a vertical distance of less than $\omega/2$ of at least one point $(\mathbf{x}_i,b_i\omega)$ (an outlier), but at a vertical distance of at least $\omega/\gamma$ of all the points $(\mathbf{x}_i,0)$ (all the nonoutliers). For each $i_1\in\mathcal{I}_\mathcal{F}$, the set $\cup_i(\mathcal{O}_i\cap\mathcal{F}_{i_1}\cap(\cap_{i_2\neq i_1} \mathcal{F}_{i_2}^c))$ represents the hyperplanes that pass at a vertical distance of less than $\omega/2$ of at least one point $(\mathbf{x}_i,b_i\omega)$ (an outlier), at a vertical distance of less than $\omega/\gamma$ of the point $(\mathbf{x}_{i_1},0)$ (a nonoutlier), but at a vertical distance of at least $\omega/\gamma$ of all the other nonoutliers, and so on.

Now, we claim that $\mathcal{O}_i\cap \mathcal{F}_{i_1}\cap \cdots\cap\mathcal{F}_{i_{p}}=\varnothing$ for all $i,i_1,\ldots,i_{p}$ with $i_j\neq i_s,\forall i_j,i_s$ such that $j\neq s$, meaning that there is no hyperplane that passes at a vertical distance of less than $\omega/2$ of the point $(\mathbf{x}_i,b_i\omega)$ (an outlier) and at the same at a vertical distance of less than $\omega/\gamma$ of $p$ points $(\mathbf{x}_{i_j},0)$ (nonoutliers). To prove this, we use the fact that $\mathbf{x}_i$ (a vector of size $p$) can be expressed as a linear combination of $\mathbf{x}_{i_1},\ldots,\mathbf{x}_{i_{p}}$. This is true because all explanatory variables are continuous, and therefore, linearly independent with probability 1. When any type of explanatory variables is considered, we select observations $1$ to $\ell+2p-1$ to be such that any $p$ vectors $\mathbf{x}_{i_1},\ldots,\mathbf{x}_{i_p}$, with $\{i_1,\ldots,i_p\}\subset\{1,\ldots, \ell+2p-1\}$, are linearly independent. This is possible given the assumption mentioned in Remark~2.2. As a result, considering that $\boldsymbol\beta\in \mathcal{F}_{i_1}\cap \cdots\cap\mathcal{F}_{i_{p}}$ and $\mathbf{x}_i=\sum_{s=1}^{p} a_s \mathbf{x}_{i_s}$ for some $a_1,\ldots,a_{p}\in\re$, we have
\begin{align*}
|b_i\omega-\mathbf{x}_i^T\boldsymbol\beta|= \left|b_i\omega-\left(\sum_{s=1}^{p}a_s\mathbf{x}_{i_s}\right)^T\boldsymbol\beta\right|\za{\geq} |b_i\omega|-\left|\sum_{s=1}^{p}a_s\mathbf{x}_{i_s}^T\boldsymbol\beta\right| &\zb{\geq}\omega-\frac{\omega}{\gamma}\sum_{s=1}^{p}|a_s| \cr
&\zc{\geq}\omega-\frac{\omega}{2}.
\end{align*}
 In Step $a$, we use the reverse triangle inequality. In Step $b$, we use that $|b_i|\geq 1$ and $|\sum_{s=1}^{p}a_s\mathbf{x}_{i_s}^T\boldsymbol\beta|\leq\sum_{s=1}^{p}|a_s||\mathbf{x}_{i_s}^T\boldsymbol\beta|\leq\sum_{s=1}^{p}|a_s|\omega/\gamma  $ because $\boldsymbol\beta\in \mathcal{F}_{i_1}\cap \cdots\cap\mathcal{F}_{i_{p}}$, which means that $|\mathbf{x}_i^T\boldsymbol\beta|<\omega/\gamma$ for all $i\in\{i_1,\ldots,i_p\}$. In Step $c$, we define the constant $\gamma$ such that $\gamma\geq2\sum_{s=1}^{p}|a_s|$ (we choose $\gamma$ such that it satisfies this inequality for any combination of $i$ and $i_1,\ldots,i_{p}$). Therefore, we have that $\boldsymbol\beta\notin \mathcal{O}_i$. This proves that $\mathcal{O}_i\cap \mathcal{F}_{i_1}\cap \cdots\cap\mathcal{F}_{i_{p}}=\varnothing$ for all $i,i_1,\ldots,i_{p}$ with $i_j\neq i_s,\forall i_j,i_s$ such that $j\neq s$. This in turn implies that (\ref{eqn_domain_beta}) can be rewritten as
\begin{align*}
 &\re^p=\left[\cap_i \mathcal{O}_i^c \right]\cup \left[\cup_i \left( \mathcal{O}_i\cap\left(\cap_{i_1} \mathcal{F}_{i_1}^c\right)\right)\right]\cup \left[\cup_{i,i_1}\left(\mathcal{O}_i\cap \mathcal{F}_{i_1}\cap\left(\cap_{i_2\neq i_1}\mathcal{F}_{i_2}^c\right)\right)\right] \cr
 &\cup\cdots\cup \left[\cup_{i,i_1,\ldots,i_{p-1} (i_j\neq i_s \, \forall i_j,i_s \text{ s.t. } j\neq s)}\left(\mathcal{O}_i\cap \mathcal{F}_{i_1}\cap \cdots\cap\mathcal{F}_{i_{p-1}}\cap \left(\cap_{i_p\neq i_1,\ldots,i_{p-1}}\mathcal{F}_{i_p}^c\right)\right)\right].
\end{align*}
This decomposition of $\re^p$ is comprised of $1+\sum_{i=0}^{p-1} {{\ell+p-1} \choose {i}}$ mutually exclusive sets given by $\cap_i \mathcal{O}_i^c$, $\cup_i ( \mathcal{O}_i\cap(\cap_{i_1} \mathcal{F}_{i_1}^c))$, $\cup_{i}(\mathcal{O}_i\cap \mathcal{F}_{i_1}\cap(\cap_{i_2\neq i_1}\mathcal{F}_{i_2}^c))$ for $i_1\in\mathcal{I}_\mathcal{F} $, and so on.

 We are now ready to bound the function on the left-hand side in (\ref{inequality1}). We first show that the function is bounded on $\boldsymbol\beta\in \cap_i \mathcal{O}_i^c$ and $1\le\sigma\le \omega/(\gamma M)$. For all $i\in\mathcal{I}_\mathcal{O}$, we have
$$
 \frac{f((b_i \omega-\mathbf{x}_i^T\boldsymbol\beta )/\sigma)}{f(\omega/\sigma)}\leq \frac{f(\omega/(2\sigma))}{f(\omega/\sigma)}\leq 2D(0,2)\leq  D(0,\gamma)\gamma,
$$
using the monotonicity of $f$ because $|b_i\omega-\mathbf{x}_i^T\boldsymbol\beta|/\sigma\geq \omega/(2\sigma)\geq \gamma M/2\geq M$ (we choose $\gamma\geq 2$), and then Lemma~\ref{cor-location-scale-transformation}. Therefore, on $\boldsymbol\beta\in \cap_i \mathcal{O}_i^c$ and $1\le\sigma\le \omega/(\gamma M)$,
$$
\prod_{i=p+1}^{n}\left[f(\mathbf{x}_i^T\boldsymbol\beta/\sigma)\right]^{k_i} \left[\frac{f((b_i \omega-\mathbf{x}_i^T\boldsymbol\beta)/\sigma)}{f(\omega/\sigma)}\right]^{\ell_i}\leq B^{k-p}[D(0,\gamma)\gamma]^{\ell},
$$
using $f\leq B$.

Now, we consider the area defined by: $1\le\sigma\le \omega/(\gamma M)$ and  $\boldsymbol\beta$ belongs to one of the $\sum_{i=0}^{p-1}$ $ {{k-p} \choose {i}}$ mutually exclusive sets $\cup_i ( \mathcal{O}_i\cap(\cap_{i_1} \mathcal{F}_{i_1}^c))$, $\cup_{i}(\mathcal{O}_i\cap \mathcal{F}_{i_1}\cap(\cap_{i_2\neq i_1}\mathcal{F}_{i_2}^c))$ for $i_1\in\mathcal{I}_\mathcal{F} $, etc. We have
\begin{align*}
 \prod_{i=p+1}^{n}\left[f(\mathbf{x}_i^T\boldsymbol\beta/\sigma)\right]^{k_i} \left[\frac{f((b_i \omega-\mathbf{x}_i^T\boldsymbol\beta)/\sigma)}{f(\omega/\sigma)}\right]^{\ell_i}  &\za{\leq} B^{\ell} \prod_{i=p+1}^{n}\frac{\left[f(\mathbf{x}_i^T\boldsymbol\beta/\sigma)\right]^{k_i}}{\left[f(\omega/\sigma)\right]^{\ell_i}} \cr
 &\hspace{-15mm}\zb{\leq} B^{\ell+(k-p)-\ell} [D(0,\gamma)\gamma]^{\ell}=B^{k-p} [D(0,\gamma)\gamma]^{\ell}.
\end{align*}
In Step $a$, we use $f\leq B$ for all $i\in \mathcal{I}_\mathcal{O}$. In Step $b$, we use the fact that in any of the sets in which $\boldsymbol\beta$ can belong, there are at least $\ell$ nonoutlying points $(\mathbf{x}_i,0)$ such that $|\mathbf{x}_i^T\boldsymbol\beta|\geq\omega/\gamma$. Indeed, the case in which there are the least nonoutliers such that $|\mathbf{x}_i^T\boldsymbol\beta|\geq\omega/\gamma$ corresponds to $\boldsymbol\beta\in\cup_i(\mathcal{O}_i\cap \mathcal{F}_{i_1}\cap \cdots\cap\mathcal{F}_{i_{p-1}}\cap (\cap_{i_p\neq i_1,\ldots,i_{p-1}}\mathcal{F}_{i_p}^c))$. In this case there are $p-1$ nonoutliers such that $|\mathbf{x}_i^T\boldsymbol\beta|<\omega/\gamma$ (observations $i_1$ to $i_{p-1}$), which leaves $\ell+p-1-(p-1)=\ell$ nonoutliers such that $|\mathbf{x}_i^T\boldsymbol\beta|\geq\omega/\gamma$ (i.e. that there are $\ell$ sets in the intersection $\cap_{i_p\neq i_1,\ldots,i_{p-1}}\mathcal{F}_{i_p}^c)$. 
Therefore, for $\ell$ nonoutliers such that $|\mathbf{x}_i^T\boldsymbol\beta|\geq\omega/\gamma$, we use
\begin{align*}
 f(\mathbf{x}_i^T\boldsymbol\beta/\sigma)/f(\omega/\sigma)\leq f(\omega/(\gamma\sigma))/f(\omega/\sigma) \leq  D(0,\gamma)\gamma,
\end{align*}
by the monotonicity of $f$ because $|\mathbf{x}_i^T\boldsymbol\beta|/\sigma\geq \omega/(\gamma\sigma)\geq M$, and then Lemma~\ref{cor-location-scale-transformation}. For the remaining $k-p-\ell$ nonoutlying points, we use $f\leq B$. Note that this argument justifies the need of the assumption $k\geq \ell+2p-1$.

\textbf{Area~2:} Consider $0<\sigma< 1$. We actually need to show that
\begin{align*}
  \lim_{\omega\rightarrow\infty}\int_{\re^p}\int_{0}^{1}
  \pi(\boldsymbol\beta,\sigma\mid \mathbf{y_k})&\prod_{i=1}^{n}\left[\frac{(1/\sigma)f((y_i-\mathbf{x}_i^T\boldsymbol\beta)/\sigma)}{f(y_i)}\right]^{\ell_i}
  \,d\sigma\,d\boldsymbol\beta \cr
  &\qquad= \int_{\re^p}\int_{0}^{1}\pi(\boldsymbol\beta,\sigma\mid \mathbf{y_k})\,d\sigma\,d\boldsymbol\beta.
\end{align*}
For Area~2, we proceed in a slightly different manner than for Area~1. We begin by dividing the first integral above into two parts as follows:
\begin{align*}
  &\lim_{\omega\rightarrow\infty}\int_{\re^p}\int_{0}^{1}
  \pi(\boldsymbol\beta,\sigma\mid \mathbf{y_k})\prod_{i=1}^{n}\left[\frac{(1/\sigma)f((y_i-\mathbf{x}_i^T\boldsymbol\beta)/\sigma)}{f(y_i)}\right]^{\ell_i}
  \,d\sigma\,d\boldsymbol\beta \\
  &\quad=\lim_{\omega\rightarrow\infty} \int_{\re^p}\int_{0}^{1}
  \pi(\boldsymbol\beta,\sigma\mid \mathbf{y_k})\prod_{i=1}^{n}\left[\frac{(1/\sigma)f((y_i-\mathbf{x}_i^T\boldsymbol\beta)/\sigma)}{f(y_i)}\right]^{\ell_i}
  \ind_{\cap_i \mathcal{O}_i^c}(\boldsymbol\beta)\,d\sigma\,d\boldsymbol\beta \\
  &\qquad+\lim_{\omega\rightarrow\infty} \int_{\cup_i \mathcal{O}_i}\int_{0}^{1}
  \pi(\boldsymbol\beta,\sigma\mid \mathbf{y_k})\prod_{i=1}^{n}\left[\frac{(1/\sigma)f((y_i-\mathbf{x}_i^T\boldsymbol\beta)/\sigma)}{f(y_i)}\right]^{\ell_i}
  \,d\sigma\,d\boldsymbol\beta,
\end{align*}
where
\begin{align*}
&\mathcal{O}_i:=\{\boldsymbol\beta:|y_i-\mathbf{x}_i^T\boldsymbol\beta|< \omega/2\},\forall i\in\mathcal{I}_{\mathcal{O}},
\end{align*}
with $\mathcal{I}_{\mathcal{O}}:=\{i:i\in\{1,\ldots,n\} \text{ and } \ell_i=1\}$. Note that the definition of $\mathcal{O}_i$ is very similar as that of the set defined in (\ref{def_O_i}) (this is why we use the same notation); its interpretation is also very similar. We show that the first part above is equal to the integral $\int_{\re^p}\int_{0}^{1}\pi(\boldsymbol\beta,\sigma\mid \mathbf{y_k})\,d\sigma\,d\boldsymbol\beta$ and that the second part is equal to 0.

For the first part, we again use Lebesgue's dominated convergence theorem in order to interchange the limit $\omega\rightarrow\infty$ and the integral. We have
\begin{align*}
&\lim_{\omega\rightarrow\infty}\int_{\re^p}\int_{0}^{1}\pi(\boldsymbol\beta,\sigma\mid \mathbf{y_k})\prod_{i=1}^{n} \left[\frac{(1/\sigma)f((y_i-\mathbf{x}_i^T\boldsymbol\beta)/\sigma)}{f(y_i)}\right]^{\ell_i} \ind_{\cap_i \mathcal{O}_i^c}(\boldsymbol\beta)\,d\sigma\,d\boldsymbol\beta\cr
 &\quad = \int_{\re^p}\int_{0}^{1}\pi(\boldsymbol\beta,\sigma\mid \mathbf{y_k}) \lim_{\omega\rightarrow\infty}\prod_{i=1}^{n}\left[\frac{(1/\sigma)f((y_i-\mathbf{x}_i^T\boldsymbol\beta)/\sigma)}{f(y_i)}\right]^{\ell_i} \ind_{\cap_i \mathcal{O}_i^c}(\boldsymbol\beta)\,d\sigma\,d\boldsymbol\beta\cr
 &\quad =\int_{\re^p}\int_{0}^{1} \pi(\boldsymbol\beta,\sigma\mid \mathbf{y_k})\times 1 \times \ind_{\re^p}(\boldsymbol\beta)
  \,d\sigma\,d\boldsymbol\beta =\int_{\re^p}\int_{0}^{1}  \pi(\boldsymbol\beta,\sigma\mid \mathbf{y_k})\,d\sigma\,d\boldsymbol\beta,
\end{align*}
using Proposition~\ref{prop-location-scale-transformation} in the second equality since $\mathbf{x}_1,\ldots,\mathbf{x}_n$ are fixed, and $\lim_{\omega\rightarrow\infty}$ $ \ind_{\cap_i\mathcal{O}_i^c}(\boldsymbol\beta)=\ind_{\re^p}(\boldsymbol\beta)=1\Leftrightarrow \lim_{\omega\rightarrow\infty} \ind_{\cup_i\mathcal{O}_i}(\boldsymbol\beta)=0$. Indeed, if $\ell_i=1$ and $b_i>0$ (which implies that $y_i>0$), $\boldsymbol\beta\in\mathcal{O}_i$ implies that $|y_i-\mathbf{x}_i^T\boldsymbol\beta|< \omega/2\leq y_i/2$, which in turn implies that  $y_i/2<\mathbf{x}_i^T\boldsymbol\beta<3y_i/2$, and in the limit, no $\boldsymbol\beta\in\re^p$ satisfies this (we have the same conclusion if $b_i<0$). Note that pointwise convergence is sufficient, for any value of $\boldsymbol\beta\in\re^p$ and $\sigma>0$, once the limit is inside the integral. We now demonstrate that the integrand is bounded, for any value of $\omega\ge \yo$, by an integrable function of $\boldsymbol\beta$ and $\sigma$ that does not depend on $\omega$.

Consider $\boldsymbol\beta\in \cap_i \mathcal{O}_i^c$, that is $\{\boldsymbol\beta:|y_i-\mathbf{x}_i^T\boldsymbol\beta|\geq \omega/2 \text{ for all } i\in\mathcal{I}_{\mathcal{O}}\}$, and $0<\sigma< 1$. Note that the integrand is equal to 0 if $\boldsymbol\beta\notin \cap_i \mathcal{O}_i^c$. For all $i\in\mathcal{I}_{\mathcal{O}}$, we have
\begin{equation*}
 (1/\sigma)f((y_i-\mathbf{x}_i^T\boldsymbol\beta)/\sigma)\leq f(y_i-\mathbf{x}_i^T\boldsymbol\beta)\leq f(\omega/2)\leq 2|b_i|D(|a_i|,2|b_i|)f(y_i),
\end{equation*}
 by the monotonicity of the tails of $|z| f(z)$ and then the monotonicity of the tails of $f(z)$, because $|y_i-\mathbf{x}_i^T\boldsymbol\beta|/\sigma\geq |y_i-\mathbf{x}_i^T\boldsymbol\beta|\geq \omega/2\geq\yo/2\geq M$, if we choose $\yo\ge 2M$. Lemma~\ref{cor-location-scale-transformation} is used in the last inequality with $\omega=(y_i-a_i)/b_i$. Therefore,
\begin{align*}
&\pi(\boldsymbol\beta,\sigma\mid \mathbf{y_k})\prod_{i=1}^{n}\left[\frac{(1/\sigma)f((y_i-\mathbf{x}_i^T\boldsymbol\beta)/\sigma)}{f(y_i)}\right]^{\ell_i} \ind_{\cap_i \mathcal{O}_i^c}(\boldsymbol\beta) \cr
&\qquad\le\pi(\boldsymbol\beta,\sigma\mid \mathbf{y_k})\prod_{i=1}^n [2|b_i|D(|a_i|,2|b_i|)]^{\ell_i},
\end{align*}
which is an integrable function.

We now prove that
\begin{equation*}
\lim_{\omega\rightarrow\infty} \int_{\cup_i \mathcal{O}_i}\int_{0}^{1}
  \pi(\boldsymbol\beta,\sigma\mid \mathbf{y_k})\prod_{i=1}^{n}\left[\frac{(1/\sigma)f((y_i-\mathbf{x}_i^T\boldsymbol\beta)/\sigma)}{f(y_i)}\right]^{\ell_i}
  \,d\sigma\,d\boldsymbol\beta=0.
  \end{equation*}
We first bound above the integrand and then we prove that the integral of the upper bound converges towards 0 as $\omega\rightarrow\infty$. In the same manner as in the proof of the inequality in (\ref{inequality1}), we split the domain of $\boldsymbol\beta$ as follows:
\begin{align*}
 &\cup_i \mathcal{O}_i=\left[\cup_i \left( \mathcal{O}_i\cap\left(\cap_{i_1} \mathcal{F}_{i_1}^c\right)\right)\right]\cup \left[\cup_{i,i_1}\left(\mathcal{O}_i\cap \mathcal{F}_{i_1}\cap\left(\cap_{i_2\neq i_1}\mathcal{F}_{i_2}^c\right)\right)\right] \cr
 &\cup\cdots\cup \left[\cup_{i,i_1,\ldots,i_{p-1} (i_j\neq i_s \, \forall i_j,i_s \text{ s.t. } j\neq s)}\left(\mathcal{O}_i\cap \mathcal{F}_{i_1}\cap \cdots\cap\mathcal{F}_{i_{p-1}}\cap \left(\cap_{i_p\neq i_1,\ldots,i_{p-1}}\mathcal{F}_{i_p}^c\right)\right)\right] \cr
 &\qquad\qquad\cup \left[\cup_{i,i_1,\ldots,i_{p} (i_j\neq i_s \, \forall i_j,i_s \text{ s.t. } j\neq s)}\left(\mathcal{O}_i\cap \mathcal{F}_{i_1}\cap \cdots\cap\mathcal{F}_{i_{p}}\right)\right],
\end{align*}
where
\begin{align*}
&\mathcal{F}_i:=\{\boldsymbol\beta:|\mathbf{x}_i^T\boldsymbol\beta|<\omega/\gamma\},\forall i\in\mathcal{I}_\mathcal{F},
\end{align*}
 and $\mathcal{I}_\mathcal{F}:=\{1,\ldots,\ell+2p-1\}$ (we assume as previously that $y_1,\ldots,$ $y_{\ell+2p-1}$ are $\ell+2p-1$ nonoutliers, and therefore $k_1=\ldots=k_{\ell+2p-1}=1$). The definition of $\mathcal{F}_i$ is the same as that of the set defined in (\ref{def_F_i}). For an interpretation of this set and of the sets involve in the decomposition of $\cup_i \mathcal{O}_i$, see the proof of (\ref{inequality1}). Given that $|y_i|\geq \omega$ for all $i\in\mathcal{O}_i$, we can use the same mathematical arguments as in the proof of (\ref{inequality1}) to show that $\mathcal{O}_i\cap \mathcal{F}_{i_1}\cap \cdots\cap\mathcal{F}_{i_{p}}=\varnothing$ for all $i,i_1,\ldots,i_{p}$ with $i_j\neq i_s$, $\forall i_j\neq i_s$ such that $j\neq s$. Therefore,
 \begin{align*}
 &\cup_i \mathcal{O}_i=\left[\cup_i \left( \mathcal{O}_i\cap\left(\cap_{i_1} \mathcal{F}_{i_1}^c\right)\right)\right]\cup \left[\cup_{i,i_1}\left(\mathcal{O}_i\cap \mathcal{F}_{i_1}\cap\left(\cap_{i_2\neq i_1}\mathcal{F}_{i_2}^c\right)\right)\right] \cr
 &\cup\cdots\cup \left[\cup_{i,i_1,\ldots,i_{p-1} (i_j\neq i_s \, \forall i_j,i_s \text{ s.t. } j\neq s)}\left(\mathcal{O}_i\cap \mathcal{F}_{i_1}\cap \cdots\cap\mathcal{F}_{i_{p-1}}\cap \left(\cap_{i_p\neq i_1,\ldots,i_{p-1}}\mathcal{F}_{i_p}^c\right)\right)\right].
\end{align*}

This decomposition of $\cup_i \mathcal{O}_i$ is comprised of $\sum_{i=0}^{p-1} {{\ell+2p-1} \choose {i}}$ mutually exclusive sets given by $\cup_i ( \mathcal{O}_i\cap(\cap_{i_1} \mathcal{F}_{i_1}^c))$, $\cup_{i}(\mathcal{O}_i\cap \mathcal{F}_{i_1}\cap(\cap_{i_2\neq i_1}\mathcal{F}_{i_2}^c))$ for $i_1\in\mathcal{I}_\mathcal{F} $, and so on. We now consider the area defined by: $0<\sigma< 1$ and  $\boldsymbol\beta$ belongs to one of these $\sum_{i=0}^{p-1} {{\ell+2p-1} \choose {i}}$ mutually exclusive sets. We have
\begin{align*}
&\pi(\boldsymbol\beta,\sigma\mid \mathbf{y_k})\prod_{i=1}^{n}\left[\frac{(1/\sigma)f((y_i-\mathbf{x}_i^T\boldsymbol\beta)/\sigma)}{f(y_i)}\right]^{\ell_i}\\
&\za{\leq}\pi(\boldsymbol\beta,\sigma\mid \mathbf{y_k})\prod_{i=1}^{n}\left[\frac{|b_i|D(|a_i|,|b_i|)(1/\sigma)f((y_i-\mathbf{x}_i^T\boldsymbol\beta)/\sigma)} {f(\omega)}\right]^{\ell_i}\\
&\propto \pi(\boldsymbol\beta,\sigma)\prod_{i=1}^{n}\left[(1/\sigma)f((a_i-\mathbf{x}_i^T\boldsymbol\beta)/\sigma )\right]^{k_i}\left[\frac{(1/\sigma)f((y_i-\mathbf{x}_i^T\boldsymbol\beta)/\sigma)}{f(\omega)}\right]^{\ell_i}\\
&\zb{\le}(B/\sigma) \left[2\gamma D(0,2\gamma)(1/\sigma) f(\omega/\sigma)\right]^{\ell+1}\prod_{i=1 (i\neq i_p,\ldots,i_{\ell+p})}^{n}\left[(1/\sigma)f((a_i-\mathbf{x}_i^T\boldsymbol\beta)/\sigma )\right]^{k_i} \cr
& \hspace{20mm}\times\left[\frac{(1/\sigma)f((y_i-\mathbf{x}_i^T\boldsymbol\beta)/\sigma)}{f(\omega)}\right]^{\ell_i}\\
&\propto (1/\sigma) \left[(1/\sigma)f(\omega/\sigma)\right]^{\ell+1}\prod_{i=1 (i\neq i_p,\ldots,i_{\ell+p})}^{n}\left[(1/\sigma)f((a_i-\mathbf{x}_i^T\boldsymbol\beta)/\sigma )\right]^{k_i} \cr
& \hspace{20mm}\times\left[\frac{(1/\sigma)f((y_i-\mathbf{x}_i^T\boldsymbol\beta)/\sigma)}{f(\omega)}\right]^{\ell_i}\\
&\zc{\le} (1/\sigma) (1/\sigma)f(\omega/\sigma)\prod_{i=1 (i\neq i_p,\ldots,i_{\ell+p})}^{n}\left[(1/\sigma)f((a_i-\mathbf{x}_i^T\boldsymbol\beta)/\sigma )\right]^{k_i} \cr
& \hspace{20mm}\times\left[(1/\sigma)f((y_i-\mathbf{x}_i^T\boldsymbol\beta)/\sigma)\right]^{\ell_i}\\
&\zd{\le} (B/\omega)(1/\sigma) \prod_{i=1 (i\neq i_p,\ldots,i_{\ell+p})}^n(1/\sigma)f((y_i-\mathbf{x}_i^T\boldsymbol\beta)/\sigma ).
\end{align*}
In Step $a$, we use Lemma~\ref{cor-location-scale-transformation} to obtain $f(\omega)/f(y_i)=f((y_i-a_i)/b_i)/f(y_i) \le |b_i|D(|a_i|,$ $|b_i|)$ for all $i\in\mathcal{I}_{\mathcal{O}}$. In Step $b$, we use $\pi(\boldsymbol\beta,\sigma)\le B\max(1,1/\sigma)=B/\sigma$. We also use that in any of the sets in which $\boldsymbol\beta$ can belong, there are at least $\ell+1$ nonoutlying points such that $|\mathbf{x}_i^T\boldsymbol\beta|\geq\omega/\gamma$ (corresponding to $\boldsymbol\beta\in\mathcal{F}_{i}^c$ for at least $\ell+1$ nonoutlying points). Indeed, the case in which there are the least nonoutliers such that $|\mathbf{x}_i^T\boldsymbol\beta|\geq\omega/\gamma$ corresponds to $\boldsymbol\beta\in\cup_i(\mathcal{O}_i\cap \mathcal{F}_{i_1}\cap \cdots\cap\mathcal{F}_{i_{p-1}}\cap (\cap_{i_p\neq i_1,\ldots,i_{p-1}}\mathcal{F}_{i_p}^c))$. In this case there are $p-1$ nonoutliers such that $|\mathbf{x}_i^T\boldsymbol\beta|<\omega/\gamma$ (say observations $i_1$ to $i_{p-1}$), which leaves at least $\ell+2p-1-(p-1)$ nonoutliers such that $|\mathbf{x}_i^T\boldsymbol\beta|\geq\omega/\gamma$ (i.e. that there are $\ell+2p-1-(p-1)$ sets in the intersection $\cap_{i_p\neq i_1,\ldots,i_{p-1}}\mathcal{F}_{i_p}^c)$, and we know that $\ell+2p-1-(p-1)=\ell+p>\ell+1$ because we only consider the models with $p\geq2$. This implies that there exists a set of $\ell+1$ indices, say $\{i_p,\ldots,i_{\ell+p}\}\subset\mathcal{I}_\mathcal{F}$, such that for all $i\in\{i_p,\ldots,i_{\ell+p}\}$,
\begin{equation*}
 f((a_i-\mathbf{x}_i^T\boldsymbol\beta)/\sigma ) \le f(\omega/(2\gamma\sigma) )\le 2\gamma D(0,2\gamma) f(\omega/\sigma),
\end{equation*}
 using the monotonicity of the tails of $f$ in the first inequality because, if we define the constant $a_{(k)}:=\max_{i\in\{1,\ldots,k\}}|a_i|$ with $\omega\ge \yo\ge (2\gamma)a_{(k)}$, we have
$|a_i-\mathbf{x}_i^T\boldsymbol\beta|/\sigma\geq (|\mathbf{x}_i^T\boldsymbol\beta|-\abs{a_i})/\sigma \ge (\omega/\gamma-a_{(k)})/\sigma \ge \omega/(2\gamma\sigma) \ge\omega/(2\gamma) \ge \yo/(2\gamma)\ge M$ if we choose $\yo\ge 2\gamma M$. In the second inequality, we use Lemma~\ref{cor-location-scale-transformation} (as mentioned in the proof of (\ref{inequality1}), we choose $\gamma\geq2$).
 In Step $c$ above, we use the monotonicity of the tails of $|z|f(z)$ to obtain $(\omega/\sigma) f(\omega/\sigma)\le \omega f(\omega)$ for $\ell$ terms, because  $\omega/\sigma\ge \omega \ge \yo\ge M$ if we choose $\yo\ge M$. In Step $d$, we use $(1/\sigma)f(\omega/\sigma)\leq B/\omega$.

 The integral of $(B/\omega)(1/\sigma) \prod_{i=1 (i\neq i_p,\ldots,i_{\ell+p})}^n(1/\sigma)f((y_i-\mathbf{x}_i^T\boldsymbol\beta)/\sigma )$ is bounded by
\begin{align*}
(B/\omega)\int_{\re^p}\int_0^\infty (1/\sigma) \prod_{i=1 (i\neq i_p,\ldots,i_{\ell+p})}^n(1/\sigma)f((y_i-\mathbf{x}_i^T\boldsymbol\beta)/\sigma )\,d\sigma\,d\boldsymbol\beta= (B/\omega) m(\mathbf{y}_{\mathcal{I}_R}),
\end{align*}
where $m(\mathbf{y}_{\mathcal{I}_R})$ is the marginal density arising from a prior proportional to $1/\sigma$ and $n-(\ell+1)=k-1$ observations $(\mathbf{x}_{i},y_{i})$, $i\in\mathcal{I}_R:=\{1,\ldots,n\}\setminus\{i_p,\ldots,i_{\ell+p}\}$. In order to prove that $(B/\omega) m(\mathbf{y}_{\mathcal{I}_R})\rightarrow 0$ as $\omega\rightarrow\infty$, it suffices to prove that $m(\mathbf{y}_{\mathcal{I}_R})$ is bounded by a constant that does not depend on $\omega$, because $1/\omega\rightarrow 0$. In Section~\ref{proof-proposition-proper-supp}, we proved that a marginal, as $m(\mathbf{y}_{\mathcal{I}_R})$, is bounded by a constant that does not depend on $\omega$ if the number of observations (which is $k-1$ in our case) is greater than or equal to $p+1$ if the prior divided by $1/\sigma$ is bounded (which is the case for $m(\mathbf{y}_{\mathcal{I}_R})$). Because we assume that $k\geq \ell+2p-1$ and $\ell\geq 1$ (the proof for the case $\ell=0$ is trivial), and because we only consider the models with $p\geq2$, $m(\mathbf{y}_{\mathcal{I}_R})$ is the marginal of $k-1\geq \ell+2p-2\geq p+1$ observations. As a result,
$$
 (B/\omega)\int_{\re^p}\int_0^\infty (1/\sigma) \prod_{i=1 (i\neq i_p,\ldots,i_{\ell+p})}^n(1/\sigma)f((y_i-\mathbf{x}_i^T\boldsymbol\beta)/\sigma )\,d\sigma\,d\boldsymbol\beta\rightarrow 0 \text{ as } \omega\rightarrow\infty.
$$
We therefore have that
$$
 \int_{\cup_i \mathcal{O}_i}\int_0^1\pi(\boldsymbol\beta,\sigma\mid \mathbf{y_k})\prod_{i=1}^{n}\left[\frac{(1/\sigma)f((y_i-\mathbf{x}_i^T\boldsymbol\beta)/\sigma)}{f(y_i)}\right]^{\ell_i}\,d\sigma\,d\boldsymbol\beta\rightarrow 0 \text{ as } \omega\rightarrow\infty.
$$
\end{proof}

\begin{proof}[Proof of Result~(b)]
Consider $(\boldsymbol\beta,\sigma)$ such that $\pi(\boldsymbol\beta,\sigma)>0$ (the proof for the case $(\boldsymbol\beta,\sigma)$ such that $\pi(\boldsymbol\beta,\sigma)=0$ is trivial).
We have, as $\omega\rightarrow \infty$,
  \begin{align*}
\frac{\pi(\boldsymbol\beta,\sigma\mid \mathbf{y_n})}{\pi(\boldsymbol\beta,\sigma\mid \mathbf{y_k})}
&= \frac{m(\mathbf{y_k})}{m(\mathbf{y_n})}\times
\frac{\pi(\boldsymbol\beta,\sigma)\prod_{i=1}^{n}(1/\sigma)f((y_i-\mathbf{x}_i^T \boldsymbol\beta)/\sigma)}
{\pi(\boldsymbol\beta,\sigma)\prod_{i=1}^{n}\left[(1/\sigma)f((y_i-\mathbf{x}_i^T \boldsymbol\beta)/\sigma)\right]^{k_i}}\\
&= \frac{m(\mathbf{y_k})}{m(\mathbf{y_n})}\prod_{i=1}^{n}\left[(1/\sigma)f((y_i-\mathbf{x}_i^T\boldsymbol\beta)/\sigma)\right]^{\ell_i}\\
&= \frac{m(\mathbf{y_k})\prod_{i=1}^{n}[f(y_i)]^{\ell_i}}{m(\mathbf{y_n})}\prod_{i=1}^{n}\left[\frac{(1/\sigma)f((y_i-\mathbf{x}_i^T\boldsymbol\beta)/\sigma)}
{f(y_i)}\right]^{\ell_i}\rightarrow 1.
   \end{align*}
The first ratio in the last equality does not depend on $\boldsymbol\beta$ and $\sigma$ and converges towards 1 as $\omega\rightarrow\infty$ using Result~(a). Also, the product converges towards 1 uniformly in any set $(\boldsymbol\beta,\sigma)\in [-\vartheta,\vartheta]^p\times[1/\eta, \eta]$ using Proposition~\ref{prop-location-scale-transformation} given that $\mathbf{x}_1,\ldots, \mathbf{x}_n$ are fixed.
Furthermore, since $f$ and $\pi(\boldsymbol\beta,\sigma)/\max(1,$ $1/\sigma)$ are bounded, $\pi(\boldsymbol\beta,\sigma\mid \mathbf{y_k})$ is also bounded on any set $(\boldsymbol\beta,\sigma)\in [-\eta,\eta]^p\times[1/\eta, \eta]$. Then, we have
   \begin{equation*}
   \big|\pi(\boldsymbol\beta,\sigma\mid \mathbf{y_n})-\pi(\boldsymbol\beta,\sigma\mid \mathbf{y_k})\big|=\pi(\boldsymbol\beta,\sigma\mid \mathbf{y_k})\abs{\frac{\pi(\boldsymbol\beta,\sigma\mid \mathbf{y_n})}{\pi(\boldsymbol\beta,\sigma\mid \mathbf{y_k})}-1}\rightarrow 0\text{ as }\omega\rightarrow\infty.
   \end{equation*}
\end{proof}

\begin{proof}[Proof of Result (c)]
Using Proposition~\ref{proposition-proper}, we know that $\pi(\boldsymbol\beta,$ $\sigma\mid\mathbf{y_k})$ and $\pi(\boldsymbol\beta,\sigma\mid\mathbf{y_n})$ are proper. Moreover, using Result~(b), we have the pointwise convergence $\pi(\boldsymbol\beta,\sigma\mid \mathbf{y_n})\rightarrow\pi(\boldsymbol\beta,\sigma\mid \mathbf{y_k})$  as $\omega\rightarrow\infty$ for any $\boldsymbol\beta\in\re^p$ and $\sigma>0$, as a result of the uniform convergence. Then, the conditions of Scheff\'{e}'s theorem (see \cite{scheffe1947useful}) are satisfied and we obtain the convergence in $L^1$ of $\pi(\boldsymbol\beta,\sigma\mid \mathbf{y_n})$ towards $\pi(\boldsymbol\beta,\sigma\mid \mathbf{y_k})$   as well as the following result:
\begin{equation*}
 \lim_{\omega\rightarrow\infty}\int_{E}\pi(\boldsymbol\beta,\sigma\mid\mathbf{y_n})\,d\boldsymbol\beta\,d\sigma=
  \int_{E}\pi(\boldsymbol\beta,\sigma\mid\mathbf{y_k})\,d\boldsymbol\beta\,d\sigma,
  \end{equation*}
uniformly for all sets $E\subset\re^p\times \re^{+}$. Result~(c) follows directly.
\end{proof}

\begin{proof}[Proof of Result (d)]
We prove that the moments converge through a mix of the strategies used to show Result (a) and that the moments exist in Proposition~\ref{proposition-proper}. For any $M$, a positive integer, we have
\begin{align*}
 \lim_{\omega\rightarrow\infty}\mathbb{E}[\sigma^M\mid \mathbf{y_n}] &= \lim_{\omega\rightarrow\infty} \int_{0}^{\infty}\int_{\re^p} \sigma^M \pi(\boldsymbol\beta,\sigma \mid \mathbf{y_n})\,d\boldsymbol\beta\,d\sigma \cr
&=  \int_{0}^{\infty}\int_{\re^p} \lim_{\omega\rightarrow\infty} \sigma^M \pi(\boldsymbol\beta,\sigma \mid \mathbf{y_n})\,d\boldsymbol\beta\,d\sigma \cr
&=  \int_{0}^{\infty}\int_{\re^p} \sigma^M \pi(\boldsymbol\beta,\sigma \mid \mathbf{y_k})\,d\boldsymbol\beta\,d\sigma=\mathbb{E}[\sigma^M\mid \mathbf{y_k}],
\end{align*}
assuming that we can interchange the limit and integral and using Result (b). To interchange the limit and integral, we again use Lebesgue's dominated convergence theorem which requires that the integrand is bounded by an integrable function of $\boldsymbol\beta$ and $\sigma$. We prove that it is the case using that
\begin{align*}
  \sigma^M\pi(\boldsymbol\beta,\sigma \mid\mathbf{y_n})&= \sigma^M \,\frac{\pi(\boldsymbol\beta,\sigma)\prod_{i=1}^n[(1/\sigma)f((y_i-\mathbf{x}_i^T\boldsymbol\beta)/\sigma)]^{k_i}}{m(\mathbf{y_n})}\cr
  &\qquad \times  \prod_{i=1}^n \left[\frac{(1/\sigma)f((y_i-\mathbf{x}_i^T\boldsymbol\beta)/\sigma)}{f(y_i)}\right]^{\ell_i} \prod_{i=1}^n \left[f(y_i)\right]^{\ell_i}.
\end{align*}
We have that $m(\mathbf{y_n})$ is bounded using Proposition~\ref{proposition-proper}, $\prod_{i=1}^n [f(y_i)]^{\ell_i}\leq B^\ell$, and
\begin{align*}
 \sigma^M \prod_{i=1}^n[(1/\sigma)f((y_i-\mathbf{x}_i^T\boldsymbol\beta)/\sigma)]^{k_i} \leq B^M \prod_{i=M+1}^n[(1/\sigma)f((y_i-\mathbf{x}_i^T\boldsymbol\beta)/\sigma)]^{k_i},
\end{align*}
using $f\leq B$ for the $M$ first observations and assuming without loss of generality that these observations are nonoutliers (therefore $k_1=\ldots=k_M=1$). Therefore, we need to show that
\begin{align*}
&\pi(\boldsymbol\beta,\sigma)\prod_{i=M+1}^n[(1/\sigma)f((y_i-\mathbf{x}_i^T\boldsymbol\beta)/\sigma)]^{k_i}  \prod_{i=1}^n \left[\frac{(1/\sigma)f((y_i-\mathbf{x}_i^T\boldsymbol\beta)/\sigma)}{f(y_i)}\right]^{\ell_i} \cr
&\qquad = m(\mathbf{y_k^*}) \, \pi(\boldsymbol\beta,\sigma\mid \mathbf{y_k^*})\prod_{i=1}^n \left[\frac{(1/\sigma)f((y_i-\mathbf{x}_i^T\boldsymbol\beta)/\sigma)}{f(y_i)}\right]^{\ell_i}
\end{align*}
is bounded by an integrable function of $\boldsymbol\beta$ and $\sigma$, where $\mathbf{y_k^*}:=\mathbf{y_k}\setminus\{y_1,\ldots,y_M\}$ (the nonoutlier group without the first $M$ nonoutliers). In the proof of Result (a), it has been shown that it is the case under the assumptions of Theorem~\ref{thm-main}, which are satisfied considering this modified data set with $k-M\geq n/2+(p-1/2)$ (see the additional assumption for Result (d) of Theorem~\ref{thm-main}).

For the expectations $\mathbb{E}[\beta_j^M \mid \mathbf{y_n}]$, we proceed in the same way, we simply additionally consider that, as in the proof of Proposition 2.1 (see Section~\ref{proof-proposition-proper-supp}), $\beta_j$ can be rewritten as $\mathbf{e}_j^T \boldsymbol\beta$, and that next, $\mathbf{e}_j$ can be expressed as a linear combination of $p$ vectors $\mathbf{x}_{i_1},\ldots,\mathbf{x}_{i_p}$, where now these are selected among the nonoutliers, i.e.\ $i_1,\ldots, i_p\in\{i: k_i=1\}$. We detail the case $M=1$. From it and what has been done before, it will be clear the result holds in general, with further technicalities. As above,
\begin{align*}
 \lim_{\omega\rightarrow\infty}\mathbb{E}[\beta_j\mid \mathbf{y_n}] &= \lim_{\omega\rightarrow\infty} \int_{0}^{\infty}\int_{\re^p} \beta_j \pi(\boldsymbol\beta,\sigma \mid \mathbf{y_n})\,d\boldsymbol\beta\,d\sigma \cr
&=  \int_{0}^{\infty}\int_{\re^p} \lim_{\omega\rightarrow\infty} \beta_j \pi(\boldsymbol\beta,\sigma \mid \mathbf{y_n})\,d\boldsymbol\beta\,d\sigma \cr
&=  \int_{0}^{\infty}\int_{\re^p} \beta_j \pi(\boldsymbol\beta,\sigma \mid \mathbf{y_k})\,d\boldsymbol\beta\,d\sigma=\mathbb{E}[\beta_j\mid \mathbf{y_k}],
\end{align*}
assuming that we can interchange the limit and integral and using Result (b). As above, we have to show that the integrand is bounded above by an integrable function of $\boldsymbol\beta$ and $\sigma$. We beforehand use that
\[
 \beta_j=\mathbf{e}_j^T \boldsymbol\beta=\sum_{s=1}^p a_s \mathbf{x}_{i_s}^T \boldsymbol\beta=\sum_{s=1}^p a_s (y_{i_s} - \mathbf{x}_{i_s}^T \boldsymbol\beta) - \sum_{s=1}^p a_s y_{i_s},
\]
as mentioned, where $a_1,\ldots,a_p\in\re$ and $i_1,\ldots, i_p\in\{i: k_i=1\}$. The integrand thus becomes a sum of $2p$ terms, and we prove that each one of them is bounded above by an integrable function of $\boldsymbol\beta$ and $\sigma$, which will complete the proof. As above
\begin{align*}
  a_s (y_{i_s} - \mathbf{x}_{i_s}^T \boldsymbol\beta)\pi(\boldsymbol\beta,\sigma \mid\mathbf{y_n})&\leq |a_s||y_{i_s} - \mathbf{x}_{i_s}^T \boldsymbol\beta| \,\frac{\pi(\boldsymbol\beta,\sigma)\prod_{i=1}^n[(1/\sigma)f((y_i-\mathbf{x}_i^T\boldsymbol\beta)/\sigma)]^{k_i}}{m(\mathbf{y_n})}\cr
  &\qquad \times  \prod_{i=1}^n \left[\frac{(1/\sigma)f((y_i-\mathbf{x}_i^T\boldsymbol\beta)/\sigma)}{f(y_i)}\right]^{\ell_i} \prod_{i=1}^n \left[f(y_i)\right]^{\ell_i}.
\end{align*}
We have that $m(\mathbf{y_n})$ is bounded using Proposition~\ref{proposition-proper}, $\prod_{i=1}^n [f(y_i)]^{\ell_i}\leq B^\ell$, and
\begin{align*}
 |y_{i_s} - \mathbf{x}_{i_s}^T \boldsymbol\beta| \prod_{i=1}^n[(1/\sigma)f((y_i-\mathbf{x}_i^T\boldsymbol\beta)/\sigma)]^{k_i} \leq B \prod_{i=1 (i\neq i_s)}^n[(1/\sigma)f((y_i-\mathbf{x}_i^T\boldsymbol\beta)/\sigma)]^{k_i},
\end{align*}
using
\[
 (|y_{i_s} - \mathbf{x}_{i_s}^T \boldsymbol\beta|/\sigma) f((y_{i_s}-\mathbf{x}_{i_s}^T\boldsymbol\beta)/\sigma)\leq B.
\]
Therefore, $ a_s (y_{i_s} - \mathbf{x}_{i_s}^T \boldsymbol\beta)\pi(\boldsymbol\beta,\sigma \mid\mathbf{y_n})$ is bounded above by a constant times
\begin{align*}
&\pi(\boldsymbol\beta,\sigma)\prod_{i=1 (i\neq i_s)}^n[(1/\sigma)f((y_i-\mathbf{x}_i^T\boldsymbol\beta)/\sigma)]^{k_i}  \prod_{i=1}^n \left[\frac{(1/\sigma)f((y_i-\mathbf{x}_i^T\boldsymbol\beta)/\sigma)}{f(y_i)}\right]^{\ell_i} \cr
&\qquad = m(\mathbf{y_k^*}) \, \pi(\boldsymbol\beta,\sigma\mid \mathbf{y_k^*})\prod_{i=1}^n \left[\frac{(1/\sigma)f((y_i-\mathbf{x}_i^T\boldsymbol\beta)/\sigma)}{f(y_i)}\right]^{\ell_i},
\end{align*}
where $\mathbf{y_k^*}:=\mathbf{y_k}\setminus\{y_{i_s}\}$ (the nonoutlier group without the $i_s$-th  nonoutlier). As mentioned above, in the proof of Result (a), it has been shown that it is the case under the assumptions of Theorem~\ref{thm-main}, which are satisfied considering this modified data set with $k-1\geq n/2+(p-1/2)$ (see the additional assumption for Result (d) of Theorem~\ref{thm-main}). The proofs for the terms with $a_s y_{i_s}\pi(\boldsymbol\beta,\sigma \mid\mathbf{y_n})$ is similar.
\end{proof}

\subsection{Complement of Section~\ref{sec_efficiency}}\label{sec_complement_3_2}

In Section~\ref{sec_efficiency}, we mention that the first derivative of the divergence
\begin{align}\label{eqn_KL_supp}
  \text{KL}(\boldsymbol\beta, \sigma):=\int \log(g(y_i) / p_{(\boldsymbol\beta,\sigma)}(y_i)) \, g(y_i) \, dy_i
 \end{align}
 with respect to $\boldsymbol\beta$ equals 0 at $\boldsymbol\beta_0$, and this for any value of $\sigma$. We also mention that while setting $\boldsymbol\beta=\boldsymbol\beta_0$ in (\ref{eqn_KL_supp}), it is minimised at $\sigma^*$ which depends on $\rho$. Finally, we mention that most of the regularity conditions in \cite{bunke1998asymptotic} are satisfied. We now show all this. We rewrite the divergence:
 \[
  \text{KL}(\boldsymbol\theta) = \mathbb{E}_g[\log g(Y)] - \mathbb{E}_g[\log p_{\boldsymbol\theta}(Y)],
 \]
 where $\mathbb{E}_g$ denotes the expectation with respect to $g$, and omitting the index $i$. The first term is computed exactly:
 \[
  \mathbb{E}_g[\log g(Y)]=-\frac{1}{2} \, \log(2\pi)-\log \sigma_0 - \frac{1}{2\sigma_0^2} \, \mathbb{E}_g[(Y - \mathbf{x}^T\boldsymbol\beta_0)^2] = -\frac{1}{2} \, (\log(2\pi) - 1) -\log \sigma_0.
 \]
 The second term is rewritten as:
 \begin{align*}
  \mathbb{E}_g[\log p_{\boldsymbol\theta}(Y)]&=\int \left(\log f\left(\frac{y - \mathbf{x}^T \boldsymbol\beta}{\sigma}\right) - \log \sigma \right) \frac{1}{\sqrt{2\pi}\sigma_0} \, \exp\left(-\frac{1}{2 \sigma_0^2} \, (y - \mathbf{x}^T \boldsymbol\beta_0)^2 \right) \, dy \cr
  &=\int \log f(z \eta + \delta \eta) \, \varphi(z) \, dz + \log \eta - \log \sigma_0,
 \end{align*}
 using the change of variable $z=(y - \mathbf{x}^T \boldsymbol\beta_0)/\sigma_0$, and denoting $\delta := \mathbf{x}^T (\boldsymbol\beta_0 - \boldsymbol\beta)/\sigma_0$ and $\eta := \sigma_0 / \sigma$. Therefore, minimising the divergence is equivalent to maximising
 \[
  \int \log f(z \eta + \delta \eta) \, \varphi(z) \, dz + \log \eta.
 \]
We now show that we can interchange the derivative with respect to $\delta$ and the integral. The first derivative of $f$ is given by
 \[
  f'(z)=\begin{cases}
  	-z \varphi(z) \quad \text{if} \quad |z| < \tau, \cr
  	-\varphi(\tau) \tau (\log \tau)^{\lambda + 1} \, \frac{\text{sign}(z)}{z^2} \, \frac{1}{(\log |z|)^{\lambda + 1}} \, \left(1 + \frac{\lambda + 1}{\log|z|} \right) \quad \text{if} \quad |z| > \tau, \cr
  	\text{does not exist if $z$ equals $-\tau$ or $\tau$,}
  \end{cases}
 \]
 where $\text{sign}(\cdot)$ is the sign function. For completeness, we assign the values $-z\varphi(z)$ to $f'(z)$ when $\abs{z}=\tau$. Note that we are allowed to do this because these points have null measure. We thus consider that
\[
 \frac{f'(z)}{f(z)}=\begin{cases}
  -z \quad \text{if} \quad |z| \leq \tau, \cr
  -\frac{\text{sign}(z)}{|z|} \, \left(1 + \frac{\lambda + 1}{\log|z|} \right) \quad \text{if} \quad |z| > \tau.
 \end{cases}
\]
This function is bounded. Consequently,
\[
 \frac{\partial}{\partial\delta} \, \log f(z \eta + \delta \eta)=\eta \, \frac{f'(z \eta + \delta \eta)}{f(z \eta + \delta \eta)}
\]
is bounded for any value of $\eta$, which implies that we can interchange the derivative and the integral. If $\delta = 0$, the integral is equal to 0, because $f'(-z)/f(z) = -f'(z)/f(z)$, and
 \begin{align*}
  \int \frac{f'(z \eta)}{f(z \eta)} \, \varphi(z) \, dz &= \int_{-\infty}^0 \frac{f'(z \eta)}{f(z \eta)} \, \varphi(z) \, dz + \int_{0}^\infty \frac{f'(z \eta)}{f(z \eta)} \, \varphi(z) \, dz \cr
  &= \int_0^\infty \frac{f'(-z \eta)}{f(-z \eta)} \, \varphi(-z) \, dz + \int_{0}^\infty \frac{f'(z \eta)}{f(z \eta)} \, \varphi(z) \, dz \cr
  &= -\int_0^\infty \frac{f'(z \eta)}{f(z \eta)} \, \varphi(z) \, dz + \int_{0}^\infty \frac{f'(z \eta)}{f(z \eta)} \, \varphi(z) \, dz.
 \end{align*}
 Notice that this is true for any value of $\eta$. Analysing the second derivative may allow to rigorously conclude that the divergence is (uniquely) minimised with respect to $\boldsymbol\beta$ at $\boldsymbol\beta_0$. If it is strictly negative for any value of $\eta$, it is the case. We now analyse
  \begin{align}\label{eqn_int_eta}
  \int \log f(z \eta) \, \varphi(z) \, dz + \log \eta.
 \end{align}
 In the same way as for $\delta$, we show that we can interchange the derivative with respect to $\eta$ and the integral. We have
 \[
 \frac{\partial}{\partial\eta} \, \log f(z \eta)=z \, \frac{f'(z \eta)}{f(z \eta)},
\]
which is bounded by $|z|$ times a constant. This is an integrable function with respect to $\varphi$. Therefore, we can interchange the integral and the derivative:
\[
 \frac{\partial}{\partial\eta}\left(\int \log f(z \eta) \, \varphi(z) \, dz + \log \eta\right)=\int z \, \frac{f'(z \eta)}{f(z \eta)} \, \varphi(z) \, dz + \frac{1}{\eta}.
\]
Setting the derivative equals to 0 leads to
\[
 \int z\eta \, \, \frac{f'(z \eta)}{f(z \eta)} \, \varphi(z) \, dz = -1.
\]
We cannot solve this explicitly, but numerical calculations show that the solution is unique. For instance, (\ref{eqn_int_eta}) as a function of $\eta$ with $\rho=0.95$ is shown in Figure~\ref{fig_eta} (a), with the maximiser $\eta^*$ as a function of $\rho$ in Figure~\ref{fig_eta} (b). The previous analysis suggests that $(\boldsymbol\beta^*,\sigma^*)=(\boldsymbol\beta_0, \sigma_0/\eta^*)$.

 \begin{figure}[ht]
  \centering
  $\begin{array}{cc}
  \includegraphics[width=0.35\textwidth]{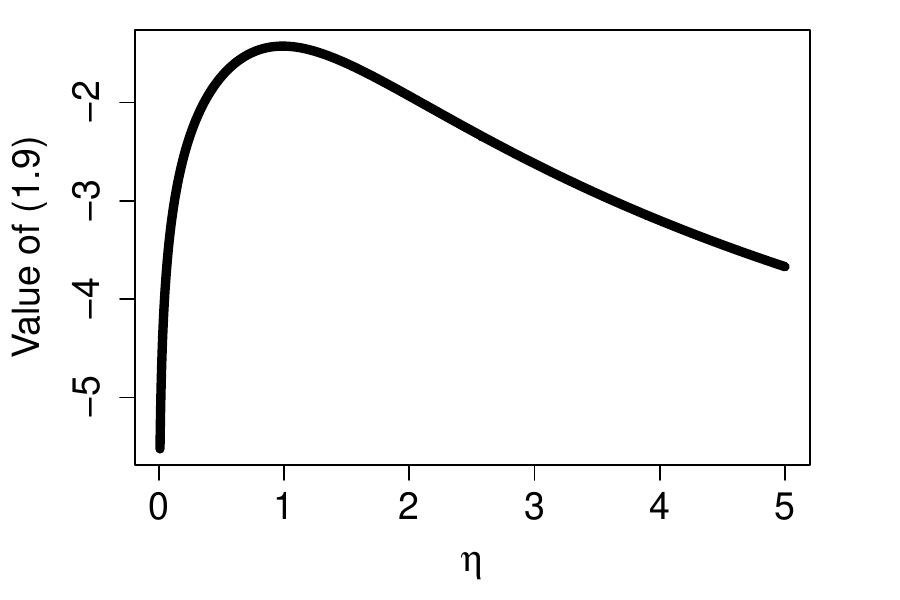} & \includegraphics[width=0.35\textwidth]{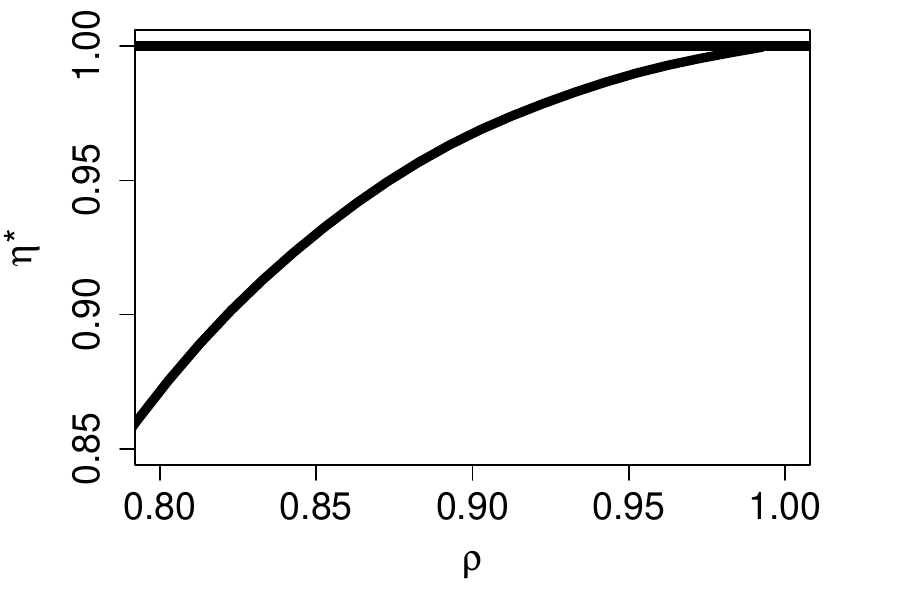}  \cr
  \textbf{(a)} & \textbf{(b)}
  \end{array}$\vspace{-1mm}
  \caption{\small (a) Value of (\ref{eqn_int_eta}) as a function of $\eta$ with $\rho=0.95$; (b) maximiser $\eta^*$ as a function of $\rho$}\label{fig_eta}
 \end{figure}
\normalsize

We now show that most of the regularity conditions in \cite{bunke1998asymptotic} are satisfied.

\begin{description}

	 \item[Condition 1.] The parameter space $\boldsymbol\Theta$ is a closed (possibly unbounded) convex set in $\re^d$ with a nonempty interior. The density $p_{\boldsymbol\theta}(y)$ is bounded for all $\boldsymbol\theta$ and $y$, and its carrier $\{y:p_{\boldsymbol\theta}(y)>0\}$ is the same for all $\boldsymbol\theta$.
	
	 \vspace{5mm} This condition is not directly satisfied because the parameter space is open ($\sigma > 0$). But it should not be a problem if we can show that it is possible to choose $\epsilon > 0$ such $\sigma_0\in [\epsilon,\infty)$ and that the mass outside of this set goes to 0 as the sample size increases. Indeed, we could ``define'' the parameter space to be $[\epsilon,\infty)\times \re^p$ which is a closed convex set and lose nothing asymptotically. On this parameter space $p_{\boldsymbol\theta}(y)$ is bounded for all $\boldsymbol\theta$ and $y$, and its carrier $\{y:p_{\boldsymbol\theta}(y)>0\}=\re$ is the same for all $\boldsymbol\theta$.
	
	 \item[Condition 2.] For all $\boldsymbol\theta$, there is a sphere $S[\boldsymbol\theta, r]$ of center $\boldsymbol\theta$ and radius $r$ with
	 \[
	  \mathbb{E}_g[\sup\{|\log[g(Y)/p_t(Y)]|:t\in S[\boldsymbol\theta, r]\}],
	 \]
	 where $\mathbb{E}_g$ is the expectation with respect to $g$.
	
	 \vspace{5mm} Lemma~\ref{cor-location-scale-transformation} implies that
	 \[
	  \frac{p_{(\boldsymbol\beta_0,\sigma_0)}(y)}{p_t(y)}
	 \]
	 is bounded both from below and above on $S[\boldsymbol\theta, r]$. Also, $g(y)/p_{(\boldsymbol\beta_0,\sigma_0)}(y)$ does not depend on $t$. Therefore, there exists a positive constant $c$ such that
	 \begin{align*}
	  \mathbb{E}_g[\sup\{|\log[g(Y)/p_t(Y)]|:t\in S[\boldsymbol\theta, r]\}]&\leq c \, \mathbb{E}_g[|\log[g(Y)/p_{(\boldsymbol\beta_0,\sigma_0)}(Y)]|] \cr
	  & = c\,\mathbb{E}_\varphi[|\log[\varphi(Z)/f(Z)]|],
	 \end{align*}
	 where a change of variables has been used in the last equality and $f$ is the LPTN density. We also have that
	 \[
	  \frac{\varphi(z)}{f(z)}=\begin{cases}
	   1 \quad \text{if} \quad |z|\leq \tau, \cr
	   c_2\exp(-z^2/2)|z|(\log |z|)^{\lambda+1} \quad \text{if} \quad |z| > \tau,
	  \end{cases}
	 \]
	 where $c_2$ is a positive constant.
	 Consequently,
	 \begin{align*}
	  c\,\mathbb{E}_\varphi[|\log[\varphi(Z)/f(Z)]|]&=c\,\mathbb{E}_\varphi[|\log[\varphi(Z)/f(Z)]|\ind(|Z|>\tau)] \cr
	  &\hspace{-25mm}\leq c\log(c_2) + c\,\mathbb{E}   _\varphi[(Z^2/2+|\log(|Z|(\log|Z|)^{\lambda +1})|)\ind(|Z|>\tau)] <\infty.
	 \end{align*}
	
	 \item[Condition 3.] For all fixed $y$, the density $p_{\boldsymbol\theta}(y)$ has a continuous derivative $p_{\boldsymbol\theta}'(y)$ with respect to $\boldsymbol\theta$ and there are positive constants $c, b_0$ such that
	 \[
	  \int \| [p_{\boldsymbol\theta}(y)]^{-1} p_{\boldsymbol\theta}'(y) \|^{4(d+1)} \, p_{\boldsymbol\theta}(y) \, dy < c\, (1+\|\boldsymbol\theta\|^{b_0}),
	 \]
	 for all $\boldsymbol\theta$, where $\|\cdot\|$ denotes a norm in $\re^d$.
	
	 \vspace{5mm} In our case, the density $p_{\boldsymbol\theta}(y)$ has an almost everywhere continuous derivative. We believe this should not cause fundamental problems for rigorously prove the result.
	
	 We have that
	 \[
	  \frac{\frac{\partial}{\partial\boldsymbol\beta}\left(\frac{1}{\sigma} f\left(\frac{y-\mathbf{x}^T \boldsymbol\beta}{\sigma}\right)\right)}{\frac{1}{\sigma} f\left(\frac{y-\mathbf{x}^T \boldsymbol\beta}{\sigma}\right)}=-\frac{\mathbf{x}^T}{\sigma} \frac{f'\left(\frac{y-\mathbf{x}^T \boldsymbol\beta}{\sigma}\right)}{f\left(\frac{y-\mathbf{x}^T \boldsymbol\beta}{\sigma}\right)}=-\frac{\mathbf{x}^T}{\sigma}\frac{f'(z)}{f(z)},
	 \]
	and
	\[
	 \frac{\frac{\partial}{\partial\sigma}\left(\frac{1}{\sigma} f\left(\frac{y-\mathbf{x}^T \boldsymbol\beta}{\sigma}\right)\right)}{\frac{1}{\sigma} f\left(\frac{y-\mathbf{x}^T \boldsymbol\beta}{\sigma}\right)}=-\frac{1}{\sigma} \frac{f\left(\frac{y-\mathbf{x}^T \boldsymbol\beta}{\sigma}\right)}{f\left(\frac{y-\mathbf{x}^T \boldsymbol\beta}{\sigma}\right)}-\frac{y-\mathbf{x}^T \boldsymbol\beta}{\sigma^2}\frac{f'\left(\frac{y-\mathbf{x}^T \boldsymbol\beta}{\sigma}\right)}{f\left(\frac{y-\mathbf{x}^T \boldsymbol\beta}{\sigma}\right)}=-\frac{1}{\sigma}\left(1+\frac{f'(z)}{f(z)}\right),
	\]	
	after the change of variable $z=(y-\mathbf{x}^T \boldsymbol\beta)/\sigma$. We use the traditional Euclidian norm. The function $f'/f$ is bounded. Therefore,
	\[
	  \int \| [p_{\boldsymbol\theta}(y)]^{-1} p_{\boldsymbol\theta}'(y) \|^{4(d+1)} \, p_{\boldsymbol\theta}(y) \, dy \leq c_0 \, \frac{1}{\sigma^{4(p+2)}},
	 \]
	 where $c_0$ is a positive constant. If the parameter space is $[\epsilon,\infty)\times \re^p$, it is easily seen that
	 \[
	  c_0 \, \frac{1}{\sigma^{4(p+2)}}  < c\, (1+\|\boldsymbol\theta\|^{b_0}).
	 \]
	
	 \item[Condition 4.] For some positive constant $b_1$ the affinity
	 \[
	  \varrho(\boldsymbol\theta):=\int [p_{\boldsymbol\theta}(y)g(y)]^{1/2} \, dy
	 \]
	 has the behaviour
	 \[
	  \varrho(\boldsymbol\theta) < c\|\boldsymbol\theta\|^{-b_1}, \quad \boldsymbol\theta\in\boldsymbol\Theta.
	 \]
	
	 \vspace{5mm} Condition 4 is, in our opinion, the condition that will require a careful analysis.

	 \item[Condition 5.] There are positive constants $b_2, b_3$ so that for all $\boldsymbol\theta$ and $r>0$ it holds that
	\[	
	 \pi(S[\boldsymbol\theta, r])\leq cr^{b_2}(1+(\|\boldsymbol\theta\| + r)^{b_3}),
	\]
	where $\pi(S[\boldsymbol\theta, r])$ measure of $S[\boldsymbol\theta, r]$ under the prior. Moreover, $\pi(S[\boldsymbol\theta, r])>0$ for all $r>0$ and $\boldsymbol\theta$.
	
	\vspace{5mm} The last part is satisfied if the prior is strictly positive over the parameter space, which is usually the case (it is true in our numerical analyses). The first part essentially requires that the measure does not ``explode'' in some areas. Under the assumption mentioned in Section~2.1 in our paper on the prior and if the parameter space is $[\epsilon,\infty)\times \re^p$, we have that
	\begin{align*}
	 \pi(S[\boldsymbol\theta, r])=\int \ind_{S[\boldsymbol\theta, r]}\, \pi(\boldsymbol\theta)\, d\boldsymbol\theta \leq \frac{1}{\epsilon}\int \ind_{S[\boldsymbol\theta, r]}\, d\boldsymbol\theta =   \frac{c}{\epsilon} \, r^{p+1},
	\end{align*}
	implying that the first part holds.
	
	 \item[Condition 6.] Let $L:\boldsymbol\Theta\times \boldsymbol\Theta \rightarrow \re^+$ be a measurable loss function with $L(\boldsymbol\theta, \boldsymbol\theta)= 0$, $c_1, c_2, c_3, b_4, b_5$ be positive constants such that
	 \[
	  (c_1\|t - \boldsymbol\theta\|^{b_4}) \wedge c_2 \leq L(t, \boldsymbol\theta)\leq c_3\|t-\boldsymbol\theta\|^{b_5},
	 \]	
	 for all $t,\boldsymbol\theta\in\boldsymbol\Theta$.
	
	 \vspace{5mm} It is easily seen that the quadratic loss function satisfies this, pointing towards the consistency of the posterior mean of $\boldsymbol\beta$. Under Conditions 1 to 5, a result in \cite{bunke1998asymptotic} indicates that the posterior density concentrates around $\boldsymbol\theta^*=(\boldsymbol\beta_0,\sigma^*)$, pointing in this case towards the consistency of the part of the posterior mode associated with $\boldsymbol\beta$.

\end{description}

\subsection{Other Result}\label{sec_other}

\begin{proposition}
 If $f=\mathcal{N}(0, 1)$ and $\pi(\boldsymbol\beta, \sigma)\propto \pi(\sigma) \times 1$, then
 \[\boldsymbol\beta\mid\sigma,\mathbf{y_n}\sim \mathcal{N}((\mathbf{X_n}^T \mathbf{X_n})^{-1} \mathbf{X_n}^T \mathbf{y_n}, \sigma^2 (\mathbf{X_n}^T \mathbf{X_n})^{-1}),
 \]
 and
 \[
  \pi(\sigma\mid \mathbf{y_n})\propto \pi(\sigma) \, \frac{1}{\sigma^{n-p}}\,\exp\left(-\frac{1}{2\sigma^2} \, \|\mathbf{y_n} - \hat{\mathbf{y}}_\mathbf{n}\|^2\right),
 \]
 where $\mathbf{X_n}$ is matrix whose rows are given by $\mathbf{x}_1^T, \ldots, \mathbf{x}_n^T$, $\hat{\mathbf{y}}_\mathbf{n}:= \mathbf{X_n} (\mathbf{X_n}^T \mathbf{X_n})^{-1}$ $ \mathbf{X_n}^T \mathbf{y_n}$, and $\|\cdot\|$ is the Euclidean norm. In particular, if $\pi(\sigma)\propto 1/\sigma$, $\sigma^2 \mid \mathbf{y_n}\sim \text{Inverse-}\Gamma((n-p)/2, \|\mathbf{y_n} - \hat{\mathbf{y}}_\mathbf{n}\|^2 / 2)$.
\end{proposition}

\begin{proof}
 The proof relies essentially on straightforward calculations. We have
 \begin{align*}
  \pi(\boldsymbol\beta, \sigma\mid \mathbf{y_n})\propto \pi(\sigma)\prod_{i=1}^n \frac{1}{\sqrt{2 \pi} \sigma}\exp\left(-\frac{1}{2\sigma^2} \, (y_i - \mathbf{x}_i^T \boldsymbol\beta)^2 \right).
 \end{align*}
 We therefore have sum of squares in the exponential and we first analyse it. We have
 \begin{align*}
  \sum_{i=1}^n(y_i - \mathbf{x}_i^T \boldsymbol\beta)^2 = \sum_{i=1}^n (y_i^2 - 2y_i \mathbf{x}_i^T \boldsymbol\beta + (\mathbf{x}_i^T \boldsymbol\beta)^2).
 \end{align*}
 We analyse the middle term:
 \begin{align}\label{eqn_middle}
  -2\sum_{i=1}^n y_i \mathbf{x}_i^T \boldsymbol\beta = -2 \boldsymbol\beta^T \sum_{i=1}^n \mathbf{x}_i y_i = -2 \boldsymbol\beta^T  \mathbf{X_n}^T \mathbf{y_n}= -2 (\mathbf{X_n}^T \mathbf{y_n})^T \boldsymbol\beta.
 \end{align}
 The last term is such that
 \begin{align*}
  \sum_{i=1}^n \mathbf{x}_i^T \boldsymbol\beta \, \mathbf{x}_i^T \boldsymbol\beta = \sum_{i=1}^n \boldsymbol\beta^T \mathbf{x}_i \mathbf{x}_i^T \boldsymbol\beta = \boldsymbol\beta^T \mathbf{X_n}^T \mathbf{X_n} \boldsymbol\beta.
 \end{align*}
 Adding and subtracting $(\mathbf{X_n}^T \mathbf{X_n})^{-1} \mathbf{X_n}^T \mathbf{y_n}$ to $\boldsymbol\beta$ before the first transpose yields
 \begin{align*}
  &(\boldsymbol\beta - (\mathbf{X_n}^T \mathbf{X_n})^{-1} \mathbf{X_n}^T \mathbf{y_n} + (\mathbf{X_n}^T \mathbf{X_n})^{-1} \mathbf{X_n}^T \mathbf{y_n})^T \mathbf{X_n}^T \mathbf{X_n} \boldsymbol\beta \cr
  &\qquad = (\boldsymbol\beta - (\mathbf{X_n}^T \mathbf{X_n})^{-1} \mathbf{X_n}^T \mathbf{y_n})^T \mathbf{X_n}^T \mathbf{X_n} \boldsymbol\beta + (\mathbf{X_n}^T \mathbf{y_n})^T (\mathbf{X_n}^T \mathbf{X_n})^{-1} \mathbf{X_n}^T \mathbf{X_n} \boldsymbol\beta.
 \end{align*}
 The last term on the RHS cancels out with one in (\ref{eqn_middle}). We again add and subtract $(\mathbf{X_n}^T \mathbf{X_n})^{-1} \mathbf{X_n}^T \mathbf{y_n}$ to $\boldsymbol\beta$:
 \begin{align*}
  & (\boldsymbol\beta - (\mathbf{X_n}^T \mathbf{X_n})^{-1} \mathbf{X_n}^T \mathbf{y_n})^T \mathbf{X_n}^T \mathbf{X_n} (\boldsymbol\beta -  (\mathbf{X_n}^T \mathbf{X_n})^{-1} \mathbf{X_n}^T \mathbf{y_n} +  (\mathbf{X_n}^T \mathbf{X_n})^{-1} \mathbf{X_n}^T \mathbf{y_n}) \cr
  & \qquad = (\boldsymbol\beta - (\mathbf{X_n}^T \mathbf{X_n})^{-1} \mathbf{X_n}^T \mathbf{y_n})^T \mathbf{X_n}^T \mathbf{X_n} (\boldsymbol\beta -  (\mathbf{X_n}^T \mathbf{X_n})^{-1} \mathbf{X_n}^T \mathbf{y_n}) \cr
  & \qquad\qquad + (\boldsymbol\beta - (\mathbf{X_n}^T \mathbf{X_n})^{-1} \mathbf{X_n}^T \mathbf{y_n})^T \mathbf{X_n}^T \mathbf{y_n}.
 \end{align*}
 The last term on the RHS is equal to $\boldsymbol\beta^T \mathbf{X_n}^T \mathbf{y_n}$, which cancels out with the remaining term in (\ref{eqn_middle}), minus $(\mathbf{X_n}^T \mathbf{y_n})^T (\mathbf{X_n}^T \mathbf{X_n})^{-1} \mathbf{X_n}^T \mathbf{y_n}$.

 Putting these results together yields
 \begin{align*}
  \pi(\boldsymbol\beta, \sigma\mid \mathbf{y_n})&\propto \pi(\sigma) \, \frac{1}{\sigma^n} \, \exp\left(-\frac{1}{2\sigma^2} \left(\sum_{i=1}^n y_i^2 - (\mathbf{X_n}^T \mathbf{y_n})^T (\mathbf{X_n}^T \mathbf{X_n})^{-1} \mathbf{X_n}^T \mathbf{y_n}\right)\right) \cr
  & \hspace{-10mm} \times \exp\left(-\frac{1}{2\sigma^2} (\boldsymbol\beta - (\mathbf{X_n}^T \mathbf{X_n})^{-1} \mathbf{X_n}^T \mathbf{y_n})^T \mathbf{X_n}^T \mathbf{X_n} (\boldsymbol\beta -  (\mathbf{X_n}^T \mathbf{X_n})^{-1} \mathbf{X_n}^T \mathbf{y_n})\right).
 \end{align*}
 Therefore,
 \[
 	\boldsymbol\beta\mid\sigma,\mathbf{y_n}\sim \mathcal{N}((\mathbf{X_n}^T \mathbf{X_n})^{-1} \mathbf{X_n}^T \mathbf{y_n}, \sigma^2 (\mathbf{X_n}^T \mathbf{X_n})^{-1}),
 \]
 and
 \begin{align*}
  \pi(\boldsymbol\beta\mid\sigma,\mathbf{y_n})&=\frac{1}{\sqrt{(2\pi)^p |(\mathbf{X_n}^T \mathbf{X_n}/\sigma)^{-1}|}} \cr
  & \hspace{-10mm} \times\exp\left(-\frac{1}{2\sigma^2} (\boldsymbol\beta - (\mathbf{X_n}^T \mathbf{X_n})^{-1} \mathbf{X_n}^T \mathbf{y_n})^T \mathbf{X_n}^T \mathbf{X_n} (\boldsymbol\beta -  (\mathbf{X_n}^T \mathbf{X_n})^{-1} \mathbf{X_n}^T \mathbf{y_n})\right) \cr
  &= \frac{1}{\sigma^p} \, \frac{1}{(2\pi)^{p/2}} \, |\mathbf{X_n}^T \mathbf{X_n}|^{1/2} \cr
   & \hspace{-10mm} \times\exp\left(-\frac{1}{2\sigma^2} (\boldsymbol\beta - (\mathbf{X_n}^T \mathbf{X_n})^{-1} \mathbf{X_n}^T \mathbf{y_n})^T \mathbf{X_n}^T \mathbf{X_n} (\boldsymbol\beta -  (\mathbf{X_n}^T \mathbf{X_n})^{-1} \mathbf{X_n}^T \mathbf{y_n})\right).
 \end{align*}
 Consequently,
 \[
  \pi(\sigma \mid \mathbf{y_n})\propto \pi(\sigma)\, \frac{1}{\sigma^{n-p}} \, \exp\left(-\frac{1}{2\sigma^2} \left(\sum_{i=1}^n y_i^2 - (\mathbf{X_n}^T \mathbf{y_n})^T (\mathbf{X_n}^T \mathbf{X_n})^{-1} \mathbf{X_n}^T \mathbf{y_n}\right)\right).
 \]
 It just remains to prove that
 \[
  \sum_{i=1}^n y_i^2 - (\mathbf{X_n}^T \mathbf{y_n})^T (\mathbf{X_n}^T \mathbf{X_n})^{-1} \mathbf{X_n}^T \mathbf{y_n} = \|\mathbf{y_n} - \hat{\mathbf{y}}_\mathbf{n}\|^2.
 \]
 We have
 \begin{align*}
  \sum_{i=1}^n y_i^2 - (\mathbf{X_n}^T \mathbf{y_n})^T (\mathbf{X_n}^T \mathbf{X_n})^{-1} \mathbf{X_n}^T \mathbf{y_n} &= \mathbf{y_n}^T \mathbf{y_n} - ( (\mathbf{X_n}^T \mathbf{X_n})^{-1} \mathbf{X_n}^T \mathbf{y_n})^T \mathbf{X_n}^T \mathbf{y_n} \cr
  &=\mathbf{y_n}^T \mathbf{y_n} - \hat{\mathbf{y}}_\mathbf{n}^T \mathbf{y_n} \cr
  &=(\mathbf{y_n} - \hat{\mathbf{y}}_\mathbf{n})^T (\mathbf{y_n} - \hat{\mathbf{y}}_\mathbf{n}^T + \hat{\mathbf{y}}_\mathbf{n}^T) \cr
 &= (\mathbf{y_n} - \hat{\mathbf{y}}_\mathbf{n})^T (\mathbf{y_n} - \hat{\mathbf{y}}_\mathbf{n}) + (\mathbf{y_n} - \hat{\mathbf{y}}_\mathbf{n})^T \hat{\mathbf{y}}_\mathbf{n} \cr
 &= (\mathbf{y_n} - \hat{\mathbf{y}}_\mathbf{n})^T (\mathbf{y_n} - \hat{\mathbf{y}}_\mathbf{n}).
 \end{align*}
\end{proof}

\end{document}